\documentclass{article}

\usepackage{arxiv}

\usepackage[utf8]{inputenc} 
\usepackage[T1]{fontenc}    
\usepackage{hyperref}       
\usepackage{url}            
\usepackage{booktabs}       
\usepackage{amsfonts}       
\usepackage{nicefrac}       
\usepackage{microtype}      
\usepackage{lipsum}
\usepackage{graphicx}
\usepackage{natbib,amsmath,amsthm,amssymb,bm,newtxmath}
\usepackage[ruled,linesnumbered]{algorithm2e}
\usepackage{subcaption}
\usepackage{multirow}
\graphicspath{ {./images/} }

\title{Near-Oracle Robustification of Finite-Difference Stochastic Gradient Estimators via Cheap Pilot Calibration}

\newtheorem{theorem}{Theorem}
\newtheorem{assumption}{Assumption}
\newtheorem{proposition}{Proposition}
\newtheorem{corollary}{Corollary}

\author{
 Haidong Li \\
  School of Economics and Management\\
  University of Chinese Academy of Sciences\\
  Beijing 100190, China \\
  \texttt{haidong.li@ucas.ac.cn} \\
   \And
 Henry Lam \\
  Department of Industrial Engineering and Operations Research\\
  Columbia University\\
  NY 10027, USA \\
  \texttt{henry.lam@columbia.edu} \\
  \And
 Yijie Peng \\
  School of Management and Engineering\\
  Nanjing University\\
  Nanjing 210093, China \\
  \texttt{pengyijie@nju.edu.cn} \\
}

\begin{document}
\maketitle
\begin{abstract}
We study stochastic gradient estimation in black-box environments where only noisy simulation observations of function values are available. Finite-difference (FD) methods are among the most widely used zeroth-order gradient estimators in such settings, by measuring the change in function values against a perturbation size. While the optimal order in choosing this perturbation size with respect to the simulation budget is well understood, the optimal constant factor relies on model characteristics that are typically unknown and viewed to be as difficult to estimate as the gradient itself. Consequently, FD estimators are often based on ad hoc tuning of the perturbation size, which may exhibit highly unstable performance across problem instances. In this paper, we challenge this conventional wisdom from both theoretical and practical perspectives. 
We show that, by pilot-estimating these model quantities using a negligible fraction of the simulation budget, substantial robustness is attained in the resulting FD estimators. Theoretically, we show that using a perturbation size governed by this pilot estimation can already achieve an MSE that is first-order identical to the ``oracle" MSE as if the optimal perturbation size is known in advance. Moreover, we show how such an approach is competitive against any choices of prescribed perturbation size, even if they are designed to be minimax-optimal over reasonable classes of target functions and FD schemes. Our proposed pilot estimation is practically easy to run, and a variety of numerical experiments demonstrate both the robustness and near-oracle optimality of our estimator relative to conventional FD schemes based on ad hoc tuning.
\end{abstract}


\section{Introduction}
We consider stochastic gradient estimation when only noisy simulation observations are available for evaluating function values or model outputs. Such estimation is commonly motivated from system sensitivity analysis (\citealp{l1991overview,fu2006gradient,asmussen2007stochastic,glasserman2013monte}), and often referred to as the zeroth-order or black-box oracle in the stochastic optimization literature (\citealp{ghadimi2013stochastic,nesterov2017random}). In particular, we investigate finite-difference (FD) methods, a widely employed class of zeroth-order stochastic gradient estimators due to their simplicity and broad applicability. FD estimators use first-principle numerical approximations of gradients and are fundamentally different from a range of unbiased gradient estimation methods. The latter include techniques such as infinitesimal perturbation analysis (\citealp{ho1983infinitesimal,heidelberger1988convergence}), the likelihood ratio or score function methods (\citealp{rubinstein1986score,reiman1989sensitivity,glynn1990likelihood}), measure-valued or weak differentiation approaches (\citealp{heidergott2008measure,heidergott2010gradient}), as well as other variants like the generalized likelihood ratio method (\citealp{peng2018new}). While unbiased gradient estimators are theoretically appealing, their applicability is often limited by restrictive regularity conditions and the requirement of access to internal model structure, which are typically unavailable in black-box simulation environments.

FD estimators consist of generating samples under perturbed input parameters, and the selection of the perturbation size is critical to their performance. More concretely, with a non-zero perturbation size, FD is typically biased against the target gradient. At the same time, a smaller perturbation size also results in a larger variance. To minimize the mean squared error (MSE) of FD estimators, the perturbation size must be chosen carefully to balance the bias–variance trade-off. Although the literature (\citealp{zazanis1993convergence,fox1989replication}) provides theoretical optimal perturbation sizes and their corresponding MSEs for FD schemes, it is widely acknowledged that these theoretical prescriptions are often impractical to implement precisely. More specifically, while the optimal order of the perturbation size, in terms of the sample size, is easy to obtain, attaining optimality in higher precision (i.e., taking into account the constant factor in front of the order) is tremendously difficult since it relies on unknown model characteristics that appear as difficult to estimate as the target gradient itself. In practice, a common fallback approach is to select a perturbation size that follows the optimal order while choosing the constant factor preceding the order in an ad hoc manner. However, because of the slow decay rate of the perturbation size, this constant is impactful, and such heuristically tuned perturbation sizes may perform well for some problem instances but very poorly for others, and thus fail to provide robust performance across heterogeneous problem settings.

As a simple illustration of the aforementioned phenomenon, we consider a generic M/M/1 queueing system, where the performance metric is the total system time of the first 10 customers and the quantity of interest is the derivative of its expectation with respect to the arrival rate. Figure~\ref{intro_example1} reports the empirical MSEs of FD estimators by simulating the queue, under three different perturbation sizes and a range of arrival rates. We see that none of the perturbation sizes performs uniformly well across the considered arrival rates. Moreover, a priori, we also do not know which choice of perturbation size is better than which others for a specific arrival rate. We note that this example is very generic (all our numerical examples in the sequel exhibit similar phenomena), and points to the inherent instability of heuristically tuned FD estimators.

\begin{figure}[h]
\centering
\begin{minipage}[t]{0.48\textwidth}
	\centering
	\includegraphics[width=7.5cm]{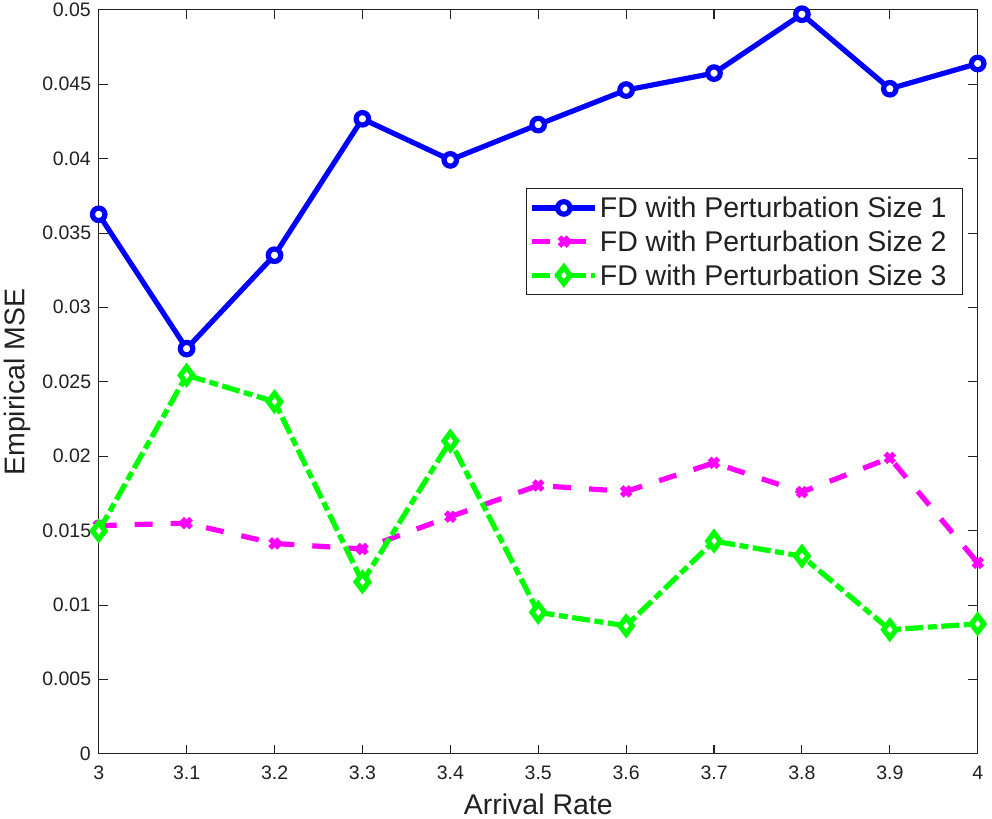}
	\caption{Empirical MSE of FD estimators using different perturbation sizes and arrival rates in an M/M/1 queueing system.}
	\label{intro_example1}
\end{minipage}
\begin{minipage}[t]{0.48\textwidth}
	\centering
	\includegraphics[width=7.5cm]{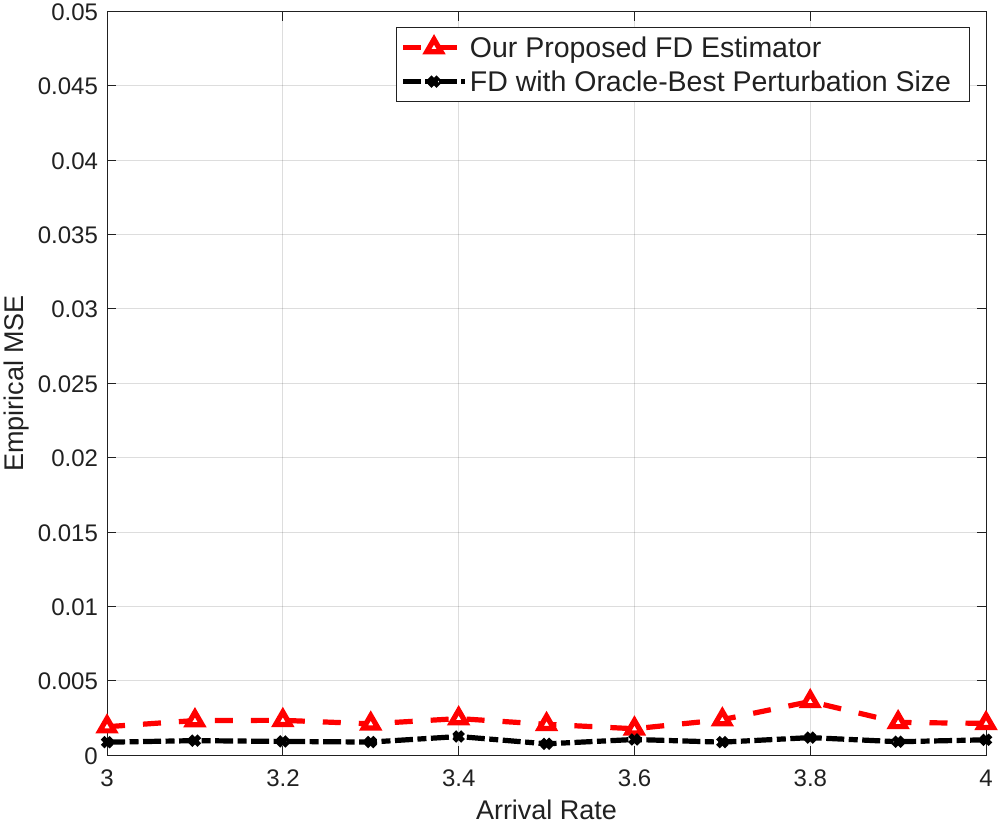}
	\caption{Empirical MSE of our proposed FD estimator and FD with oracle-best perturbation size for different arrival rates in an M/M/1 queueing system.}
	\label{intro_example2}
\end{minipage}
\end{figure}

Our main goal of this paper is to offer a simple, and arguably very natural, method to substantially enhance the robustness of FD estimators. In fact, this robustness comes from a stronger property of our approach: Our estimator attains an MSE that nearly matches the ``oracle" MSE, namely the theoretically best possible MSE obtained as if the unknown model quantities were known in advance. Before describing our idea, we first illustrate its performance with our example. Figure \ref{intro_example2} shows the MSE of our new estimator, using the \emph{same} simulation budget as all the FD estimators constructed in Figure \ref{intro_example1}. Note that our MSE is not only lower than those of all the FD estimators in Figure \ref{intro_example1} across all considered arrival rates, but it also always nearly matches the oracle MSE obtained from the theoretically best chosen perturbation size tailored to each specific arrival rate. In other words, our new estimator is both \emph{robust} and \emph{nearly oracle-optimal}.

Our key idea is essentially a rebuttal to the prevailing belief on the difficulty in obtaining optimal FD estimators. As alluded earlier, it is commonly believed that while order-optimal perturbation size is easy to achieve, optimizing the constant factor in front of the order is much harder since it requires knowledge on model quantities, such as higher-order derivatives and the noise variance, that appear more difficult to estimate than even the target gradient itself. In other words, estimating these quantities in order to improve the gradient estimation would undermine the latter task to begin with. We challenge this belief, with the key insight that \emph{a low accuracy of these model estimates suffices to produce highly robust and nearly oracle-optimal FD estimators}. In fact, the required accuracy is so low that only a negligible fraction of the overall simulation budget is required to run these estimates. These give rise to our proposed two-stage procedure: In the first stage, we allocate a small portion of the simulation budget to obtain coarse pilot estimates of the required model quantities. In the second stage, we plug these estimates to obtain a data-driven optimal perturbation size and then run a standard FD estimator using the remaining budget. Asymptotically as the sample size grows, the performance of our estimator matches the MSE, up to the exact first-order term, of an oracle that knows the optimal instance-specific perturbation size in advance. In this way, our estimator is nearly oracle-optimal across different problem instances.

The near oracle-optimality of our estimator implies robustness, in the sense that regardless of the problem instance, our estimator performs reasonably well (in fact, almost the best) compared to other FDs. To offer a broader perspective regarding our robustness strength, we further consider a worst-case MSE framework to quantify the performance of FD estimators. This worst case is cast over a natural class of functions whose gradients are our target. If an FD estimator is robust, then it would have a small worst-case MSE, and the most robust FD, as long as no pilot estimation for tuning the perturbation size is allowed, will be the corresponding minimax solution to this worst-case formulation. Via a ``weak duality", we show that by selecting the perturbation size through our pilot estimation, our performance is at least as competitive as the minimax solution under the worst-case MSE. Intuitively, this competitive performance stems from the data-adaptive nature of our approach, in contrast to the use of a fixed or ad-hoc-chosen choice. Moreover, our development of this minimax framework also reveals how our approach outperforms a broader class of FD schemes, not only deterministic perturbation size. More precisely, we show that even if the perturbation size is allowed to be randomized, or unequal weights can be placed over individual FD runs when forming the ultimate estimator, our approach is still superior. We establish this result by systematically characterizing the minimax-optimal estimators within a very broad family of FD schemes including randomized and unequal-weighted schemes. In particular, we show that equal-weight deterministic-perturbation FD estimators are in fact minimax optimal, and thus there is no need to consider beyond this FD type in terms of worst-case MSE. Thus, by outperforming this simple FD type, our new estimator also ``for free" outperforms randomized and unequal-weighted FD schemes.

For the majority of this paper, we present our key theoretical development and procedure in the single-dimensional central FD setting to simplify the exposition. The same ideas extend naturally to multi-dimensional FD estimators and to forward or backward FD. In particular, in the later part of our paper, we demonstrate how our procedure applies to the simultaneous perturbation (SP) scheme, which employs randomized perturbation directions and is widely used in high-dimensional problems. We also derive the exact first-order optimal MSE for the SP scheme, show how our estimator nearly matches so, and compare this MSE with that of applying coordinate-wise FD, thus providing mild sufficient conditions under which SP is preferable. These extensions are especially relevant for gradient checking (or gradient validation) in deep neural networks (DNNs), where FD remains a standard tool for verifying analytic gradients produced by hand-crafted operations or customized gradient estimators (\citealp{Goodfellow2016,karpathy2019neural}). Across domains such as software engineering (\citealp{ben2023testing}), aerospace (\citealp{kulkarni2016data}), and computational biology (\citealp{niu2015dsep}), gradient checking is an indispensable step in neural network training, as incorrect gradients can lead to misleading optimization behavior and wasted computational effort. FD-based gradient checking mitigates this risk by comparing analytic gradients with numerical FD approximations, which should agree locally on the same points. Our numerical results demonstrate the practical benefits of our proposed estimator in these contexts.

\subsection{Related Literature}
In general, stochastic gradient estimation methods can be broadly classified into two categories: direct and indirect approaches~(\citealp{fu2014stochastic}). Direct gradient estimators, including infinitesimal perturbation analysis (\citealp{ho1983infinitesimal,heidelberger1988convergence}), likelihood ratio or score function methods (\citealp{rubinstein1986score,reiman1989sensitivity,glynn1990likelihood}), measure-valued or weak differentiation approaches (\citealp{heidergott2008measure,heidergott2010gradient}), and the generalized likelihood ratio method (\citealp{peng2018new}), aim to unbiasedly estimate the gradient by exploiting additional structural information about the underlying stochastic model (e.g., input distributions or system dynamics). In contrast, indirect gradient estimation relies only on function evaluations (i.e., performance-measure samples) and produces a biased approximation to the gradient. Representative methods include standard FD (e.g.,~\citealp{asmussen2007stochastic},~\citealp{glasserman2013monte},~\citealp{fu2014stochastic}), simultaneous perturbation (\citealp{spall1992multivariate}), and randomized FD schemes (e.g.,~\citealp{ghadimi2013stochastic},~\citealp{duchi2015optimal},~\citealp{shamir2017optimal},~\citealp{fazel2018global},~\citealp{scheinberg2022finite}). Most indirect methods generate samples under perturbed input parameters and then use finite differences to approximate derivatives. Their performance depends critically on the choice of perturbation size, which governs the bias–variance trade-off and hence the overall MSE. While the optimal MSE rate (in terms of the simulation budget) is well understood and attainable (\citealp{zazanis1993convergence},~\citealp{fox1989replication}), achieving the exact first-order optimal MSE (i.e., optimizing the leading constant) remains substantially more challenging. The difficulty is that the theoretically optimal tuning depends on unknown model quantities such as higher-order derivatives and noise variance. Consequently, to the best of our knowledge, relatively few methods explicitly target MSE optimality at the level of the leading constant.

Next, we point out that a preliminary conference version of this work has appeared in \cite{li2020optimally}. Our current paper is a substantial enhancement of \cite{li2020optimally}, in that the latter only contains the basic idea of the tuning strategy (which they call estimation-minimization) and some theorems with incomplete proofs. In addition to presenting complete proofs and a more rigorous theoretical development of the pilot tuning methodology (with sharper asymptotic characterizations where applicable), the current paper develops the minimax robustness framework that elevates the discussion from instance-wise tuning to worst-case performance guarantees over a natural function class, and derives structural optimality results that clarify what constitutes a robust FD estimator under this criterion. The current paper also generalizes the methodology and analysis systematically to multidimensional settings, which is essential for practical high-dimensional applications. Finally, the current paper substantially expands the numerical study by considering a much broader range of test functions and dimensions and by adding a deep-neural-network gradient-checking case study, providing more comprehensive empirical evidence of robustness and near-oracle performance. Next, we also mention a related work (\citealp{liang2024correlation}) which, like us, studies sample-driven tuning for FD gradient estimation. Their approach uses a bootstrap technique, inspired by~\cite{zhang2022bootstrap}, to estimate the optimal perturbation size, and then applies a normalization–transformation step to reuse samples in the final gradient estimation. As cited and described in \cite{liang2024correlation}, their approach can be viewed as a data re-use enhancement of our conference version (\citealp{li2020optimally}). Moreover, our study differs in several methodological focuses. Rather than using the bootstrap, our main methodology is to use a pilot and, more importantly, our main insight is that this pilot budget is asymptotically negligible; consequently, strong performance can be achieved without adding significant computational burden. In terms of theory, we show the near-oracle optimality of our estimator. We further substantiate our robustness property by introducing a minimax framework that encompasses a very wide class of FD estimators including randomized and unequal-weighted schemes, and show that our estimator outperforms all these estimators in terms of the worst-case MSE. These optimality and robustness analyses are all beyond the scope of \cite{liang2024correlation}.

Another related work is \cite{lam2023enhanced}. They develop an enhanced class of FD estimators by combining simulation runs across a predetermined sequence of tuning parameter values, and show that these estimators can consistently outperform a conventional FD estimator with a fixed tuning choice, regardless of the unknown model characteristics. Their guarantees are established via an asymptotic minimax risk ratio, defined as the worst-case (over admissible values of the model unknowns) asymptotic ratio between the MSE of their estimator and that of a baseline FD estimator. While part of our work also adopts a worst-case perspective over model unknowns, our study is in a sense stronger than \cite{lam2023enhanced}, in that we measure performance directly through the MSE itself rather than a relative MSE ratio. Specifically, we invent a methodology to achieve the asymptotically optimal MSE, based on sample-driven tuning of FD, whereas \cite{lam2023enhanced} aims to improve the tuning parameters of general schemes through careful construction of predetermined tuning sequences.

Besides zeroth-order stochastic gradient estimation, many other simulation-based estimation methods involve tuning parameters that govern the bias–variance trade-off, making principled tuning a central practical challenge. For example, nested simulation is widely used to estimate quantities that are analytically intractable by embedding an inner simulation within an outer simulation. Its accuracy depends critically on how the computational budget is allocated across the two levels. Although asymptotic theory provides the optimal allocation rates, the associated leading constants typically depend on unknown features of the underlying model (e.g., the density function of a conditional expectation and its derivatives), so different tuning choices can yield substantially different MSEs in practice (\citealp{gordy2010nested}). To address this difficulty, \cite{zhang2022bootstrap} propose a sample-driven budget allocation rule that achieves asymptotic optimality. A similar tuning issue arises in kernel density estimation, a nonparametric method for estimating a probability density from samples, where performance hinges on selecting an appropriate bandwidth (\citealp{wand1994kernel,simonoff2012smoothing}). While the optimal bandwidth is characterized by minimizing integrated error criteria, it depends on unknown quantities (e.g., derivatives of the target density) and must be estimated from data. This has motivated a large literature on data-driven bandwidth selection, including classical plug-in methods (\citealp{park1990comparison,sheather1991reliable}) and more recent refinements (\citealp{tenreiro2020bandwidth}). These examples underscore a common theme: even when optimal tuning rates are known, selecting the instance-specific leading constants in a practically reliable way remains crucial for stable finite-sample performance. In this regard, our work contributes to this theme in the context of gradient estimation, with a key new insight that only a small pilot budget is required to attain the best possible performance by leveraging the simulation flexibility and finite-difference structure. It is possible that our insight can apply to other problems under this theme, which would constitute interesting future works.

The rest of the paper is organized as follows. In Section~\ref{sec:setting}, we introduce the motivation of developing an asymptotically optimal FD scheme under a minimax framework. Section~\ref{sec:single} proposes a two-stage scheme to optimize the performance of FD in the single-dimensional case, and discusses the asymptotic properties of parameter estimation and the overall MSE. Section~\ref{sec:multi} extends the proposed scheme to the multi-dimensional case and provides other relevant extensions. Section~\ref{sec:NR} presents numerical results. The proofs of the theorems and propositions in the paper can be found in the e-companion.

\section{Setting and Motivation}\label{sec:setting}
To motivate our approach, we begin with the single-dimensional setting of zeroth-order stochastic gradient estimation. Let $f(\cdot):\mathbb{R}\to\mathbb{R}$ denote the (unknown) performance function of interest. For any input $x\in\mathbb{R}$, we can obtain a simulation output $\hat{f}(x)$ such that $\mathbb{E}[\hat{f}(x)]=f(x)$ and $\text{Var}(\hat{f}(x))=\sigma^{2}(x)$. That is, $\hat{f}(x)$ is an unbiased noisy observation of $f(x)$ with (possibly input-dependent) variance $\sigma^{2}(x)$. Our goal is to estimate the first derivative $f^{(1)}(x_{0})$ at a given point
$x_{0}\in\mathbb{R}$. A standard approach is the central finite-difference (CFD) scheme, which for a perturbation size $\delta>0$ forms the single-run estimator
\begin{equation}
Y(\delta)=\frac{\hat{f}(x_{0}+\delta)-\hat{f}(x_{0}-\delta)}{2\delta}.\label{single-run}
\end{equation}
For $n$ independent replications  $\{Y_{i}(\delta)\}_{i=1}^{n}$, the CFD estimator of $f^{(1)}(x_{0})$ is the sample average 
\begin{equation}
\widehat{\theta}=\frac{1}{n}\sum_{i=1}^{n}Y_{i}(\delta).\label{standard FD}
\end{equation}

\subsection{Background on Finite-Difference and Error Behaviors}\label{sec:background}
We next quantify how the perturbation size $\delta$ affects the accuracy of the CFD estimator. To this end, we impose the following standard local regularity conditions.
\begin{assumption}\label{asp:1}
	The performance function $f(x)$ is five-times continuously differentiable in a neighborhood of $x_{0}$, with $f^{(3)}(x_{0})\neq 0$.
\end{assumption}
The differentiability condition enables the Taylor expansion used below, and the non-degeneracy condition $f^{(3)}(x_{0})\neq 0$ ensures that the leading bias term of the CFD estimator, which we will discuss in detail, does not vanish.
\begin{assumption}\label{asp:2}
	The output variance $\sigma^{2}(x)$ varies smoothly around $x_{0}$, in the sense that as $\delta\to0$, $\sigma^{2}(x_{0}+\delta)=\sigma^{2}(x_{0})+\Theta(|\delta|^{s})$ for some $s>0$. Specifically, $g(\delta)=\Theta(|\delta|^{s})$ means that there exist constants $c_{1},c_{2}>0$ and $\delta_{0}>0$ such that for all $0<|\delta|<\delta_{0}$, $c_{1}|\delta|^s\leq|g(\delta)|\leq c_{2}|\delta|^s$.
\end{assumption}
This condition formalizes the idea that a small perturbation in the input does not cause the noise variance to change abruptly, and it typically holds in simulation systems driven by continuous inputs.

Under Assumptions~\ref{asp:1} and~\ref{asp:2}, a Taylor expansion with the Lagrange-form remainder yields, for $\delta>0$,
\begin{equation}
Y(\delta)=f^{(1)}(x_{0})+\left[B\delta^{2}+R(\delta)\delta^{4}\right]+\frac{\epsilon(\delta)}{2\delta},\label{basic expansion}
\end{equation}
where $B=f^{(3)}(x_{0})/6$ and $R(\delta)$ is the Lagrange-remainder coefficient induced by the fifth derivative evaluated at intermediate points between $x_{0}$ and $x_{0}\pm\delta$; that is,
$$R(\delta)=\left[f^{(5)}(x_{0}+t(\delta)\delta)+f^{(5)}(x_{0}-t(-\delta)\delta)\right]/240,$$ with $t(\cdot)\in[0,1]$.
Moreover, $\epsilon(\delta)\in\mathbb{R}$ is a mean-zero noise term with $\text{Var}(\epsilon(\delta))=\eta^{2}(\delta)$. Under Assumption~\ref{asp:2}, as $\delta\to0$, we may write $\eta^{2}(\delta)=2k\sigma^{2}(x_{0})+\Theta(\delta^{s})$ , where $k=1$ when $\hat{f}(x_{0}+\delta)$ and $\hat{f}(x_{0}-\delta)$ are generated independently, and $0<k<1$ when common random numbers (CRNs) are used to induce positive correlation for variance reduction. Since $k$ enters subsequent derivations only as a multiplicative constant in the variance term and does not affect the form of the analysis, we focus on the independent-sampling case and set $k=1$ for the remainder of this section.

We consider the sample-average CFD estimator $\widehat{\theta}=n^{-1}\sum_{i=1}^{n}Y_{i}(\delta)$.
Under the preceding expansion \eqref{basic expansion}, $\widehat{\theta}$ admits a bias $\mathbb{E}[\widehat{\theta}]-f^{(1)}(x_{0})=B\delta^{2}+R(\delta)\delta^{4}$
while its variance is 
$\text{Var}(Y(\delta))/n=\eta^{2}(\delta)/(4n\delta^{2})$.
Consequently, the mean squared error (MSE) can be written as
\begin{eqnarray*}
	\text{MSE}=\mathbb{E}[(\widehat{\theta}-f^{(1)}(x_{0}))^{2}]=\left[B+R(\delta)\delta^{2}\right]^2\delta^{4}+\eta^{2}(\delta)/(4n\delta^{2}).
\end{eqnarray*}
Using that $R(\delta)$ is bounded locally and $\eta^{2}(\delta)=2\sigma^{2}(x_{0})+\Theta(\delta^{s})$ as $\delta\to0$, this yields
\begin{equation}
\text{MSE}=(B+o(1))^2\delta^{4}+(\sigma^{2}(x_{0})/2+o(1))/(n\delta^{2}),\label{basic MSE}
\end{equation}
where $o(1)$ denotes a term that vanishes as $\delta\to0$. 

We now consider choosing $\delta$ as a function of $n$ to minimize the MSE. The leading terms in the MSE \eqref{basic MSE} reveal a bias–variance trade-off: the squared bias scales as $B^{2}\delta^{4}$, while the variance scales as $(\sigma^{2}(x_{0})/2)/(n\delta^{2})$. To achieve the optimal MSE rate, these two leading terms should be of the same order; otherwise, one can adjust the order of $\delta$ to reduce the dominant term. Balancing $B^{2}\delta^{4}$ and $(\sigma^{2}(x_{0})/2)/(n\delta^{2})$ yields $\delta=\Theta(n^{-1/6})$, and the resulting optimal MSE rate is $n^{-2/3}$. We call a scheme \emph{rate-optimal} if it achieves this optimal order.

Thanks to the above calculation, which is well-known in the literature, implementing a rate-optimal CFD scheme is straightforward. However, to sharpen the accuracy further is nontrivial. To understand this, we next characterize the exact first-order term of the MSE. The leading terms of \eqref{basic MSE} satisfy the inequality
$$B^{2}\delta^{4}+(\sigma^{2}(x_{0})/2)/(n\delta^{2})\geq3\left(B\sigma^{2}(x_{0})/(4n)\right)^{2/3},$$
with equality if and only if 
\begin{equation}\delta=\left(\sigma^{2}(x_{0})/(4nB^{2})\right)^{1/6}.\label{optimal perturbation}
\end{equation}
This choice of $\delta$ depends on the priori unknown quantities $B$ and $\sigma^{2}(x_{0})$. Moreover, if we define the asymptotic scaled MSE
$$\mathcal{R}\triangleq\lim_{n\to+\infty}n^{2/3}\text{MSE},$$
(whenever the limit exists), then $\mathcal{R}_{\text{opt}}\triangleq3(B\sigma^{2}(x_{0})/4)^{2/3}$, which is the $\mathcal R$ attained by using a CFD with \eqref{optimal perturbation}, is the smallest achievable asymptotic scaled MSE. In other words, among CFD schemes, the first-order optimal MSE takes the form
\begin{equation}\text{MSE}_{\text{opt}}=3\left(B\sigma^{2}(x_{0})/4\right)^{2/3}n^{-2/3}(1+o(1)),\label{optimal MSE}
\end{equation}
with leading constant factor precisely $\mathcal{R}_{\text{opt}}$. We call a CFD scheme \emph{oracle-optimal} if it achieves \eqref{optimal MSE}, where ``oracle" refers to the observation that this is achievable if we know the model quantities $B$ and $\sigma^{2}(x_{0})$.

The distinction between \emph{rate-optimal} and \emph{oracle-optimal} highlights why tuning is difficult in practice. Setting $\delta=\alpha n^{-1/6}$ for some $\alpha>0$ suffices to achieve the optimal MSE rate $n^{-2/3}$, whereas attaining the exact first-order constant requires the specific choice $\alpha=\left(\sigma^{2}(x_{0})/(4B^{2})\right)^{1/6}$, which in turn depends on the unknown model quantities $B$ and $\sigma^{2}(x_{0})$. Because these quantities are typically unavailable, a conventional CFD estimator with an ad hoc $\alpha$ may perform well for some instances but poorly for others, as illustrated in Figure~\ref{intro_example1}. Our main goal in this paper is to develop an efficient, data-driven choice of $\alpha$ that attains oracle-optimality.

\subsection{Broader Class of Schemes and Notion of Robustness}
The preceding analysis shows that the first-order MSE of a CFD estimator depends on problem-specific quantities (e.g., $B$ and $\sigma^{2}(x_{0})$), so a tuning choice that is appropriate for one function may be far from optimal for another. Before we present our methodology to address this challenge, in this subsection we first introduce a notion of robustness that will not only be useful to discuss our methodology's benefits, but also reveal the broader class of schemes, beyond the CFD discussed in Section \ref{sec:background}, that our methodology is competitive against.

To begin, we consider the class $\Omega_{\hat{f}}$ of simulation models defined by 
$$\Omega_{\hat{f}}\triangleq\left\{\hat{f}(\cdot)\Big|\exists r>0\ \text{such that} \ 
\begin{aligned}
& f:(x_0-r,x_0+r)\to\mathbb{R} \\ 
& f(\cdot)\in C^{5}(x_0-r,x_0+r) 
\end{aligned},
\ |f^{(3)}(x_0)|\in[\underline{\kappa},\overline{\kappa}],
\ \text{and}\ \sigma^{2}(x_{0})\in(0,\overline{\sigma}^2]\right\}$$
for some given positive constants $\underline{\kappa}$, $\overline{\kappa}$ and $\overline{\sigma}^2$. The local smoothness condition $f(\cdot)\in C^{5}(x_0-r,x_0+r)$ guarantees the validity of the Taylor expansion around $x_0$ used to characterize the truncation bias of the CFD estimator (including a controlled Lagrange remainder). The lower bound $|f^{(3)}(x_0)|\geq \underline{\kappa}>0$ rules out degenerate instances in which the leading bias term vanishes, thereby changing the tuning structure, while the upper bound $|f^{(3)}(x_0)|\leq \overline{\kappa}$ prevents excessively large curvature near $x_{0}$. Finally, $\sigma^{2}(x_{0})\leq \overline{\sigma}^2$ bounds the local simulation noise level. Collectively, these conditions specify a broad yet well-behaved regime in which we can compare the MSE behavior of different FD schemes in a uniform manner; the minimax (worst-case) criterion introduced next will be taken over this class $\Omega_{\hat{f}}$.

Next, we consider a general family of CFD schemes. We allow the $i$-th replication of \eqref{single-run} to use a (possibly random) perturbation size of the form $\delta_{i}=\alpha_{i}n^{-1/6}$, $i=1,\ldots,n$, where $(\alpha_{1},\ldots,\alpha_{n})$ is generated from a joint distribution $\mathcal{P}$ satisfying $\alpha_{i}>0$ almost surely and mild moment conditions $\mathbb{E}[\alpha_{i}^{8}]<+\infty$ and $\mathbb{E}[1/\alpha_{i}^{2}]<+\infty$. To concretize, we define
$$\Omega_{\mathcal{P}}\triangleq\{\text{\ joint distribution on\ }(\alpha_{1},\ldots,\alpha_{n})\ |\  \alpha_{i}>0\text{\ a.s.\ }, \mathbb{E}[\alpha_{i}^{8}]<+\infty, \mathbb{E}[1/\alpha_{i}^{2}]<+\infty,\ \forall i\}.$$
The conditions above ensure that the scaled bias and variance contributions remain well-defined under the $n^{2/3}$ normalization. Moreover, we allow the CFD estimator to take the weighted form $$\widehat{\theta}_{G}=\sum_{i=1}^{n}w_{i}Y(\delta_{i}),$$ where $w_{i}>0,\ i=1,\ldots,n$ are possibly unequal. Here, $\mathcal{P}$ and $\{w_{i}\}_{i=1}^n$ constitute the estimator design and are fixed a priori, i.e., they cannot depend on the particular instance $\hat{f}(\cdot)\in\Omega_{\hat{f}}$; across different instances, the perturbations $\{\delta_{i}\}_{i=1}^n$ 
are generated independently according to the same $\mathcal{P}$.

We now formalize a robustness notion. We consider the worst-case scaled MSE
\begin{equation}
\underset{\hat{f}(\cdot)\in\Omega_{\hat f}}{\sup} n^{2/3}\mathbb{E}[(\widehat{\theta}_{G}-f^{(1)}(x_{0}))^{2}]\label{worst-case MSE}
\end{equation}
where $\mathbb{E}$ is an expectation with respect to both the simulation noise and the randomness in generating the perturbation sizes. If a scheme achieves a smaller value of \eqref{worst-case MSE}, then it is viewed as more robust since it performs better under the worst-case instance of $\hat f$. With this notion in place, let us consider the most robust scheme. Specifically, consider the asymptotic minimax scaled MSE:
\begin{eqnarray*}	\mathcal{R}_{G}^{*}&\triangleq&\lim_{n\to+\infty}\underset{\mathcal{P}\in\Omega_{\mathcal{P}},w_{i}>0}{\inf}\ \underset{\hat{f}(\cdot)\in\Omega_{\hat f}}{\sup} n^{2/3}\mathbb{E}[(\widehat{\theta}_{G}-f^{(1)}(x_{0}))^{2}]\\	&=&\lim_{n\to+\infty}\underset{\mathcal{P}\in\Omega_{\mathcal{P}},w_{i}>0}{\inf}\ \underset{\hat{f}(\cdot)\in\Omega_{\hat f}}{\sup} n^{2/3}\left(\mathbb{E}_{\mathcal{P}}\left[(\sum_{i=1}^{n}w_{i}-1)f^{(1)}(x_{0})+B\sum_{i=1}^{n}w_{i}\delta_{i}^{2}+\right.\right.\\	&&\left.\left.\sum_{i=1}^{n}w_{i}R(\delta_{i})\delta_{i}^{4}\right]^{2}+\mathbb{E}_{\mathcal{P}}\left[\sum_{i=1}^{n}(w_{i}/2\delta_{i})^{2}\eta^{2}(\delta_{i})\right]\right),
\end{eqnarray*}
where $\mathbb{E}_{\mathcal{P}}$ denotes the expectation with respect to the randomness in generating the perturbation sizes. The quantity $\mathcal{R}_{G}^{*}$ is the smallest possible worst-case MSE, scaled in order $n^{2/3}$ asymptotically as $n\to\infty$. That is, any $\hat\theta_G$ can exhibit at best a worst-case scaled MSE $\mathcal{R}_{G}^{*}$ asymptotically. Moreover, we call a design that attains $\mathcal{R}_{G}^{*}$ \emph{asymptotic minimax-optimal}. 

We will provide some characterization of asymptotic minimax-optimal designs, which would subsequently facilitate the illustration of the strengths of our proposed approach. First of all, intuitively, if the weights do not sum to one, the estimator will incur a non-vanishing bias term proportional to $f^{(1)}(x_{0})$, which cannot be controlled uniformly over the function class and therefore cannot be optimal under the worst-case criterion. The next proposition formalizes this idea and shows that, without loss of optimality, we may restrict attention to schemes with normalized weights $\sum_{i=1}^{n}w_{i}=1$.
\begin{proposition}\label{pro_general}
    We have
	\begin{eqnarray*}	    \mathcal{R}_{G}^{*}=\lim_{n\to+\infty}\underset{\mathcal{P}\in\Omega_{\mathcal{P}},\sum_{i=1}^{n}w_{i}=1,w_{i}>0}{\inf}\ \underset{\hat{f}(\cdot)\in\Omega_{\hat f}}{\sup}\left(\mathbb{E}_{\mathcal{P}}\left[B\sum_{i=1}^{n}w_{i}\alpha_{i}^{2}\right]^{2}+\mathbb{E}_{\mathcal{P}}\left[n\sum_{i=1}^{n}(w_{i}/2\alpha_{i})^{2}2\sigma^{2}(x_{0})\right]\right).
	\end{eqnarray*}
\end{proposition}

Next, we consider a deterministic subclass of $\Omega_{\mathcal{P}}$. Let $\Omega_{d}\subset\Omega_{\mathcal{P}}$ denote the set of degenerate distributions $\mathcal{P}$ under which each $\alpha_{i}$ is nonrandom (i.e., $\alpha_{i}$ takes a fixed value almost surely). Equivalently, under $\mathcal{P}\in\Omega_{d}$, the perturbation sizes $\delta_{i}=\alpha_{i}n^{-1/6}$ are deterministic design parameters rather than random variables. Since $\mathcal{P}\in\Omega_{d}$ is degenerate, optimizing $\mathcal{P}$ over $\Omega_{d}$ is equivalent to optimizing directly over $\{\delta_{i}\}_{i=1}^n$ (or $\{\alpha_{i}\}_{i=1}^n$) and $\{w_{i}\}_{i=1}^n$. To this end, we define the asymptotic minimax scaled MSE, over deterministic perturbations only, as
\begin{eqnarray*}
	\mathcal{R}_{\text{deterministic}}^{*}&\triangleq&\lim_{n\to+\infty}\underset{\mathcal{P}\in\Omega_{d},w_{i}>0}{\inf}\ \underset{\hat{f}(\cdot)\in\Omega_{\hat f}}{\sup} n^{2/3}\mathbb{E}[(\widehat{\theta}_{G}-f^{(1)}(x_{0}))^{2}]\\
	&=&\lim_{n\to+\infty}\underset{\delta_{i},w_{i}>0}{\inf}\ \underset{\hat{f}(\cdot)\in\Omega_{\hat f}}{\sup} n^{2/3}\left(\left[(\sum_{i=1}^{n}w_{i}-1)f^{(1)}(x_{0})+B\sum_{i=1}^{n}w_{i}\delta_{i}^{2}+\right.\right.\\
	&&\left.\left.\sum_{i=1}^{n}w_{i}R(\delta_{i})\delta_{i}^{4}\right]^{2}+\sum_{i=1}^{n}(w_{i}/2\delta_{i})^{2}\eta^{2}(\delta_{i})\right).
\end{eqnarray*}
The next corollary is in parallel to Proposition~\ref{pro_general}, restricting the admissible set of distributions from $\Omega_{\mathcal{P}}$ to $\Omega_{d}$.
\begin{corollary}
    We have
	\begin{eqnarray*}		\mathcal{R}_{\text{deterministic}}^{*}=\lim_{n\to+\infty}\underset{\alpha_{i}>0,\sum_{i=1}^{n}w_{i}=1,w_{i}>0}{\inf}\ \underset{\hat{f}(\cdot)\in\Omega_{\hat f}}{\sup}\left[B\sum_{i=1}^{n}w_{i}\alpha_{i}^{2}\right]^{2}+n\sum_{i=1}^{n}(w_{i}/2\alpha_{i})^{2}2\sigma^{2}(x_{0}).
	\end{eqnarray*}
\end{corollary}

We now establish the main conclusion of this section:
\begin{theorem}\label{thm:1}
    We have
	\begin{eqnarray}		\mathcal{R}_{G}^{*}&=&\mathcal{R}_{\text{deterministic}}^{*}\label{thm1.1}\\
		&=&\underset{\alpha>0}{\inf}\ \underset{\hat{f}(\cdot)\in\Omega_{\hat f}}{\sup}\ B^{2}\alpha^{4}+\sigma^{2}(x_{0})/(2\alpha^{2})\label{thm1.2}\\
		&\geq&\underset{\hat{f}(\cdot)\in\Omega_{\hat f}}{\sup}\ \underset{\alpha>0}{\inf}\ B^{2}\alpha^{4}+\sigma^{2}(x_{0})/(2\alpha^{2}).\label{thm1.3}
	\end{eqnarray}
\end{theorem}

We explain Theorem~\ref{thm:1}. First, \eqref{thm1.1} implies that, in terms of the asymptotic minimax scaled MSE, there is no distinction between using the general class of schemes $\hat\theta_G$ and deterministic perturbations. In other words, to achieve asymptotic minimax scaled MSE, it suffices to focus on deterministic perturbations, which are discussed in Section \ref{sec:background}. Thanks to this result, it is easy to deduce the form of asymptotic minimax scaled MSE as \eqref{thm1.2}. Finally, \eqref{thm1.3} arises from a weak duality that the interchange of inf and sup results in a lower value.

Now we connect Theorem \ref{thm:1} to our proposed estimator discussed in detail in the next section. Our estimator operates by lightly estimating $B$ and $\sigma^2(x_0)$, consequently approximately optimizing $\alpha$ in the perturbation size at the instance-specific level. In terms of asymptotic scaled MSE, our estimator achieves the outer objective function $\underset{\alpha>0}{\inf}\ B^{2}\alpha^{4}+\sigma^{2}(x_{0})/(2\alpha^{2})$ in \eqref{thm1.3}, for any $\hat f\in\Omega_{\hat f}$. This leads to an oracle-optimal performance in the sense discussed in Section \ref{sec:background}. Consequently, this also leads to a worst-case asymptotic scaled MSE that is at least as good as an asymptotic minimax-optimal estimator that can attain $\mathcal R_G^*$. Moreover, \eqref{thm1.3} implies that this performance is at least as good as any possible CFD schemes $\hat\theta_G$ in the general class, including randomized and unequal-weighted schemes, as long as the scheme needs to select $\mathcal P$ and $w_i$'s in advance. That is, our proposed estimator will be at least as robust as any of these general schemes.

\section{Pilot-Calibrated Finite-Difference}\label{sec:single}
We now present our methodology. We propose a two-stage scheme, which we call Pilot-Calibrated Central Finite-Difference (PC-CFD), to obtain a robust and nearly oracle-optimal estimator that we hinted in the previous sections. In the first stage of our scheme, we estimate the model quantities, namely $B$ and $\sigma^{2}(x_{0})$, needed to calibrate the oracle-optimal perturbation size described in \eqref{optimal perturbation}. In the second stage, we plug these estimates into \eqref{optimal perturbation} and run a standard CFD, namely \eqref{standard FD}, using this calibrated perturbation size. This will give our estimate of $f^{(1)}(x_{0})$. 

To highlight the purpose of each stage, we call the first stage the \emph{pilot estimation} stage and the second stage the \emph{calibrated deployment} stage. The samples used in the pilot estimation and the calibrated deployment stages are independent. Throughout this section, one unit of simulation budget corresponds to one independent single-run CFD replication (i.e., one realization of~\eqref{single-run}, which requires two function evaluations at $x_0\pm\delta$). Thus, given a total budget of $n$ replications, we allocate $n_1$ replications to the first stage and $n_2=n-n_1$ replications to the second stage. Moreover, and importantly as a main insight, the first-stage allocation $n_1$ can be an asymptotically negligible fraction of $n$, which is the key to making our PC-CFD statistically efficient.

The basic skeleton of PC-CFD is outlined in Algorithm~\ref{alg:pccfd_main}. 

\begin{algorithm}[H]
	\caption{Pilot-Calibrated Central Finite-Difference (PC-CFD)}\label{alg:pccfd_main}
    \textbf{Input}: target point $x_0$, total budget of $n$ single-run CFD replications, pilot budget $n_1$.\\
	\textbf{Pilot Estimation Stage}: Use $n_1$ replications to obtain $(\widehat{B},\widehat{\sigma}^2(x_0))$ via Algorithm~\ref{alg:pccfd_pilot}.\\
	\textbf{Calibrated Deployment Stage}: Set $n_{2}=n-n_{1}$ and $\widehat{\delta}=\left(\widehat{\sigma}^{2}(x_{0})/(4n_{2}\widehat{B}^2)\right)^{1/6}$. Use $n_{2}$ new replications with perturbation size $\widehat{\delta}$ to compute $\widehat{\theta}_{\text{PC}}$ via Algorithm~\ref{alg:pccfd_deploy}.\\
	\Return $\widehat{\theta}_{\text{PC}}$ as an estimate of $f^{(1)}(x_0)$.
\end{algorithm}

\subsection{Pilot Estimation of Needed Model Quantities}\label{sec:estimation}
In the \emph{pilot estimation} stage, our goal is to obtain coarse yet reliable enough estimates of the two model quantities
$$B = \frac{f^{(3)}(x_0)}{6},\qquad \sigma^2(x_0)= \text{Var}(\hat f(x_0)),$$
which are needed to calibrate the oracle-optimal perturbation size in~\eqref{optimal perturbation}. 

We use two distinct perturbation sizes \(\delta_1\neq \delta_2>0\) and split the pilot budget \(n_1\) evenly across them. For simplicity, assume that $n_1$ is even and let $m=n_1/2$. For each $i\in\{1,2\}$ and each replication \(k=1,\ldots,m\), we generate two independent function evaluations at $x_0+\delta_i$ and $x_0-\delta_i$ (also independent across $k$ and across $i$), and form
\[
y_i^{(k)} \triangleq \frac{(\hat f(x_0+\delta_i))^{(k)}-(\hat f(x_0-\delta_i))^{(k)}}{2}, \qquad i\in\{1,2\},\ k\in\{1,\ldots,m\},
\]
where the superscript $k$ indexes independent replications and the subscript $i$ indicates which perturbation size is used. Each $y_i^{(k)}$ uses two function evaluations; equivalently $y_i^{(k)}/\delta_{i}$ is a realization of~\eqref{single-run}, while we store $y_i^{(k)}$ for pilot estimation. Using the paired outputs $\{(y_1^{(k)},y_2^{(k)})\}_{k=1}^{m}$, we estimate $B$ by
\begin{gather}
	\widehat{B}=\frac{1}{m}\sum_{k=1}^{m}\frac{\delta_{2}y_{1}^{(k)}-\delta_{1}y_{2}^{(k)}}{\delta_{2}\delta_{1}^{3}-\delta_{1}\delta_{2}^{3}},\label{eq.FD}
\end{gather}
and estimate $\sigma^2(x_0)$ by two times the sample variance of $\{y_1^{(k)}\}_{k=1}^{m}$, i.e.,
\begin{gather}
	\widehat{\sigma}^{2}(x_{0})=\frac{2\|\bm{y}_{1}-\bar{y}_{1}\vmathbb{1}_{m\times 1}\|_{2}^{2}}{m-1}\label{sigma_hat_FD_single}
\end{gather}
where $\bm{y}_{1}=(y_{1}^{(1)},\ldots,y_{1}^{(m)})'$, $\bar{y}_{1}=m^{-1}\sum_{k=1}^{m}y_{1}^{(k)}$, and $\vmathbb{1}_{m\times 1}$ is a column vector with all elements equal to one.
Algorithm~\ref{alg:pccfd_pilot} summarizes the above pilot estimation procedure.

\begin{algorithm}[!h]
\caption{Pilot Estimation Stage}\label{alg:pccfd_pilot}
\textbf{Input}: target point $x_0$, pilot budget $n_1$ (assumed even), pilot perturbations $\delta_1\neq\delta_2>0$.\\
Set $m\leftarrow n_1/2$.\\
\For{$k=1,\ldots,m$}{
    \For{$i\in\{1,2\}$}{
      Generate two independent outputs $(\hat f(x_0+\delta_i))^{(k)}$ and $(\hat f(x_0-\delta_i))^{(k)}$.\\
      Compute $y_i^{(k)}\leftarrow\big((\hat f(x_0+\delta_i))^{(k)}-(\hat f(x_0-\delta_i))^{(k)}\big)/2$.
    }
}
Compute
\[
\widehat B \leftarrow \frac{1}{m}\sum_{k=1}^{m}\frac{\delta_2\,y_1^{(k)}-\delta_1\,y_2^{(k)}}{\delta_2\delta_1^{3}-\delta_1\delta_2^{3}}.
\]\\
Compute
\[
\bar y_1\leftarrow \frac{1}{m}\sum_{k=1}^m y_1^{(k)}, \qquad
S_1^2 \leftarrow \frac{1}{m-1}\sum_{k=1}^m \bigl(y_1^{(k)}-\bar y_1\bigr)^2,
\]
and
\[
\widehat\sigma^2(x_0)\leftarrow 2S_1^2.
\]
\Return $(\widehat B,\widehat\sigma^2(x_0))$ as an estimate of $(B,\sigma^2(x_0))$.
\end{algorithm}

Intuitively, the pilot estimator of $B$ is based on a local Taylor expansion of the symmetric difference around $x_0$. More precisely,
\[
y_i^{(k)} 
= f^{(1)}(x_0)\delta_i + B\delta_i^3 + \left(R(\delta_i)\delta_i^5 + \frac{\epsilon^{(k)}(\delta_i)}{2}\right), \qquad i\in\{1,2\},\ k\in\{1,\ldots,m\},
\]
where $\epsilon^{(k)}(\delta_i)$ denotes the noise term from the $k$-th independent replication. Therefore, each pilot observation $y_i^{(k)}$ can be viewed as a noisy realization of a two-regressor model whose mean is approximately linear in $\delta_i$ and $\delta_i^3$, with coefficients $f^{(1)}(x_0)$ and $B$, respectively. Using two distinct perturbations \(\delta_1\neq \delta_2\) then allows us to eliminate the coefficient $f^{(1)}(x_0)$ and isolate $B$ up to higher-order truncation effects and simulation noise. The estimator in \eqref{eq.FD} comes exactly from the corresponding two-regressor, two-point regression that recovers the coefficient of \(\delta^3\), by viewing the response at each point as an average over the $m$ independent replications that reduces its variance at the standard $1/m$ rate. For $\sigma^2(x_0)$, note that $y_1^{(k)}$ is formed from two independent function evaluations at $x_0\pm\delta_1$. When $\delta_1$ is small, the output variances at $x_0\pm\delta_1$ satisfy $\text{Var}(\hat f(x_0\pm\delta_1))=\sigma^2(x_0)+o(1)$ under Assumption~\ref{asp:2}. This motivates estimating $\sigma^2(x_0)$ by 2 times the sample variance of $\{y_1^{(k)}\}_{k=1}^{m}$, as in \eqref{sigma_hat_FD_single}.

We next quantify the estimation error of $\widehat{B}$ in~\eqref{eq.FD}. Throughout, the analysis is conducted under Assumptions~\ref{asp:1} and \ref{asp:2} stated in Section~2.1 and the independent-sampling convention adopted there (i.e., $k=1$). The next proposition gives the leading bias and variance of $\widehat{B}$, and the subsequent corollary derives an MSE expansion of $\widehat{B}$ under the scaling $\delta_1,\delta_2=\Theta(n_1^\beta)$ for $\beta<0$ and identifies regimes of $\beta$ for which the MSE vanishes, implying the consistency of $\widehat{B}$.

\begin{proposition}\label{pro_B}
Under Assumptions~\ref{asp:1} and~\ref{asp:2}, consider $\widehat B$ in~\eqref{eq.FD} with $\delta_1\neq \delta_2$.

\noindent\textup{(i) Bias behavior:} As $\delta_1,\delta_2\to 0$, there exists a quantity $D(\delta_1,\delta_2)$ (induced by the Lagrange remainder term) such that
\begin{equation}\label{eq:EBhat}
\mathbb{E}[\widehat B]
= B + D(\delta_1,\delta_2)\,(\delta_1^2+\delta_2^2).
\end{equation}
If, in addition, the pilot perturbations satisfy $\delta_2=c\delta_1$ for some fixed $c\in(0,1)$, then $D(\delta_1,\delta_2)$ is uniformly bounded for sufficiently small $\delta_1$ and satisfies
\begin{equation}\label{eq:D_asymp}
D(\delta_1,\delta_2)=\frac{f^{(5)}(x_0)}{120}+o(1),\qquad \text{as }\delta_1\to 0.
\end{equation}

\noindent\textup{(ii) Variance behavior:} As $\delta_1,\delta_2\to 0$,
\begin{equation}\label{eq:VarBhat}
\text{Var}(\widehat{B})=\frac{(\delta_{1}^{2}+\delta_{2}^{2})\sigma^{2}(x_{0})/2+\Theta(\delta_{1}^{s}\delta_{2}^{2})+\Theta(\delta_{2}^{s}\delta_{1}^{2})}{(n_{1}/2)(\delta_{2}\delta_{1}^{3}-\delta_{1}\delta_{2}^{3})^{2}}.
\end{equation}
\end{proposition}

\begin{corollary}\label{col_B}
Under Assumptions~\ref{asp:1} and~\ref{asp:2}, consider $\widehat{B}$ in~\eqref{eq.FD} with $\delta_{1}\neq\delta_{2}$ and $\delta_{1},\delta_{2}=\Theta(n_{1}^{\beta})$ for some $\beta<0$. Assume in addition that $\delta_2=c\delta_1$ for a fixed constant $c\in(0,1)$. Then, as $n_{1}\to+\infty$,
\[
\mathbb{E}\!\left[(\widehat{B}-B)^{2}\right]
=
\left(\frac{f^{(5)}(x_{0})}{120}\right)^{2}(\delta_{1}^{2}+\delta_{2}^{2})^{2}
+o(n_{1}^{4\beta})
+\frac{(\delta_{1}^{2}+\delta_{2}^{2})\sigma^{2}(x_{0})}
{n_{1}(\delta_{2}\delta_{1}^{3}-\delta_{1}\delta_{2}^{3})^{2}}
+o(n_{1}^{-1-6\beta}).
\]
In particular, for any $-1/6<\beta<0$,
\[
\lim_{n_{1}\to+\infty}\mathbb{E}\!\left[(\widehat{B}-B)^2\right]=0,
\]
and hence $\widehat{B}$ is consistent.
\end{corollary}

The two-point regression estimator in~\eqref{eq.FD} is the simplest instance of a more general regression-based approach for estimating $B$. In general, one may collect $n_1$ independent pilot outputs at (possibly repeated) perturbation sizes $\{\delta_j\}_{j=1}^{n_1}$:
\[
y_j \triangleq \frac{\hat f(x_0+\delta_j)-\hat f(x_0-\delta_j)}{2},\qquad j=1,\ldots,n_1.
\]
Under Assumption~\ref{asp:1}, these outputs satisfy the approximate two-regressor model
\[
y_j = f^{(1)}(x_0)\delta_j + B\delta_j^3 + \Big(R(\delta_j)\delta_j^5+\frac{\epsilon(\delta_j)}{2}\Big),
\]
so that estimating $(f^{(1)}(x_0),B)$ reduces to a least-squares problem.
Let $\bm y=(y_1,\ldots,y_{n_1})'$, 
$\bm X=\big[(\delta_1,\ldots,\delta_{n_1})',\,(\delta_1^3,\ldots,\delta_{n_1}^3)'\big]$,
and $\bm\beta=(f^{(1)}(x_0),B)'$. 
If \(\bm X'\bm X\) is nonsingular, the ordinary least-squares estimator is
$\widehat{\bm\beta}=(\bm X'\bm X)^{-1}\bm X'\bm y$, whose second component yields
\[
\widehat{B}
=
\frac{\Big(\sum_{j=1}^{n_{1}}\delta_{j}^{2}\Big)\Big(\sum_{j=1}^{n_{1}}\delta_{j}^{3}y_j\Big)
-\Big(\sum_{j=1}^{n_{1}}\delta_{j}^{4}\Big)\Big(\sum_{j=1}^{n_{1}}\delta_{j}y_j\Big)}
{\Big(\sum_{j=1}^{n_{1}}\delta_{j}^{2}\Big)\Big(\sum_{j=1}^{n_{1}}\delta_{j}^{6}\Big)-\Big(\sum_{j=1}^{n_{1}}\delta_{j}^{4}\Big)^{2}}.
\]
The estimator in~\eqref{eq.FD} can be viewed as a special case of this regression approach, obtained by using two distinct perturbation sizes and averaging over replications at each perturbation. In practice, multi-point regression designs may improve the finite-sample accuracy of $\widehat B$ and reduce the leading constants in its MSE, for instance by choosing a more favorable spread of perturbation sizes and leveraging averaging.

Nonetheless, importantly, even with these more general regression designs, under Assumptions~\ref{asp:1} and \ref{asp:2}, the convergence rate of $\mathbb{E}[(\widehat B-B)^2]$ with respect to the pilot budget $n_1$ remains governed by the same bias-variance orders as in Corollary~\ref{col_B}. Specifically, when the pilot perturbations are chosen on a common scale $\delta_j=\Theta(n_1^\beta)$ with $\beta<0$, and the regression design is non-degenerate, the truncation bias is of order $\delta^2$, yielding a squared-bias order \(\delta^4\), while the variance contribution scales as $1/(n_1\delta^6)$ up to constants. Consequently, richer regression designs generally do not change the convergence-rate structure
\[
\mathbb{E}[(\widehat B-B)^2]=\Theta(n_1^{4\beta})+\Theta(n_1^{-1-6\beta}),
\]
although they may improve the leading constants and finite-sample stability. For clarity, we therefore adopt the two-perturbation pilot design, while noting that multi-point regressions provide a viable practical alternative when one aims to improve finite-sample performance.

We now turn to the pilot estimator of $\sigma^2(x_0)$. Compared with the estimation of $B$, this estimator is simpler because it only relies on the sample variance of the pilot outputs at perturbation size $\delta_1$. Since $y_1^{(k)}$ is formed from two independent function evaluations at $x_0\pm\delta_1$, Assumption~\ref{asp:2} and the independent-sampling convention ($k=1$) imply \[ \text{Var}(y_1^{(k)}) = \frac{\text{Var}(\hat f(x_0+\delta_1))+\text{Var}(\hat f(x_0-\delta_1))}{4} = \frac{\sigma^2(x_0)}{2}+\Theta(\delta_1^s). \] Therefore, as long as $\delta_1\to0$, two times the sample variance of $\{y_1^{(k)}\}_{k=1}^{m}$ consistently estimates $\sigma^2(x_0)$ in probability under the mild moment condition stated below. This requirement is weaker than the scaling condition used in Corollary~\ref{col_B}; in particular, it is not limited to the case $\delta_1,\delta_2=\Theta(n_1^\beta)$ with $-1/6<\beta<0$.

\begin{proposition}\label{pro_sigma}
Under Assumption~\ref{asp:2}, consider $\widehat\sigma^2(x_0)$ in~\eqref{sigma_hat_FD_single} with $\delta_{1}\to0$ as $n_{1}\to\infty$. Suppose that the pilot outputs have uniformly bounded fourth moments in a neighborhood of $x_0$, i.e., there exists $\delta_0>0$ such that
\[
\sup_{0<\delta_1\le \delta_0}
\mathbb{E}\left[|y_1^{(k)}|^4\right]<\infty .
\]
Then \[ \widehat{\sigma}^{2}(x_{0})\xrightarrow{p}\sigma^{2}(x_{0}), \qquad \text{as } n_1\to\infty. \] 
\end{proposition}

The uniformly bounded fourth-moment condition in Proposition~\ref{pro_sigma} is a mild technical condition used to ensure the weak consistency of the sample variance. It is satisfied, for example, if the simulation output $\hat f(x)$ has uniformly bounded fourth moments in a neighborhood of $x_0$, i.e., if there exists $\delta_0>0$ such that \[ \sup_{|x-x_0|\le \delta_0}\mathbb E\big[|\hat f(x)|^4\big]<\infty. \] Indeed, since $y_1^{(k)}$ is formed from the difference of two simulation outputs at $x_0\pm\delta_1$, this condition implies a uniformly bounded fourth moment for $y_1^{(k)}$ when $\delta_1$ is sufficiently small. Such a moment condition is standard for establishing the weak consistency of sample variances and mainly rules out heavy-tailed outputs with unstable fourth moments near $x_0$.

The estimator in~\eqref{sigma_hat_FD_single} is a simple pilot estimator of the local simulation variance $\sigma^2(x_0)$. Other variance estimators are also possible, such as using outputs from another perturbation size, pooling sample variances across multiple perturbation sizes, or fitting a local variance regression from pilot samples. These alternatives may improve leading constants and finite-sample stability. However, they share the same standard rate structure with respect to the pilot budget and the local perturbation size. Specifically, under standard moment conditions and Assumption~\ref{asp:2}, the stochastic error of a variance estimator based on $n_1$ independent pilot samples is of order $O_p(n_1^{-1/2})$, while the local approximation bias is of order $O(\delta^s)$. Thus, the estimator in~\eqref{sigma_hat_FD_single} already achieves the standard convergence rate for variance estimation and is consistent whenever $n_1\to\infty$ and $\delta_1\to0$. Therefore, while richer variance estimators may improve leading constants and finite-sample stability, they are not needed for the asymptotic calibration goal of PC-CFD; the simple estimator in~\eqref{sigma_hat_FD_single} is sufficient for constructing the plug-in perturbation size in~\eqref{optimal perturbation}.

When CRNs are used, the two simulation outputs forming each finite-difference pair are generated using the same underlying random numbers, while different pairs remain independent. In this case, the relevant variance quantity is the effective variance of the paired difference rather than the marginal variance of an individual simulation output. Specifically, if this effective variance is locally represented by \(k\sigma^2(x_0)\), then two times the sample variance of the pilot half-differences consistently estimates \(k\sigma^2(x_0)\) directly. Therefore, the multiplier \(k\) and the marginal variance \(\sigma^2(x_0)\) need not be estimated separately, because the oracle perturbation rule depends only on their product. Accordingly, the same pilot-calibration procedure applies under CRN after interpreting \(\widehat{\sigma}^2(x_0)\) as an estimator of the effective paired-difference variance, provided this variance remains positive and of constant order as the perturbation size tends to zero.

\subsection{Calibrated Deployment of Finite-Difference Estimator}
Given the total simulation budget of $n$ single-run CFD replications, the pilot estimation stage uses $n_1$ replications to estimate the model quantities $B$ and $\sigma^2(x_0)$. The remaining budget $n_2=n-n_1$ is then allocated to the calibrated deployment stage for estimating the target gradient $f^{(1)}(x_0)$. In this stage, we plug the pilot estimates $\widehat B$ and $\widehat{\sigma}^2(x_0)$ into the oracle perturbation formula~\eqref{optimal perturbation} and obtain a calibrated perturbation size. This perturbation is then used in a standard CFD estimator based on the $n_2$ new replications that are independent of the pilot samples. The calibrated deployment stage is summarized in Algorithm~\ref{alg:pccfd_deploy}.

\begin{algorithm}[!h]
\caption{Calibrated Deployment Stage}\label{alg:pccfd_deploy}
\textbf{Input}: target point $x_0$, deployment budget $n_2$, pilot estimates $\widehat B$ and $\widehat{\sigma}^2(x_0)$.\\
Set
\[
\widehat{\delta}
\leftarrow
\left(\frac{\widehat{\sigma}^{2}(x_{0})}{4n_{2}\widehat{B}^{2}}\right)^{1/6}.
\]
\For{$k=1,\ldots,n_2$}{
    Generate two independent outputs $(\hat f(x_0+\widehat{\delta}))^{(k)}$ and $(\hat f(x_0-\widehat{\delta}))^{(k)}$.\\
    Compute
    \[
    Y^{(k)}(\widehat{\delta})
    \leftarrow
    \frac{(\hat f(x_0+\widehat{\delta}))^{(k)}-(\hat f(x_0-\widehat{\delta}))^{(k)}}{2\widehat{\delta}}.
    \]
}
Compute
\[
\widehat{\theta}_{\mathrm{PC}}
\leftarrow
\frac{1}{n_2}\sum_{k=1}^{n_2}Y^{(k)}(\widehat{\delta}).
\]
\Return $\widehat{\theta}_{\mathrm{PC}}$ as an estimate of $f^{(1)}(x_0)$.
\end{algorithm}

Challenging the prevailing belief that attaining oracle-optimal FD performance is difficult without prior knowledge of problem-specific quantities, the following theorem shows that PC-CFD achieves the oracle-optimal first-order MSE in~\eqref{optimal MSE}, even though its pilot estimation stage uses only an asymptotically negligible fraction of the total simulation budget.
\begin{theorem}\label{thm:pc_oracle}
Under Assumptions~\ref{asp:1} and~\ref{asp:2}, consider the PC-CFD estimator $\widehat{\theta}_{\mathrm{PC}}$ generated by Algorithms~\ref{alg:pccfd_main}-\ref{alg:pccfd_deploy}. Suppose that the pilot perturbations satisfy
\(
\delta_2=c\delta_1 \ \text{with}\ c\in(0,1)
\)
and
\(
\delta_1=\Theta(n_1^\beta)\ \text{with}\ -1/6<\beta<0.
\)
Suppose further that the moment condition in Proposition~\ref{pro_sigma} holds and that
\[
n_1=\Theta(n^\gamma),\qquad 0<\gamma<1.
\]
Then
\[
\lim_{n\to\infty}
n^{2/3}
\mathbb E\!\left[
\left(\widehat{\theta}_{\mathrm{PC}}-f^{(1)}(x_0)\right)^2
\right]
=
\mathcal R_{\mathrm{opt}}.
\]
Therefore, PC-CFD is oracle-optimal.
\end{theorem}

Theorem~\ref{thm:pc_oracle} shows that PC-CFD attains the same first-order MSE constant as the oracle CFD estimator, even though the oracle perturbation size depends on the unknown quantities \(B\) and \(\sigma^2(x_0)\). The pilot estimation stage is used only to learn these quantities sufficiently well for calibration. Indeed, the condition \(n_1=\Theta(n^\gamma)\) with \(0<\gamma<1\) implies \(n_1/n\to0\), so the pilot estimation stage consumes an asymptotically negligible fraction of the total simulation budget. Consequently, almost all replications are still available for the final CFD estimation.

The proof of Theorem~\ref{thm:pc_oracle} also explains why the calibration error propagated from the pilot estimation stage has an asymptotically negligible effect on the final estimator. The pilot estimates affect the calibrated deployment stage only through the plug-in perturbation constant
\[
\widehat{\alpha}
=
\left(\frac{\widehat{\sigma}^{2}(x_0)}{4\widehat B^2}\right)^{1/6}.
\]
Once \(\widehat B\xrightarrow{p} B\) and \(\widehat{\sigma}^2(x_0)\xrightarrow{p}\sigma^2(x_0)\), we have \[\widehat{\alpha}\xrightarrow{p}\alpha^*=\left(\frac{\sigma^2(x_0)}{4B^2}\right)^{1/6}.\] The first-order MSE constant of a CFD estimator with perturbation size $\delta=\alpha n^{-1/6}$ is minimized at \(\alpha^*\). Hence, its first-order derivative vanishes at the oracle choice, and small deviations of \(\widehat{\alpha}\) from \(\alpha^*\) affect the first-order MSE only through higher-order terms. This explains why a small pilot budget is sufficient: the pilot estimation stage does not need to estimate \(B\) and \(\sigma^2(x_0)\) at a fast rate. It only needs to estimate them consistently enough so that the deployment-stage perturbation constant is calibrated near the oracle choice.

This observation is central to the proposed methodology. It challenges the view that oracle-optimal FD performance is difficult to attain in practice because the optimal perturbation size depends on unknown, problem-specific quantities. PC-CFD shows that these quantities can be estimated sufficiently well using only a vanishing fraction of the simulation budget, and that the resulting calibration error does not alter the oracle first-order MSE. Thus, the procedure combines the adaptivity of data-driven calibration with the statistical efficiency of the oracle CFD rule.

We next refine the oracle-optimality result by characterizing the order of the remaining error beyond the leading term. To avoid ambiguity, we refer to the \emph{second-order remainder} as the difference between the MSE of PC-CFD and its oracle first-order approximation:
\[
\mathrm{Rem}_{\mathrm{PC}}(n)
\triangleq
\mathbb E\!\left[
\left(\widehat{\theta}_{\mathrm{PC}}-f^{(1)}(x_0)\right)^2
\right]
-
\mathcal R_{\mathrm{opt}}\,n^{-2/3}.
\]
Thus, Theorem~\ref{thm:pc_oracle} states that \(\mathrm{Rem}_{\mathrm{PC}}(n)=o(n^{-2/3})\). 
The following theorem provides a more explicit rate for this remainder, which helps explain how the pilot budget \(n_1\) and pilot perturbation size should be chosen in practice.

\begin{theorem}\label{thm:pc_second_order}
Under the same conditions of Theorem~\ref{thm:pc_oracle}, the second-order remainder of PC-CFD satisfies
\[
\mathrm{Rem}_{\mathrm{PC}}(n)
=
O\!\left(
n^{-2/3-\kappa(\beta,\gamma)}
\right),
\]
where
\[
\kappa(\beta,\gamma)
=
\min\left\{
1-\gamma,\,
-2\gamma s\beta,\,
-4\gamma\beta,\,
\frac{1}{3},\,
\frac{s}{6},\,
\gamma(6\beta+1)
\right\}.
\]
Equivalently,
\[
\mathbb E\!\left[
\left(\widehat{\theta}_{\mathrm{PC}}-f^{(1)}(x_0)\right)^2
\right]
=
\mathcal R_{\mathrm{opt}}\,n^{-2/3}
\left(1+O\!\left(n^{-\kappa(\beta,\gamma)}\right)\right).
\]
\end{theorem}

Theorem~\ref{thm:pc_second_order} complements Theorem~\ref{thm:pc_oracle} by quantifying the rate of the remaining error beyond the oracle first-order term. While Theorem~\ref{thm:pc_oracle} establishes that the remainder is \(o(n^{-2/3})\), Theorem~\ref{thm:pc_second_order} decomposes its order into interpretable components. These components correspond to distinct sources of inefficiency: the loss of deployment budget due to pilot allocation, the local approximation errors in the pilot estimates of \(B\) and \(\sigma^2(x_0)\), the stochastic error in pilot calibration, and the higher-order terms in the deployment-stage CFD expansion. Hence, the second-order result does not change the oracle-optimality conclusion, but it explains how finite-sample performance depends on the choices of \(n_1\) and the pilot perturbation rate.

This decomposition also provides practical guidance. A larger pilot budget improves calibration accuracy but leaves fewer replications for deployment; a smaller pilot budget preserves deployment samples but may increase calibration error. Similarly, pilot perturbations that shrink too slowly induce local approximation bias, whereas perturbations that shrink too quickly amplify the stochastic error in estimating \(B\). The second-order remainder captures these tradeoffs explicitly through \(\kappa(\beta,\gamma)\), and therefore offers a principled way to tune the pilot estimation stage within the broad first-order valid regime \(0<\gamma<1\) and \(-1/6<\beta<0\).

From a technical perspective, Theorems \ref{thm:pc_oracle} and \ref{thm:pc_second_order} separate the analysis into two layers. The first layer shows that consistency of the pilot calibration is enough for oracle first-order efficiency. The second layer tracks how pilot estimation errors, budget split, and deployment-stage approximation errors propagate into the second-order remainder. They collectively explain why a data-driven perturbation rule can match the oracle benchmark at first order, while still allowing a transparent second-order analysis for finite-sample implementation.

For a given local variance smoothness parameter \(s\) and pilot allocation exponent \(\gamma\), Theorem~\ref{thm:pc_second_order} suggests choosing the pilot perturbation exponent \(\beta\) to maximize \(\kappa(\beta,\gamma)\). The following corollary reports the resulting optimal choices of \(\beta\) and the corresponding second-order remainder rates.
\begin{corollary}\label{cor_discussion}
Under the same conditions of Theorem~\ref{thm:pc_oracle}, for each pair \((s,\gamma)\), suppose we choose \(\beta\in(-1/6,0)\) to maximize \(\kappa(\beta,\gamma)\). 
Then the optimal choice(s) of $\beta$ and the corresponding order of \(\mathrm{Rem}_{\mathrm{PC}}(n)\)
are given in Table~\ref{tab:second_order_tuning}.
\begin{table}[h]
\centering
\caption{Optimal choice(s) of \(\beta\) and the corresponding order of \(\mathrm{Rem}_{\mathrm{PC}}(n)\).}
\label{tab:second_order_tuning}
\begin{tabular}{|c|c|c|c|}
\hline
\(s\) & \(\gamma\) & Optimal \(\beta\) & Order of \(\mathrm{Rem}_{\mathrm{PC}}(n)\) \\
\hline
\([2,+\infty)\) 
& \([\frac{5}{7},1)\) 
& \(\left[\frac{1-2\gamma}{6\gamma},\frac{\gamma-1}{4\gamma}\right]\) 
& \(O\!\left(n^{\gamma-1-\frac{2}{3}}\right)\) \\
\hline
\([2,+\infty)\) 
& \((0,\frac{5}{7}]\) 
& \(-\frac{1}{10}\) 
& \(O\!\left(n^{-\frac{2}{5}\gamma-\frac{2}{3}}\right)\) \\
\hline
\([\frac{3}{2},2]\) 
& \([\frac{2s+6}{4s+6},1)\) 
& \(\left[\frac{1-2\gamma}{6\gamma},\frac{\gamma-1}{2\gamma s}\right]\) 
& \(O\!\left(n^{\gamma-1-\frac{2}{3}}\right)\) \\
\hline
\([\frac{3}{2},2]\) 
& \((0,\frac{2s+6}{4s+6}]\) 
& \(-\frac{1}{2s+6}\) 
& \(O\!\left(n^{-\frac{2s\gamma}{2s+6}-\frac{2}{3}}\right)\) \\
\hline
\((0,\frac{3}{2}]\) 
& \([1-\frac{s}{6},1)\) 
& \(\left[\frac{1-2\gamma}{6\gamma},\frac{\gamma-1}{2\gamma s}\right]\) 
& \(O\!\left(n^{\gamma-1-\frac{2}{3}}\right)\) \\
\hline
\((0,\frac{3}{2}]\) 
& \(\left[\frac{s+3}{6},1-\frac{s}{6}\right]\) 
& \(\left[\frac{s-6\gamma}{36\gamma},\frac{-1}{12\gamma}\right]\) 
& \(O\!\left(n^{-\frac{s}{6}-\frac{2}{3}}\right)\) \\
\hline
\((0,\frac{3}{2}]\) 
& \((0,\frac{s+3}{6}]\) 
& \(-\frac{1}{2s+6}\) 
& \(O\!\left(n^{-\frac{2s\gamma}{2s+6}-\frac{2}{3}}\right)\) \\
\hline
\end{tabular}
\end{table}
\end{corollary}

Two special cases are particularly useful for interpreting the tuning rule in Corollary~\ref{cor_discussion}. First, consider the case where the simulation noise variance is locally constant across nearby function evaluation points. This corresponds to the limiting case \(s=+\infty\). Optimizing further over \(\gamma\), the best second-order rate in this locally constant-variance case is obtained at
\(\gamma=5/7\), \(\beta=-1/10\), which gives \(\mathrm{Rem}_{\mathrm{PC}}(n)=O\!\left(n^{-20/21}\right).\) Second, consider the case \(s=1\), where the local variance changes linearly with the perturbation size. In this case, the best second-order rate is
\(\mathrm{Rem}_{\mathrm{PC}}(n)=O\!\left(n^{-5/6}\right),\)
which is attained by any $2/3\leq\gamma\leq5/6$ with $\frac{1-6\gamma}{36\gamma}\leq\beta\leq-\frac{1}{12\gamma}$. 

These two cases illustrate how the smoothness of the local variance affects the preferred pilot design. When the variance is locally constant, the local variance-approximation penalty disappears, and the best balance is achieved at \(\gamma=5/7\) and \(\beta=-1/10\). When the variance varies linearly around \(x_0\), the local variance approximation error becomes more pronounced, and the optimal pilot allocation shifts to the broader range \(2/3\le\gamma\le5/6\). 

\section{Generalizations to Multi-Dimensional Gradient Estimation}\label{sec:multi}
We now extend the PC-CFD framework to multi-dimensional gradient estimation. Let \(f:\mathbb R^p\to\mathbb R\), \(p\ge2\), be the performance measure of interest. For any chosen \(\bm x\in\mathbb R^p\), suppose that we can observe an unbiased simulation output \(\hat f(\bm x)\) satisfying $\mathbb{E}[\hat{f}(\bm{x})]=f(\bm{x})$ and $\text{Var}(\hat{f}(\bm{x}))=\sigma^{2}(\bm{x})$. The goal is to estimate the gradient \(\nabla f(\bm x_0)\) at a fixed point \(\bm x_0\in\mathbb R^p\). 

The multi-dimensional setting gives rise to two natural finite-difference strategies. The first uses deterministic coordinate-wise perturbations and estimates each partial derivative separately. The second uses randomized simultaneous perturbations and estimates all gradient components jointly from each pair of function evaluations. We show that the pilot-calibrated framework extends to both strategies, but the problem-specific quantities required for calibration differ substantially.

Throughout this section, for a gradient estimator $\widehat{\bm{\theta}}=(\widehat{\theta}_{1},\ldots,\widehat{\theta}_{p})^{'}$, we evaluate its performance by the mean squared error
\begin{equation} \label{MSE_Multi} 
\text{MSE} = \mathbb E\!\left[ \|\widehat{\bm\theta}-\nabla f(\bm x_0)\|_2^2 \right] = \sum_{i=1}^p \mathbb E\!\left[ \left(\widehat\theta_i-(\nabla f(\bm x_0))_i\right)^2 \right]. 
\end{equation}
Corresponding to the regularity conditions imposed in the single-dimensional setting, we make the following assumptions. 
\begin{assumption}\label{asp:multi_smooth}
	The performance function $f(\bm{x})$ is five-times continuously differentiable in a neighborhood of $\bm{x}_{0}$.
\end{assumption}
\begin{assumption}\label{asp:multi_var}
	The output variance $\sigma^{2}(\bm x)$ varies smoothly around $\bm x_{0}$, in the sense that as $\|\bm\Delta\|_2\to0$, $\sigma^{2}(\bm x_{0}+\bm\Delta)=\sigma^{2}(\bm x_{0})+\Theta(\|\bm\Delta\|_2^s)$ for some $s>0$.
\end{assumption}

\subsection{Deterministic Coordinate-Wise Perturbations}
We first consider deterministic coordinate-wise perturbations. Let \(\bm e_i\) denote the \(i\)-th coordinate vector, and define \[ B_i \triangleq \frac{\nabla_{iii}^3 f(\bm x_0)}{6}, \qquad i=1,\ldots,p, \] where \(\nabla_{iii}^3 f(\bm x_0)\) denotes the third-order partial derivative of \(f\) with respect to the \(i\)-th coordinate. Suppose that the total budget is split equally across the \(p\) coordinates, so that each coordinate receives \(n/p\) single-run CFD replications. Applying the single-dimensional oracle rule~\eqref{optimal MSE} separately to each coordinate yields the oracle-optimal MSE \[ \sum_{i=1}^p 3\left( \frac{B_i\sigma^2(\bm x_0)}{4(n/p)} \right)^{2/3} (1+o(1)). \] Define \[ B_{\mathrm{cw}} \triangleq \left( \sum_{i=1}^p B_i^{2/3} \right)^{3/2},\] then the deterministic coordinate-wise oracle-optimal MSE can be written as 
\begin{equation} \label{eq:det_oracle_mse} 
\text{MSE}_{\mathrm{CW,opt}} = 3\left( \frac{p\sigma^2(\bm x_0)B_{\mathrm{cw}}}{4n} \right)^{2/3} (1+o(1)). 
\end{equation} 
Thus, to attain this benchmark without oracle knowledge, one needs to estimate the \(p\) curvature quantities \(B_1,\ldots,B_p\) and the local variance \(\sigma^2(\bm x_0)\).

The coordinate-wise PC-CFD procedure applies the single-dimensional pilot-calibration framework separately to each coordinate. For each \(i\), the pilot estimation stage uses perturbations along \(\bm e_i\) to estimate \(B_i\) and \(\sigma^2(\bm x_0)\), while the calibrated deployment stage uses these estimates to determine the coordinate-specific perturbation size. For notational clarity, we write the pilot estimation procedure in Algorithm~\ref{alg:pccfd_pilot} using the functional notation
\[ \text{Pilot Estimation Stage} (f,\hat f;x_0,n_1,\delta_1,\delta_2) \]
where \((f,\hat f)\) denotes the underlying response function and its simulation output, and \((x_0,n_1,\delta_1,\delta_2)\) denotes the target point, pilot budget, and two pilot perturbations, respectively. Similarly, we write the calibrated deployment procedure in Algorithm~\ref{alg:pccfd_deploy} as
\[  \text{Calibrated Deployment Stage} (f,\hat f;x_0,n_2,\widehat B,\widehat\sigma^2(x_0)) \]
where \(n_2\) denotes the deployment budget, and \((\widehat B,\widehat\sigma^2(x_0))\) are the estimates obtained from the pilot estimation stage. The complete procedure is summarized in Algorithm~\ref{alg:multi_det_pccfd}. For simplicity, assume that \(n/p\) and \(n_1\) are integers and that \(0<n_1<n/p\).

\begin{algorithm}[!h]
\caption{Coordinate-Wise PC-CFD for Multi-Dimensional Gradients}
\label{alg:multi_det_pccfd}
\textbf{Input}: target point \(\bm x_0\), total budget \(n\), pilot budget \(n_1\) per coordinate, pilot perturbations \(\delta_1\neq\delta_2>0\).\\
Set \(n_2\leftarrow n/p-n_1\).\\
\For{\(i=1,\ldots,p\)}{
    Define the single-dimensional restriction and its simulation output by
    \[
    g_i(t)\triangleq f(\bm x_0+t\bm e_i),
    \qquad
    \hat g_i(t)\triangleq \hat f(\bm x_0+t\bm e_i),
    \qquad t\in\mathbb R.
    \]\\
    Apply Algorithm~\ref{alg:pccfd_pilot} to \((g_i,\hat g_i)\) at target point \(t=0\) with pilot budget \(n_1\) and pilot perturbations \(\delta_1,\delta_2\) to obtain
    \[
    (\widehat B_i,\widehat\sigma_i^2)
    \leftarrow
    \text{Pilot Estimation Stage}(g_i,\hat g_i;0,n_1,\delta_1,\delta_2).
    \]\\
    Apply Algorithm~\ref{alg:pccfd_deploy} to \((g_i,\hat g_i)\) at target point \(t=0\) with deployment budget \(n_2\), pilot estimates \(\widehat B_i\) and \(\widehat\sigma_i^2\) to obtain
    \[
    \widehat\theta_{\mathrm{CWPC},i}
    \leftarrow
    \text{Calibrated Deployment Stage}(g_i,\hat g_i;0,n_2,\widehat B_i,\widehat\sigma_i^2).
    \]
}
\Return
\(
\widehat{\bm\theta}_{\mathrm{CWPC}}
=
(\widehat\theta_{\mathrm{CWPC},1},\ldots,\widehat\theta_{\mathrm{CWPC},p})'
\) as an estimate of \(\nabla f(\bm x_0)\).
\end{algorithm}

The following result is a direct multi-dimensional extension of Theorem~\ref{thm:pc_oracle}. It shows that coordinate-wise PC-CFD attains the same first-order MSE as the coordinate-wise oracle CFD rule.
\begin{corollary}
\label{cor:multi_det_oracle}
Under Assumption~\ref{asp:multi_smooth} with $B_i\neq0,\ i=1,\ldots,p$ and Assumption~\ref{asp:multi_var}, consider the coordinate-wise PC-CFD estimator \(\widehat{\bm\theta}_{\mathrm{CWPC}}\) in Algorithm~\ref{alg:multi_det_pccfd}. Suppose that the pilot perturbations satisfy \(\delta_2=c\delta_1\) with \(c\in(0,1)\) and \(\delta_1=\Theta(n_1^\beta)\) with \(-1/6<\beta<0\). Assume further that the moment condition in Proposition~\ref{pro_sigma} holds for each coordinate, and that
\[
n_1=\Theta(n^\gamma),\qquad 0<\gamma<1.
\]
Then
\[
\lim_{n\to\infty}n^{2/3}\mathbb E\!\left[
\|\widehat{\bm\theta}_{\mathrm{CWPC}}-\nabla f(\bm x_0)\|_2^2
\right]
=
3\left(
\frac{p\sigma^2(\bm x_0)B_{\mathrm{cw}}}{4}
\right)^{2/3}.
\]
Consequently, coordinate-wise PC-CFD is oracle-optimal relative to the coordinate-wise CFD benchmark in \eqref{eq:det_oracle_mse}.
\end{corollary}

The implication is the same as that in the single-dimensional case, but now applied coordinate by coordinate. The pilot estimation stage only needs to estimate the coordinate-wise curvature constants \(B_i\) and the local variance sufficiently well to calibrate the coordinate-specific perturbation sizes. Since \(n_1/n\to0\), the total pilot fraction remains asymptotically negligible, while the calibrated deployment stage attains the same leading MSE constant as if all \(B_i\)'s and \(\sigma^2(\bm x_0)\) were known in advance.

\subsection{Randomized Simultaneous Perturbations}
Coordinate-wise perturbations estimate one partial derivative at a time. In high-dimensional stochastic optimization, it is also common to use randomized simultaneous perturbations, where each pair of function evaluations perturbs all coordinates simultaneously and produces a vector-valued gradient estimate. Such schemes include simultaneous perturbation methods and Gaussian smoothing methods.

We focus on simultaneous perturbation (SP) with Rademacher directions. Let
\[
\bm\Delta_j=(\delta_{1,j},\ldots,\delta_{p,j})',
\qquad j=1,\ldots,n,
\]
where all \(\delta_{i,j}\)'s are mutually independent Rademacher random variables, taking values in \(\{-1,1\}\) with equal probability. For a scaling parameter \(h>0\), define
\[
\phi(\bm\Delta_j)
=
(\delta_{1,j}^{-1},\ldots,\delta_{p,j}^{-1})',
\]
and consider the simultaneous-perturbation estimator
\begin{equation}
\label{eq:sp_estimator}
\widehat{\bm\theta}_{\mathrm{SP}}
=
\frac{1}{n}\sum_{j=1}^n
\frac{
\hat f(\bm x_0+h\bm\Delta_j)
-
\hat f(\bm x_0-h\bm\Delta_j)
}
{2h}
\phi(\bm\Delta_j).
\end{equation}
To identify the oracle perturbation scale and the corresponding first-order MSE, we establish the following proposition.
\begin{proposition}
\label{pro_SP}
Under Assumptions~\ref{asp:multi_smooth} and~\ref{asp:multi_var}, consider the simultaneous-perturbation estimator~\eqref{eq:sp_estimator}. Then, as \(h\to0\) and \(n\to\infty\),
\[
\mathbb E\!\left[
\|\widehat{\bm\theta}_{\mathrm{SP}}-\nabla f(\bm x_0)\|_2^2
\right]
=
(B_{\mathrm{sp}}^2+o(1)) h^4
+
\frac{p\sigma^2(\bm x_0)}{2nh^2}(1+o(1)),
\]
where 
\[
G_i
\triangleq
\nabla_{iii}^3 f(\bm x_0)
+
3\sum_{k\ne i}\nabla_{ikk}^3 f(\bm x_0),
\qquad i=1,\ldots,p,
\]
and
\[
B_{\mathrm{sp}}
\triangleq
\frac{1}{6}
\left(
\sum_{i=1}^p G_i^2
\right)^{1/2}.
\]
For the non-degenerative case of \(B_{\mathrm{sp}}>0\), the oracle-optimal MSE is
\begin{equation}
\label{eq:sp_oracle_mse}
\text{MSE}_{\mathrm{SP,opt}}
=
3\left(
\frac{p\sigma^2(\bm x_0)B_{\mathrm{sp}}}{4n}
\right)^{2/3}
(1+o(1)),
\end{equation}
which is attained by \(\widehat{\bm\theta}_{\mathrm{SP}}\) with the oracle perturbation scale
\begin{equation}
\label{eq:sp_oracle_perturb}
h=\left(\frac{p\sigma^{2}(\bm{x}_{0})}{4nB_{\mathrm{sp}}^{2}}\right)^{1/6}.
\end{equation}
\end{proposition}

The quantity \(G_i/6\) is the leading \(h^2\)-bias coefficient of the \(i\)-th component of the simultaneous-perturbation estimator. Unlike coordinate-wise CFD, this bias coefficient depends not only on the pure third derivative \(\nabla_{iii}^3 f(\bm x_0)\), but also on the pairwise-mixed third derivatives \(\nabla_{ikk}^3 f(\bm x_0)\). This dependence constitutes the main structural difference between deterministic coordinate-wise perturbations and randomized simultaneous perturbations.

Proposition~\ref{pro_SP} shows that simultaneous perturbation attains the same \(n^{-2/3}\) first-order rate as coordinate-wise CFD, but with a different curvature constant. Its oracle perturbation scale depends on \(B_{\mathrm{sp}}\), which aggregates the contributions of the pure and pairwise-mixed third derivatives that arise under randomized perturbation directions. Therefore, a pilot-calibrated simultaneous-perturbation method must consistently estimate \(B_{\mathrm{sp}}\) and \(\sigma^2(\bm x_0)\). In the pilot estimation stage developed below, this is accomplished by directly estimating the quantities \(G_1,\ldots,G_p\) and then combining them to form an estimator of \(B_{\mathrm{sp}}\).

The pilot estimation stage estimates the quantities \(G_1,\ldots,G_p\) and \(\sigma^2(\bm x_0)\) using randomized simultaneous perturbations. For the \(j\)-th replication of Rademacher direction, define the half-difference output at perturbation scale \(h>0\) by
\[
y(h\bm \Delta_j)
\triangleq
\frac{
\hat f(\bm x_0+h\bm \Delta_j)
-
\hat f(\bm x_0-h\bm \Delta_j)
}{2}.
\]
Taking expectation over both the simulation outputs and the randomized direction, a Taylor expansion gives
\begin{equation}\label{taylor_Gi}
\mathbb{E}[y(h\bm \Delta_j)\delta_{i,j}]
=
h(\nabla f(\bm x_0))_{i}
+
\frac{h^3}{6}G_i
+
O(h^5),
\qquad i=1,\ldots,p.
\end{equation}
Therefore, we use two distinct perturbation scales \(h_1\neq h_2>0\) and split the pilot budget \(n_1\) evenly across them. For simplicity, assume that $n_1$ is even and let $m=n_1/2$. For each \(j=1,\ldots,m\), the same Rademacher direction \(\bm\Delta_j\) is used at both perturbation scales, while all simulation outputs are generated independently. Using the paired outputs \(\{(y(h_1\bm \Delta_j),y(h_2\bm \Delta_j))\}_{j=1}^{m}\), we estimate $G_i$ by
\begin{equation}\label{hat_Gi}
    \widehat{G}_i \triangleq \frac{1}{m} \sum_{j=1}^{m} \frac{ 6\left(h_2 y(h_1\bm \Delta_j)-h_1 y(h_2\bm \Delta_j)\right)\delta_{i,j} }{ h_2 h_1^3-h_1 h_2^3 }, \qquad i=1,\ldots,p.
\end{equation}
The curvature aggregate \(B_{\mathrm{sp}}\) required by the oracle perturbation rule is then estimated by 
\begin{equation}\label{hat_Bsp} 
\widehat B_{\mathrm{sp}} \triangleq \frac{1}{6} \left( \sum_{i=1}^p \widehat G_i^2 \right)^{1/2}. \end{equation}
We estimate \(\sigma^2(\bm x_0)\) by two times the sample variance of \(\{y(h_1\bm \Delta_j)\}_{j=1}^{m}\), i.e.,
\begin{gather}
	\widehat{\sigma}^{2}(\bm x_{0})=\frac{2\|\bm{y}(h_1)-\bar{y}(h_1)\vmathbb{1}_{m\times 1}\|_{2}^{2}}{m-1}\label{sigma_hat_FD_multi}
\end{gather}
where $\bm{y}(h_1)=(y(h_1\bm \Delta_1),\ldots,y(h_1\bm \Delta_m))'$, $\bar{y}(h_1)=m^{-1}\sum_{j=1}^{m}y(h_1\bm \Delta_j)$. In the calibrated deployment stage, we plug \(\widehat B_{\mathrm{sp}}\) and \(\widehat{\sigma}^2(\bm x_0)\) into the oracle perturbation formula~\eqref{eq:sp_oracle_perturb} and obtain a calibrated perturbation scale. The calibrated scale is then used in the simultaneous-perturbation estimator~\eqref{eq:sp_estimator} based on the $n_2=n-n_1$ new replications that are independent of the pilot samples. Algorithm~\ref{alg:pcsp} summarizes the resulting pilot-calibrated simultaneous-perturbation procedure.

The following theorem states that simultaneous-perturbation PC-CFD attains the corresponding oracle benchmark.
\begin{theorem}
\label{thm:pcsp}
Under Assumption~\ref{asp:multi_smooth} with \(B_{\mathrm{sp}}>0\) and Assumption~\ref{asp:multi_var}, consider the simultaneous-perturbation PC-CFD estimator $\widehat{\bm\theta}_{\mathrm{SPPC}}$ in Algorithm~\ref{alg:pcsp}. Suppose that the pilot perturbation scales satisfy \(h_2=ch_1\) with \(c\in(0,1)\) and \(h_1=\Theta(n_1^\beta)\) with \(-1/6<\beta<0\). Assume further that the uniform fourth-moment condition analogous to that in Proposition~\ref{pro_sigma} holds for the randomized pilot outputs, and that 
\[
n_1=\Theta(n^\gamma),\qquad 0<\gamma<1.
\]
Then
\[
\lim_{n\to\infty}n^{2/3}\mathbb E\!\left[
\|\widehat{\bm\theta}_{\mathrm{SPPC}}-\nabla f(\bm x_0)\|_2^2
\right]
=
3\left(
\frac{p\sigma^2(\bm x_0)B_{\mathrm{sp}}}{4}
\right)^{2/3}.
\]
Moreover, its second-order remainder admits the same upper bound as in Theorem~\ref{thm:pc_second_order}, namely
\[
\mathbb E\!\left[
\|\widehat{\bm\theta}_{\mathrm{SPPC}}-\nabla f(\bm x_0)\|_2^2
\right]=3\left(
\frac{p\sigma^2(\bm x_0)B_{\mathrm{sp}}}{4}
\right)^{2/3}n^{-2/3}\left(1+O\!\left(
n^{-\kappa(\beta,\gamma)}
\right)\right),
\]
where \(\kappa(\beta,\gamma)=
\min\left\{
1-\gamma,\,
-2\gamma s\beta,\,
-4\gamma\beta,\,
1/3,\,
s/6,\,
\gamma(6\beta+1)
\right\}\).
\end{theorem}

Theorem~\ref{thm:pcsp} demonstrates that the pilot-calibration framework extends beyond coordinate-wise finite differences to randomized simultaneous perturbations. In this setting, the pilot stage estimates the curvature aggregate \(B_{\mathrm{sp}}\) and the local variance \(\sigma^2(\bm x_0)\), which jointly determine the oracle perturbation scale for the simultaneous-perturbation estimator. The result parallels the single-dimensional analysis. The consistency of the pilot estimators ensures first-order oracle efficiency, while the second-order remainder quantifies the effects of pilot approximation bias, pilot stochastic error, budget splitting, and the higher-order bias and variance terms arising in the deployment stage.

\begin{algorithm}[!h]
\caption{Simultaneous-Perturbation PC-CFD for Multi-Dimensional Gradients}
\label{alg:pcsp}
\textbf{Input}: target point \(\bm x_0\), total budget \(n\), pilot budget \(n_1\), pilot perturbation scales \(h_1\neq h_2>0\).\\
Set \(m\leftarrow n_1/2\), \(n_2\leftarrow n-n_1\).\\ 
\textbf{Pilot estimation stage}:\\ 
\For{\(j=1,\ldots,m\)}{ 
Generate an independent Rademacher direction \(\bm\Delta_j = (\delta_{1,j},\ldots,\delta_{p,j})'.\)\\
Compute \( y(h_r\bm\Delta_j) \leftarrow \frac{\hat f(\bm x_0+h_r\bm\Delta_j) - \hat f(\bm x_0-h_r\bm\Delta_j)}{2},\ r=1,2. \) 
} 
\For{\(i=1,\ldots,p\)}{
Estimate \(G_i\) by \( \widehat G_i \leftarrow \frac{1}{m} \sum_{j=1}^{m} \frac{6\left(h_2 y(h_1\bm \Delta_j)-h_1 y(h_2\bm \Delta_j)\right)\delta_{i,j}}{h_2 h_1^3-h_1 h_2^3 }. \) 
} 
Set \( \widehat B_{\mathrm{sp}} \leftarrow \frac{1}{6}\left( \sum_{i=1}^p\widehat G_i^2 \right)^{1/2}. \)\\
Compute \( \bar{y}(h_1) = \frac{1}{m}\sum_{j=1}^m y(h_1\bm \Delta_j)\), \(\widehat\sigma^2(\bm x_0) \leftarrow \frac{2}{m-1}\sum_{j=1}^{m}(y(h_1\bm \Delta_j)-\bar{y}(h_1))^2. \)\\ 
\textbf{Calibrated deployment stage}:\\
Set the calibrated perturbation scale \( \widehat h \leftarrow \left( \frac{ p\widehat\sigma^2(\bm x_0) }{ 4n_2\widehat B_{\mathrm{sp}}^2 } \right)^{1/6}. \) \\
\For{\(j=1,\ldots,n_2\)}{ 
Generate a new independent Rademacher direction \(\bm\Delta_j = (\delta_{1,j},\ldots,\delta_{p,j})'.\) 
} 
Set \[ \widehat{\bm\theta}_{\mathrm{SPPC}} \leftarrow \frac{1}{n_2} \sum_{j=1}^{n_2}\frac{
\hat f(\bm x_0+\widehat h\bm\Delta_j)
-
\hat f(\bm x_0-\widehat h\bm\Delta_j)
}
{2\widehat h}
\phi(\bm\Delta_j). \]\\
\Return \( \widehat{\bm\theta}_{\mathrm{SPPC}} \) as an estimate of \(\nabla f(\bm x_0)\).
\end{algorithm}

\subsection{Comparison of Coordinate-Wise and Simultaneous-Perturbation PC-CFD Estimators}
Under the same total simulation budget, the coordinate-wise and simultaneous-perturbation PC-CFD estimators have oracle first-order MSEs of the same functional form. Specifically,
\[ \text{MSE}_{\mathrm{CW,opt}} = 3\left( \frac{p\sigma^2(\bm x_0)B_{\mathrm{cw}}}{4n} \right)^{2/3} (1+o(1)), \]
where \[ B_{\mathrm{cw}} = \left( \sum_{i=1}^p \left( \frac{\nabla_{iii}^3 f(\bm x_0)}{6} \right)^{2/3} \right)^{3/2}, \]
whereas
\[ \text{MSE}_{\mathrm{SP,opt}} = 3\left( \frac{p\sigma^2(\bm x_0)B_{\mathrm{sp}}}{4n} \right)^{2/3} (1+o(1)), \]
with \[ B_{\mathrm{sp}} = \frac{1}{6} \left( \sum_{i=1}^p G_i^2 \right)^{1/2}, \qquad G_i = \nabla_{iii}^3 f(\bm x_0) + 3\sum_{k\ne i} \nabla_{ikk}^3 f(\bm x_0). \]
Thus, the oracle first-order MSE comparison between the two estimators is determined by \(B_{\mathrm{cw}}\) and \(B_{\mathrm{sp}}\). Because \(B_{\mathrm{sp}}\) depends on both pure and mixed third-order derivatives, the mixed-derivative terms may either amplify or offset the pure third-order derivative terms. Consequently, neither curvature constant uniformly dominates the other, and neither estimator is universally superior in terms of oracle first-order MSE.

A sufficient condition under which simultaneous-perturbation PC-CFD has no larger oracle first-order MSE is
\[ \left| \nabla_{iii}^3 f(\bm x_0) + 3\sum_{k\ne i} \nabla_{ikk}^3 f(\bm x_0) \right| \leq \left| \nabla_{iii}^3 f(\bm x_0) \right|, \qquad i=1,\ldots,p. \]
Indeed, under this condition,
\begin{gather*}
B_{\mathrm{sp}} \leq \frac{1}{6} \left( \sum_{i=1}^p \left| \nabla_{iii}^3 f(\bm x_0) \right|^2 \right)^{1/2} \leq \frac{1}{6} \left( \sum_{i=1}^p \left( \nabla_{iii}^3 f(\bm x_0) \right)^{2/3} \right)^{3/2} = B_{\mathrm{cw}}. 
\end{gather*}
A notable special case occurs when all mixed third-order derivatives vanish, i.e.,
\[
\nabla_{ikk}^3 f(\bm x_0)=0,
\qquad i\ne k.
\]
Simultaneous-perturbation PC-CFD may then be particularly attractive because each pair of function evaluations contributes simultaneously to the estimation of all gradient components. 

The two estimators also differ in how the pilot and deployment samples are used. Coordinate-wise PC-CFD estimates the \(p\) pure curvature coefficients separately along the coordinate directions, and each deployment replication contributes to only one gradient component. By contrast, simultaneous-perturbation PC-CFD directly estimates the \(p\) curvature aggregates \(G_1,\ldots,G_p\) using the same set of randomized directions, and each deployment replication contributes to all gradient components. In particular, the simultaneous-perturbation pilot does not require the individual estimation of the \(p(p-1)\) mixed third-order derivative terms. Nevertheless, its curvature estimates are affected by the variability of the randomized directions and combine the effects of multiple pure and mixed third-order derivatives. The preferred estimator therefore depends on both the third-order derivative structure and the benefit of sharing function evaluations across gradient components. Simultaneous-perturbation PC-CFD can be advantageous when the mixed-derivative effects are weak or offsetting and evaluation sharing is valuable, whereas coordinate-wise PC-CFD provides a more decoupled and transparent calibration procedure when the mixed third-order structure is unfavorable or difficult to exploit.

\section{Numerical Results}\label{sec:NR}
This section evaluates the finite-sample performance of the proposed pilot-calibrated finite-difference estimators. The numerical study is organized around two properties that motivate the methodology. The first is \emph{robustness}: an estimator should perform consistently well across heterogeneous problem instances, whereas a conventional finite-difference estimator with a prescribed perturbation constant may perform well for some instances but poorly for others. The second is \emph{near-oracle performance}: the MSE of PC-CFD should be close to the oracle benchmark, namely an estimator that uses the instance-specific optimal perturbation size based on the oracle model quantities.

We first study six single-dimensional test functions for which the oracle perturbation size can be computed, thereby allowing us to examine both robustness and proximity to the oracle benchmark. We then consider an \(M/M/1\) queueing system, where the relevant curvature and variance quantities are not assumed to be known. Finally, we investigate multi-dimensional gradient estimation using \(M/M/p\) queueing systems and a deep-neural-network gradient-checking example. These experiments compare coordinate-wise PC-CFD and simultaneous-perturbation PC-CFD with conventional finite-difference estimators based on prescribed perturbation constants. For every reported MSE, we use 10,000 independent experimental replications. For all MSE comparisons, the competing methods use the same total simulation budget; the neural-network experiment instead reports the budget required to reach a common accuracy threshold. 

\subsection{Single-Dimensional Experiments}

\subsubsection{Benchmark Estimators and Test Functions}
For the six analytical examples, PC-CFD is compared with the following benchmarks:
\begin{itemize}
    \item \textbf{Oracle-CFD}: This oracle estimator knows \(B\) and \(\sigma^2(x_0)\), uses
    \[
    \delta
    =
    \left(
    \frac{\sigma^2(x_0)}{4B^2}
    \right)^{1/6}
    n^{-1/6},
    \]
    and averages \(n\) independent single-run CFD replications.

    \item \textbf{CFD-\(\boldsymbol{ k }\), \(k=1,2,3\)}: These conventional estimators use prescribed surrogate values
    \((\widetilde B_k,\widetilde\sigma_k^2(x_0))\) and set
    \[
    \delta
    =
    \left(
    \frac{\widetilde\sigma_k^2(x_0)}
    {4\widetilde B_k^2}
    \right)^{1/6}
    n^{-1/6}.
    \]
    The surrogate values are fixed in advance and are not adapted to the particular evaluation point.
\end{itemize}

For all six functions, the simulation observation is
\[
\hat f(x)=f(x)+\varepsilon,
\qquad
\varepsilon\sim N(0,0.05),
\]
so that \(\sigma^2(x_0)=0.05\). The MSE is reported over a range of target points \(x_0=x\), which generates heterogeneous curvature values and therefore tests whether a method remains robust when the same prescribed perturbation rule is used across different instances.

\paragraph{Function 1.} Consider
\(
f(x)=1-6x+36x^2-53x^3+22x^5
\)
as in \citet{de2011nonparametric,de2015smoothed}. In this case,
\(
B=220x^2-53.
\)
The three fixed-tuning benchmarks use
\(
(\widetilde B_1,\widetilde B_2,\widetilde B_3)=(5,10,20)\), and \(
\widetilde\sigma_k^2(x_0)=0.05.
\)
Figure~\ref{fig.1} illustrates both main properties of PC-CFD. Across the considered values of \(x\), its MSE generally remains close to that of Oracle-CFD. The fixed-tuning estimators can be competitive around points at which their prescribed curvature value happens to be close to the oracle \(B\), but their performance deteriorates substantially elsewhere. In particular, PC-CFD performs better than the fixed-tuning estimators when \(x<0.3\) or \(x>0.6\), and at \(x=1\) the MSE of CFD-1 is more than ten times that of PC-CFD. These illustrate the comparatively much more robust performance of PC-CFD against fixed-tuning estimators across the full range of instances while remaining close to the oracle benchmark.

\paragraph{Function 2.}
Consider
\(
f(x)=\sin(2\pi x)
\)
as in
\citet{de2011nonparametric,de2015smoothed,dai2016optimal,mahmoud2020robust}.
Here, $B=-4\pi^{3}\cos(2\pi x)/3$. The fixed-tuning benchmarks use
\(
(\widetilde B_1,\widetilde B_2,\widetilde B_3)=(3,6,9)\),
and
\( \widetilde\sigma_k^2(x_0)=0.05.
\)
Like Figure~\ref{fig.1}, Figure~\ref{fig.2} illustrates the robust performance of PC-CFD, but it also highlights the transition from a small-budget regime to the asymptotic regime. When \(n=40\), the pilot estimates are based on a limited number of observations, and PC-CFD does not yet dominate every fixed-tuning benchmark. This can be attributed to the finite-sample cost of calibration. As the budget increases, the pilot estimates become more accurate, the pilot fraction becomes relatively small, and PC-CFD moves toward Oracle-CFD. At \(n=500\), PC-CFD performs better than all three fixed-tuning estimators across the considered
instances. Note that, even though PC-CFD does not dominate every fixed-tuning benchmark in the small-sample case, a priori we would not know if a specific fixed-tuning is a ``good" choice or not, and thus, even in this case, PC-CFD is a safer choice than using fixed-tuning.

\begin{figure}[!h]
	\centering
	\begin{subfigure}[b]{0.42\textwidth}
		\centering
		\includegraphics[width=\textwidth]{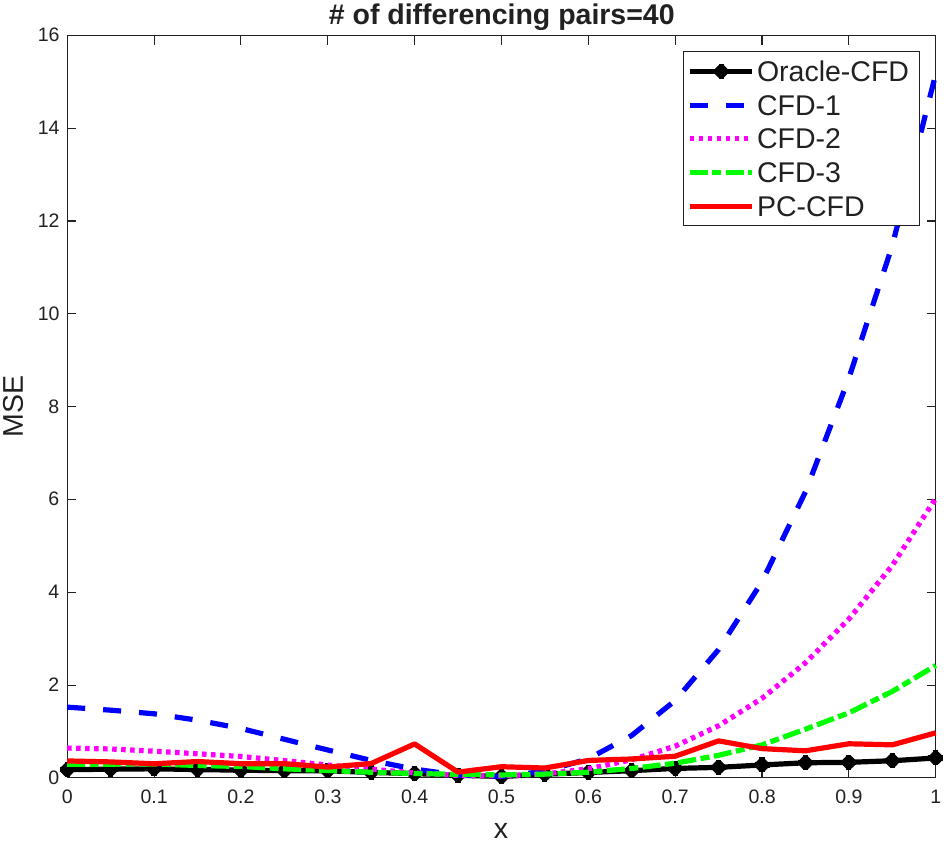}
	\end{subfigure}
	\hfill
	\begin{subfigure}[b]{0.42\textwidth}
		\centering
		\includegraphics[width=\textwidth]{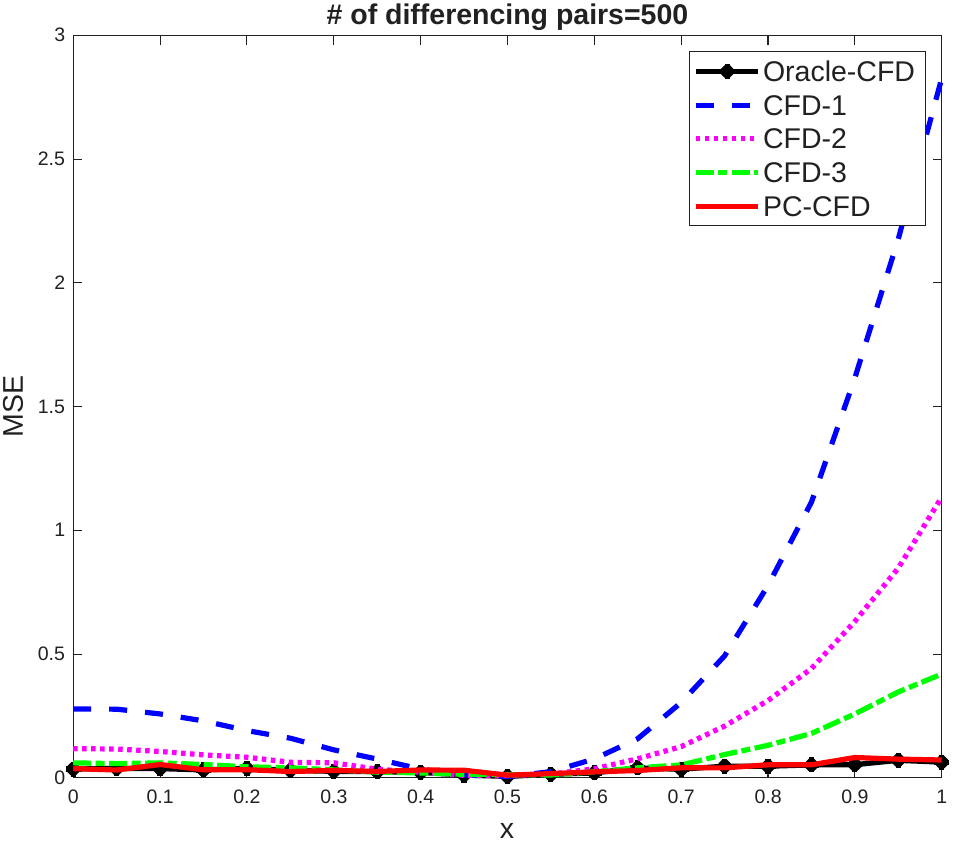}
	\end{subfigure}
	\caption{Empirical MSEs of PC-CFD, Oracle-CFD, and the fixed-tuning CFD estimators for Function 1.}
	\label{fig.1}
\end{figure}

\begin{figure}[!h]
	\begin{subfigure}[b]{0.42\textwidth}
		\centering
		\includegraphics[width=\textwidth]{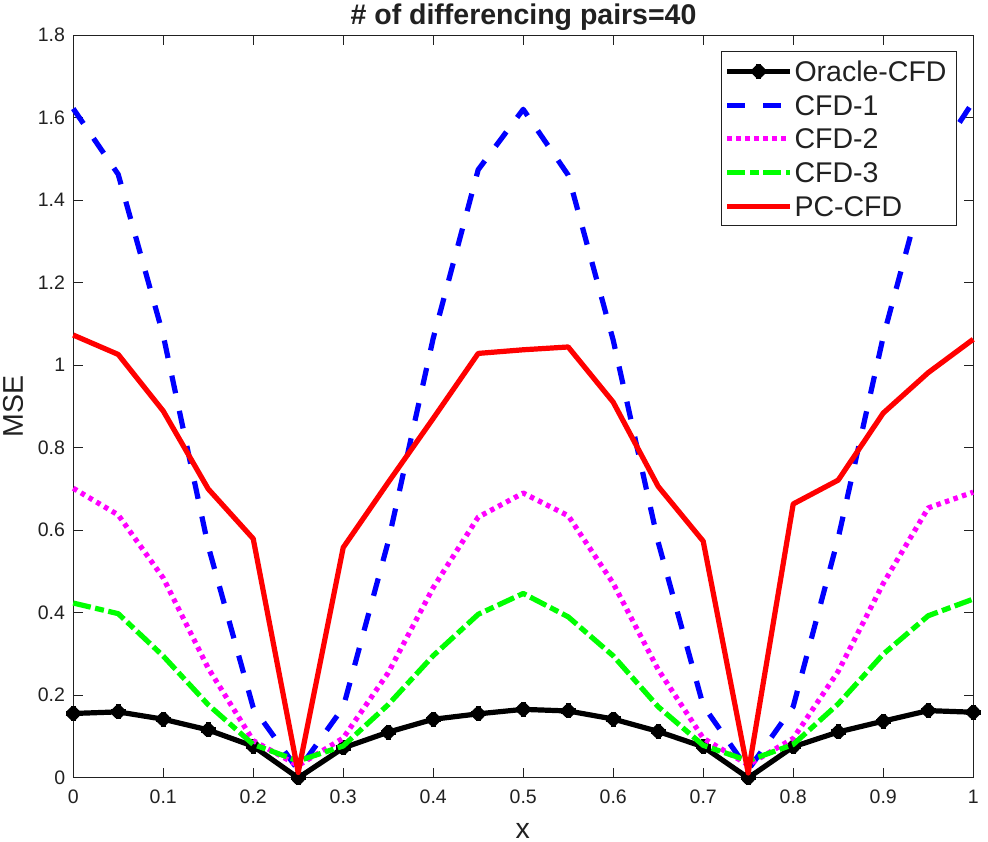}
	\end{subfigure}
	\hfill
	\begin{subfigure}[b]{0.42\textwidth}
		\centering
		\includegraphics[width=\textwidth]{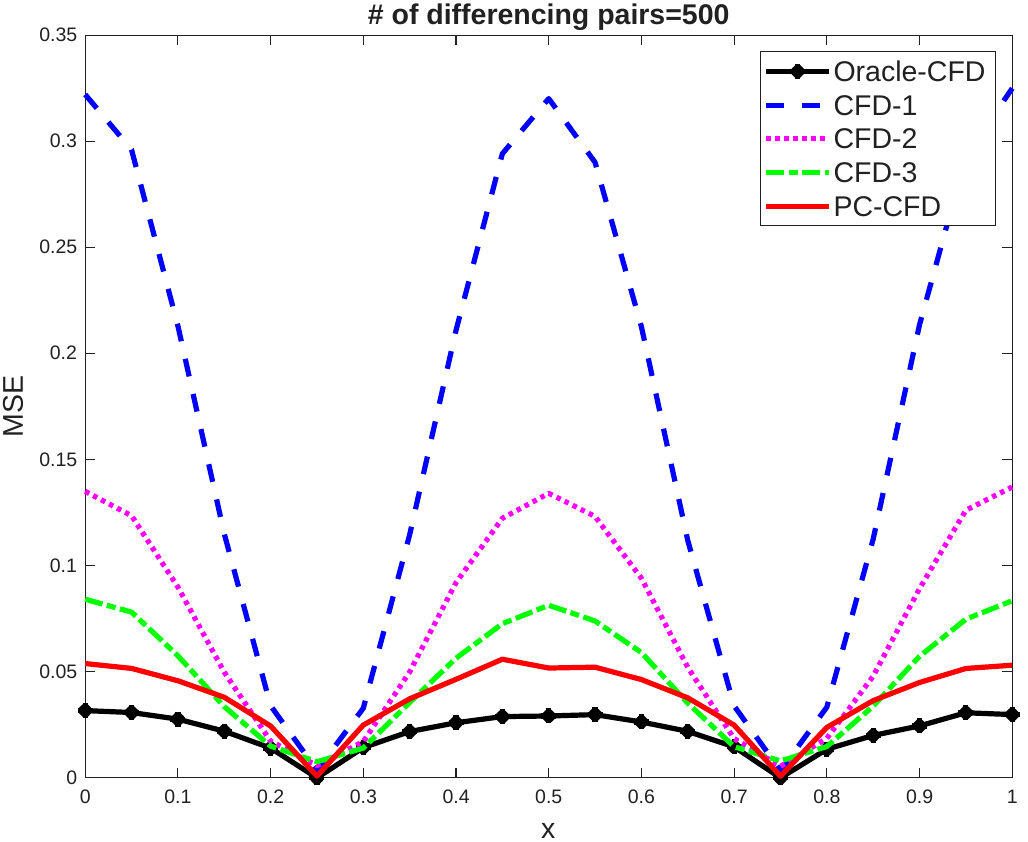}
	\end{subfigure}
	\caption{Empirical MSEs of PC-CFD, Oracle-CFD, and the fixed-tuning CFD estimators for Function 2.}
	\label{fig.2}
\end{figure}

\paragraph{Function 3.}
Consider \( f(x)=1+x\sin(x^2) \) as in \citet{de2011nonparametric}. In this case,
$B=\cos(x^{2})-4x^{2}\sin(x^{2})-3x^{4}\cos(x^{2})/2$. The fixed-tuning benchmarks use
\(
(\widetilde B_1,\widetilde B_2,\widetilde B_3)=(10,15,20)\),
and
\( \widetilde\sigma_k^2(x_0)=0.05.
\)
As shown in Figure~\ref{fig.3}, some prescribed choices are competitive over portions of the domain, but they are sensitive to changes in the oracle curvature. PC-CFD performs better than the fixed-tuning estimators when \(x>3\) and avoids the sharp deterioration exhibited by the conventional estimators near \(x=4\). When \(x\leq3\), PC-CFD remains competitive even when a prescribed choice happens to be locally favorable. Like in the previous examples, this pattern illustrates the robustness of PC-CFD in that it sacrifices little in favorable instances while protecting against large losses in unfavorable ones.

\paragraph{Function 4.}
Consider \(f(x)=x+2e^{-16x^2}\) as in \citet{mahmoud2020robust}. Here, $B=(1024x-32768x^{3}/3)e^{-16x^{2}}$. The fixed-tuning benchmarks use
\(
(\widetilde B_1,\widetilde B_2,\widetilde B_3)=(3,6,9)\),
and
\( \widetilde\sigma_k^2(x_0)=0.05.
\)
Figure~\ref{fig.4} highlights the robustness of PC-CFD across heterogeneous problem instances. Although a prescribed perturbation constant may perform well over part of the considered range, its performance deteriorates substantially when the underlying curvature changes, and there is no prior information indicating which fixed choice will be appropriate. By adapting the perturbation size through pilot calibration, PC-CFD maintains a more stable MSE profile over \(x\in[0,0.8]\). In particular, its worst-case MSE is approximately one half of that of CFD-3 and one tenth of that of CFD-1. These results show that PC-CFD provides substantially stronger protection against the performance instability caused by instance-specific variation.

\paragraph{Function 5.}
Consider \( f(x)=\left[\sin(2x^3)\right]^3 \) as in \citet{mahmoud2020robust}. The corresponding curvature coefficient is
$B=6(\sin(2x^{3}))^{2}\cos(2x^{3})+216x^{3}\sin(2x^{3})(\cos(2x^{3}))^{2}-108x^{3}(\sin(2x^{3}))^{3}+216x^{6}(\cos(2x^{3}))^{3}-756x^{6}(\sin(2x^{3}))^{2}\cos(2x^{3})$. The fixed-tuning benchmarks use
\(
(\widetilde B_1,\widetilde B_2,\widetilde B_3)=(10,15,20)\),
and
\( \widetilde\sigma_k^2(x_0)=0.05.
\)
The results in Figure~\ref{fig.5} display a pattern similar to that in Function~3. PC-CFD performs better than the fixed-tuning estimators when \(x>0.8\), where the prescribed curvature values become poorly matched to the instance. For \(x\leq0.8\), its MSE remains controlled even when a fixed-tuning estimator is locally competitive. In contrast, the conventional estimators can incur a pronounced loss near \(x=1\). The experiment again demonstrates that PC-CFD is more robust across curvature regimes.

\begin{figure}[!h]
	\centering
	\begin{subfigure}[b]{0.42\textwidth}
		\centering
		\includegraphics[width=\textwidth]{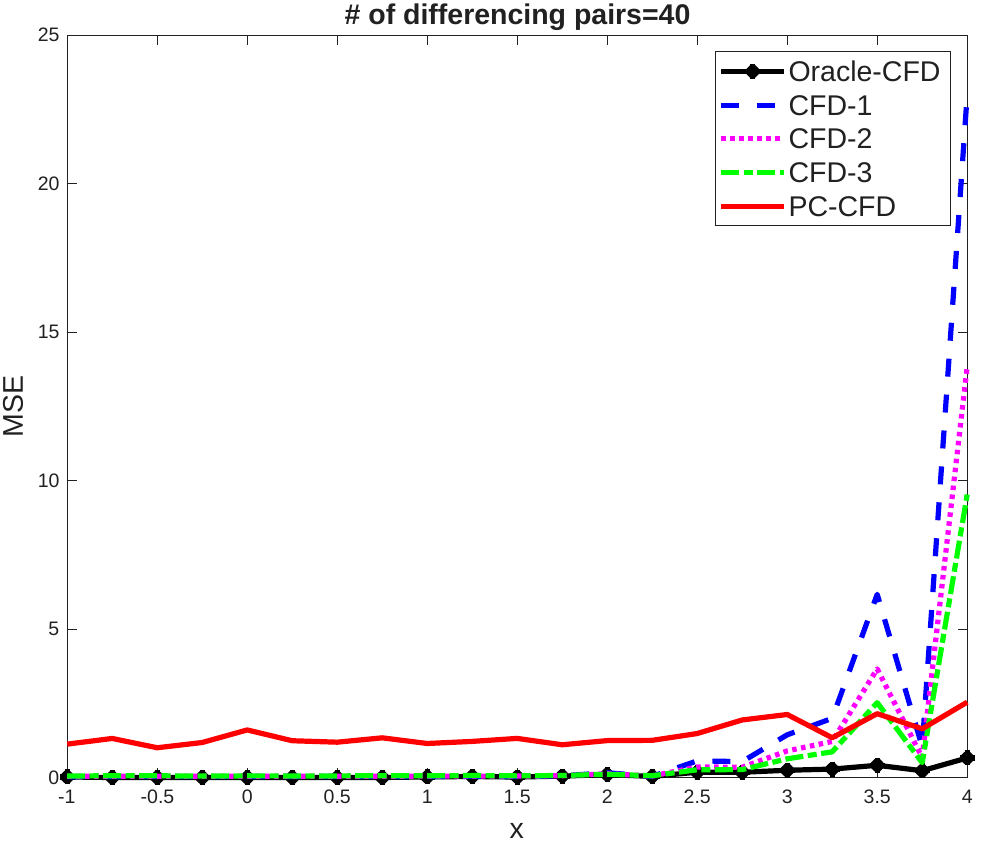}
	\end{subfigure}
	\hfill
	\begin{subfigure}[b]{0.42\textwidth}
		\centering
		\includegraphics[width=\textwidth]{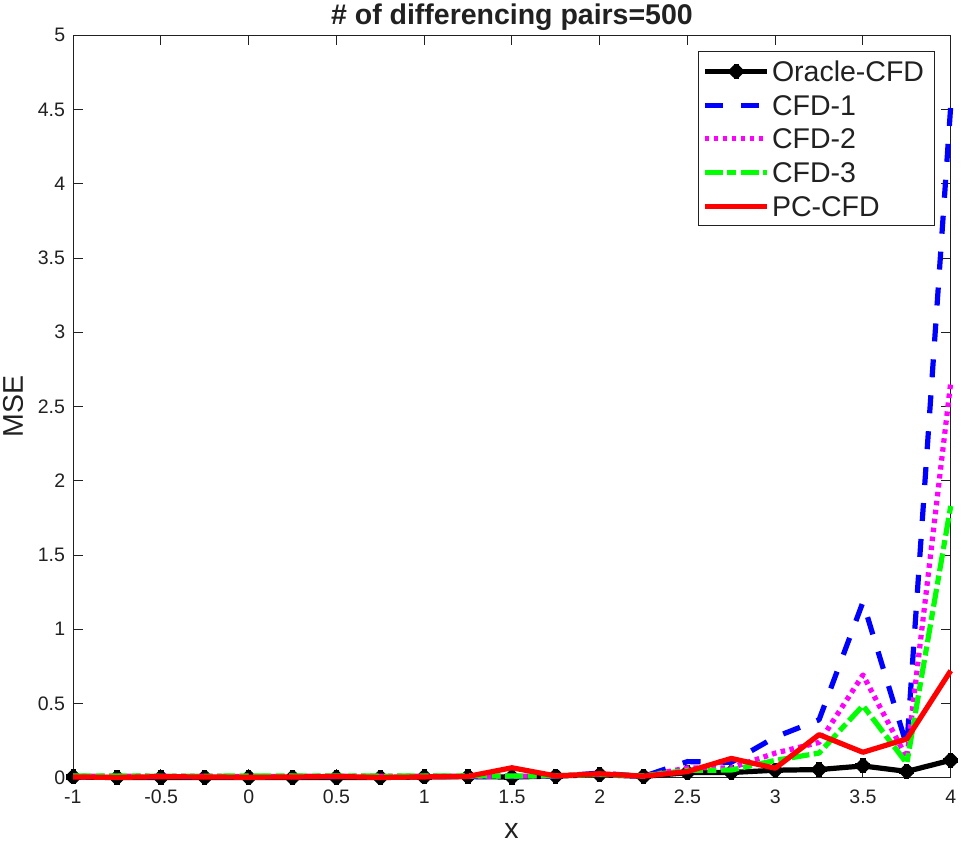}
	\end{subfigure}
	\caption{Empirical MSEs of PC-CFD, Oracle-CFD, and the fixed-tuning CFD estimators for Function 3.}
	\label{fig.3}
\end{figure}

\begin{figure}[!h]
	\centering
	\begin{subfigure}[b]{0.42\textwidth}
		\centering
		\includegraphics[width=\textwidth]{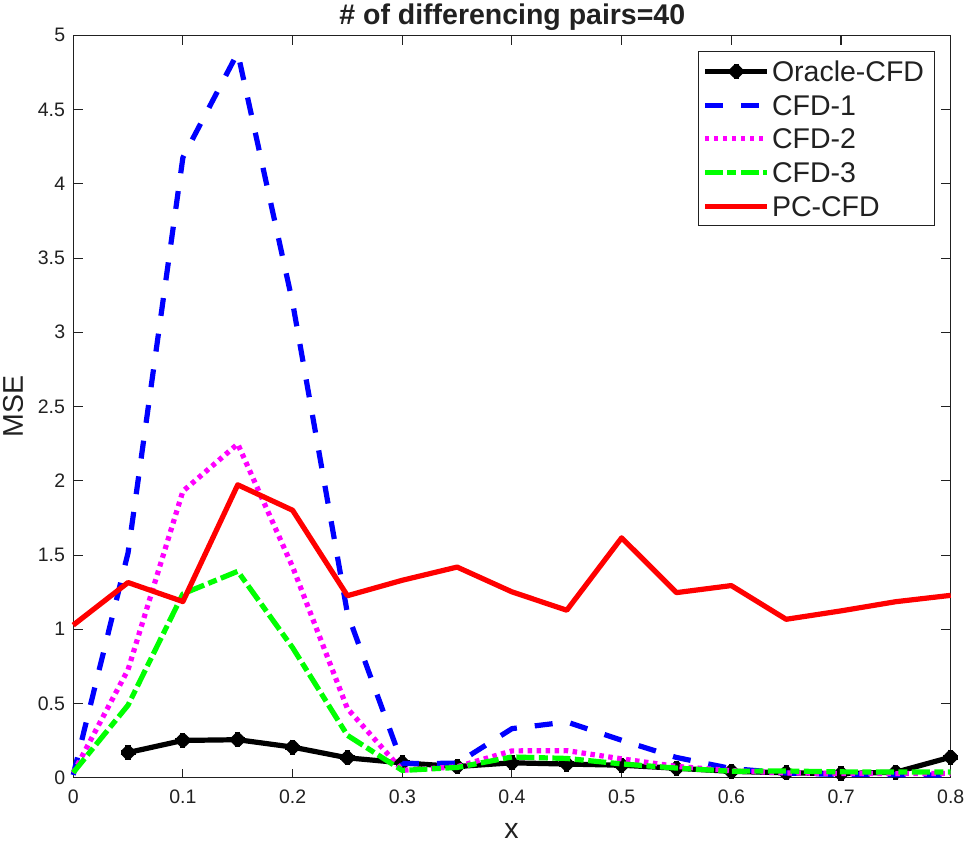}
	\end{subfigure}
	\hfill
	\begin{subfigure}[b]{0.42\textwidth}
		\centering
		\includegraphics[width=\textwidth]{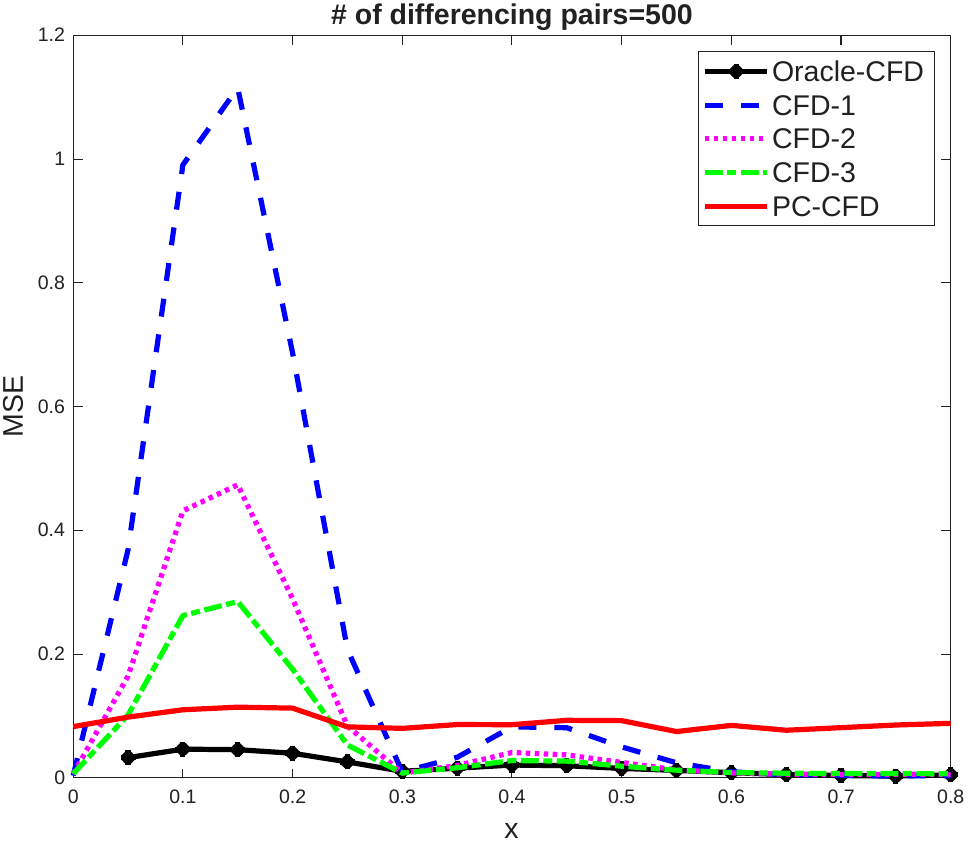}
	\end{subfigure}
	\caption{Empirical MSEs of PC-CFD, Oracle-CFD, and the fixed-tuning CFD estimators for Function 4.}
	\label{fig.4}
\end{figure}

\begin{figure}[!h]
	\centering
	\begin{subfigure}[b]{0.42\textwidth}
		\centering
		\includegraphics[width=\textwidth]{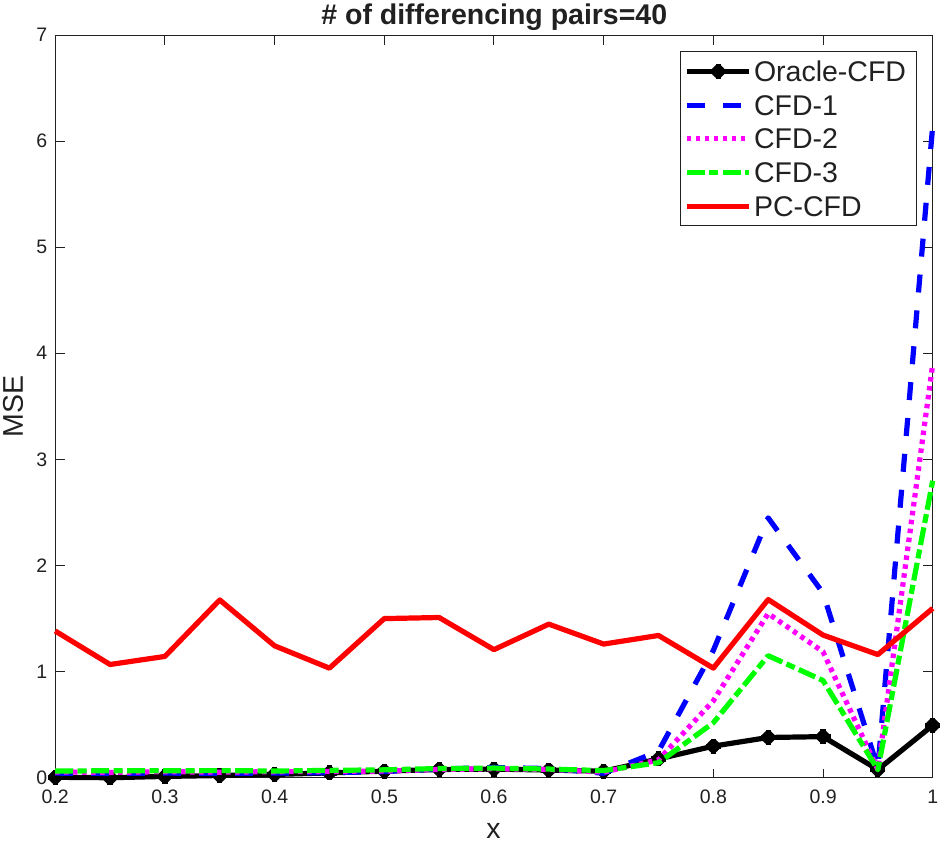}
	\end{subfigure}
	\hfill
	\begin{subfigure}[b]{0.42\textwidth}
		\centering
		\includegraphics[width=\textwidth]{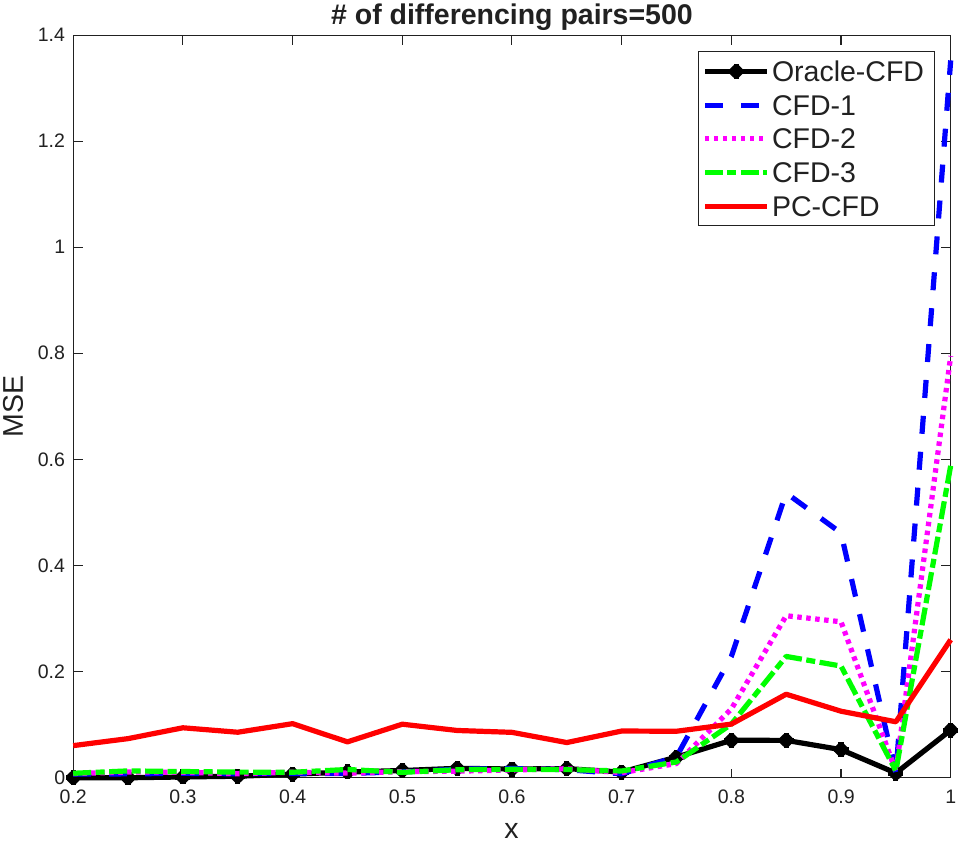}
	\end{subfigure}
	\caption{Empirical MSEs of PC-CFD, Oracle-CFD, and the fixed-tuning CFD estimators for Function 5.}
	\label{fig.5}
\end{figure}

\paragraph{Function 6.}
Consider \( f(x)=\left[1+e^{-20(x-0.5)}\right]^{-1} \)
as in \citet{mahmoud2020robust}. In this case, $B=\frac{4000}{3}\left(\frac{6e^{-60(x-0.5)}}{(1+e^{-20(x-0.5)})^{4}}-\frac{4e^{-40(x-0.5)}}{(1+e^{-20(x-0.5)})^{3}}-\frac{2e^{-20(x-0.5)}}{(1+e^{-20(x-0.5)})^{3}}+\frac{e^{-20(x-0.5)}}{(1+e^{-20(x-0.5)})^{2}}\right)$. The fixed-tuning benchmarks use
\(
(\widetilde B_1,\widetilde B_2,\widetilde B_3)=(10,15,20)\),
and
\( \widetilde\sigma_k^2(x_0)=0.05.
\)
Like Figure~\ref{fig.2}, Figure~\ref{fig.6} illustrates the robustness of PC-CFD but also the small-budget effect. At \(n=40\), even though CFD-1 and CFD-2 do not outperform PC-CFD, CFD-3 can do so. In this case, the pilot budget of PC-CFD is too limited to estimate the tuning quantities accurate enough to dominate CFD-3. Once the total budget increases, the advantage of adaptation becomes apparent. At \(n=500\), over \(x\in[0,0.9]\), the worst-case MSE of PC-CFD is less than one half of that of CFD-3 and one fourth of that of CFD-1. Note that, like our discussion for Figure~\ref{fig.2}, even though PC-CFD does not dominate every fixed-tuning benchmark in the small-sample case, we do not know a priori if a chosen fixed-tuning is ``good", and so PC-CFD still appears a safer choice than using fixed-tuning even in small sample.

\begin{figure}[!h]
	\centering
	\begin{subfigure}[b]{0.42\textwidth}
		\centering
		\includegraphics[width=\textwidth]{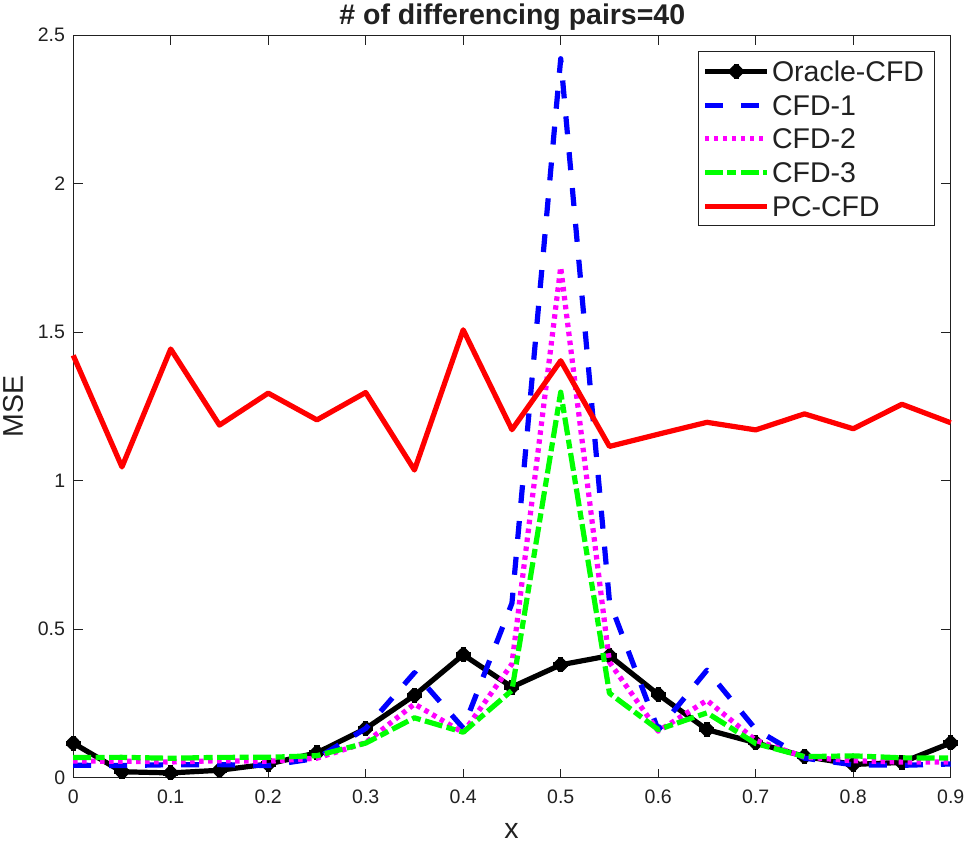}
	\end{subfigure}
	\hfill
	\begin{subfigure}[b]{0.42\textwidth}
		\centering
		\includegraphics[width=\textwidth]{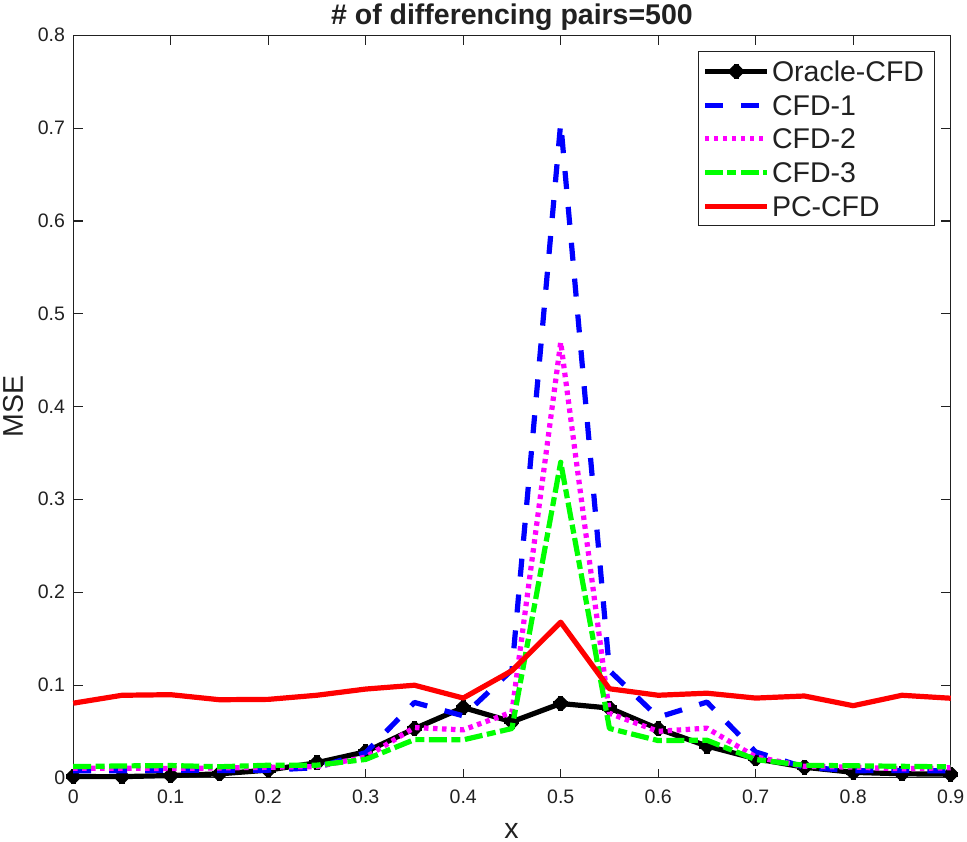}
	\end{subfigure}
	\caption{Empirical MSEs of PC-CFD, Oracle-CFD, and the fixed-tuning CFD estimators for Function 6.}
	\label{fig.6}
\end{figure}

Taken together, the numerical results in Figures~\ref{fig.1}--\ref{fig.6} consistently support the two central claims of the paper. First, pilot calibration improves robustness: PC-CFD avoids the large instance-dependent deterioration that can arise from a prescribed perturbation constant and performs reliably across functions and evaluation points. Second, PC-CFD approaches the performance of the oracle benchmark as the budget increases. The results also clarify the role of small samples. Under a very limited budget, because of the noise of pilot estimation, sometimes PC-CFD cannot uniformly dominate a particular fixed-tuning estimator. However, a priori we would not know if this favorable fixed-tuning estimator is selected, and thus even in this situation, PC-CFD appears still a safe choice. At the same time, when the budget increases, PC-CFD is consistently more robust than all considered fixed-tuning estimators. Overall, the experiments show that pilot calibration provides a practically robust alternative to ad hoc finite-difference tuning, and moreover achieves near-oracle efficiency.

\subsubsection{\(M/M/1\) Queueing System}
We next consider an \(M/M/1\) queueing system. The simulation output is the total system time of the first 10 customers, and the quantity of interest is the derivative of its expectation with respect to the arrival rate. The oracle derivative is \(0.946\) when both the arrival and service rates are equal to \(4\). Unlike the analytical functions above, the curvature coefficient \(B\) and the local variance \(\sigma^2(x_0)\) are not available to create the oracle benchmark. This experiment therefore focuses solely on assessing robustness relative to fixed-tuning CFD estimators.

The fixed-tuning estimators use
\(
(\widetilde B,\widetilde\sigma^2(x_0))
\in
\{(5,1),(10,1),(20,1)\}
\)
for CFD-1, CFD-2, and CFD-3, respectively. Figure~\ref{fig.7} reports the results. At \(n=1000\), PC-CFD attains an MSE of \(0.13\), whereas the reported MSEs of the fixed-tuning estimators are \(0.36\) or larger. Thus, relative to an MSE of \(0.36\), pilot calibration reduces the MSE by approximately \(64\%\). The experiment is more challenging than the six analytical examples because the simulation variance changes with the input. The favorable performance therefore indicates that, as our theory entails, the robustness of PC-CFD extends beyond additive homoscedastic noise to a stochastic simulation system with input-dependent variance.

\begin{figure}[!h]
	\centering
	\includegraphics[width=0.42\textwidth]{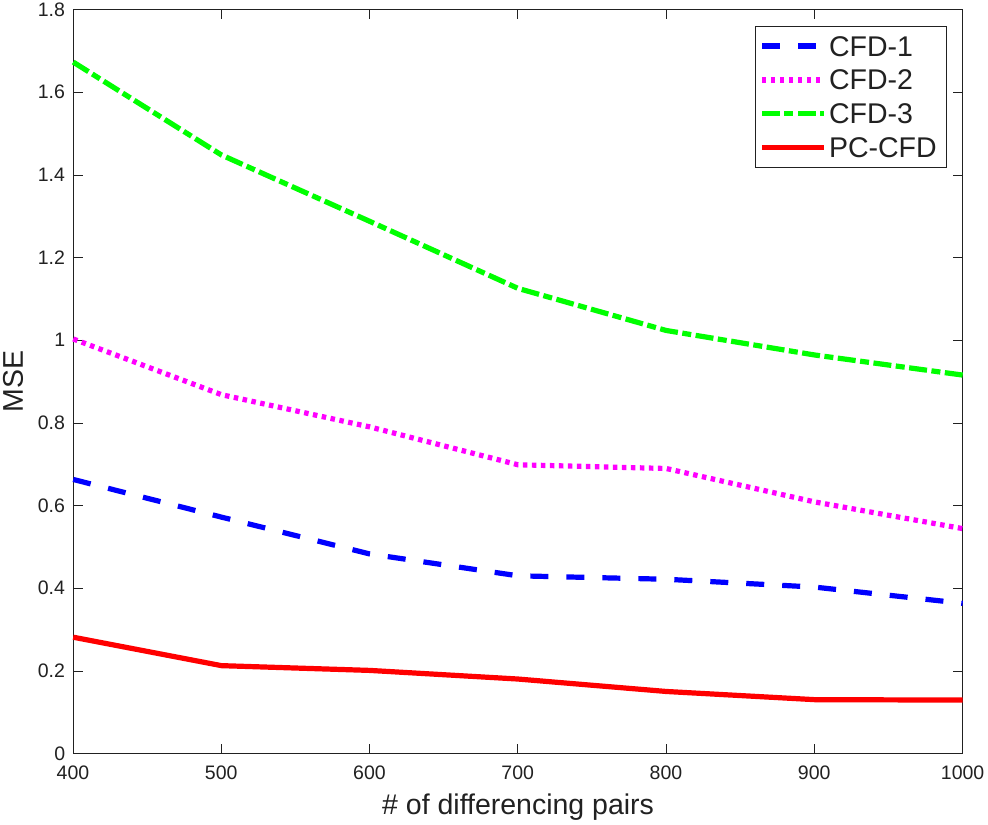}
	\caption{Empirical MSEs of PC-CFD and the fixed-tuning CFD estimators
    for the \(M/M/1\) queueing system.}
	\label{fig.7}
\end{figure}

\subsection{Multi-Dimensional Experiments}
We next evaluate the multi-dimensional extensions developed in Section~\ref{sec:multi}. The proposed estimators are coordinate-wise PC-CFD, which calibrates a separate coordinate-wise perturbation rule, and simultaneous-perturbation PC-CFD, which uses randomized directions and estimates the curvature aggregates \(G_1,\ldots,G_p\). The comparison is intended to assess whether pilot calibration remains effective in multi-dimensional stochastic systems and whether sharing function evaluations across gradient components is beneficial in the tested instances.

\subsubsection{\(M/M/p\) Queueing Systems}
We consider \(M/M/p\) queueing systems with heterogeneous server rates. The simulation output is the total system time of the first 50 customers, and the target is the gradient of its expectation with respect to the vector of service rates. The system configurations and reference gradients are reported in Table~\ref{tab:1}.
\begin{table}[!h]
    \centering
    \small
    \begin{tabular}{|c|c|c|c|}
        \hline
        \(p\) & Arrival rate \(\lambda\)
        & Service rates \(\bm\mu=(\mu_1,\ldots,\mu_p)\)
        & Reference gradient \\
        \hline
        2 & 10 & \((5,9)\) & \((-2.140,-1.954)\) \\
        \hline
        3 & 10 & \((2,5,8)\) & \((-2.258,-1.580,-1.426)\) \\
        \hline
        4 & 10 & \((1,3,5,7)\) & \((-3.152,-1.593,-1.212,-1.098)\) \\
        \hline
        5 & 10 & \((1,2,3,4,5)\)
        & \((-3.641,-2.276,-1.894,-1.614,-1.397)\) \\
        \hline
    \end{tabular}
    \caption{Configurations and reference gradients for the \(M/M/p\)
    queueing experiments.}
    \label{tab:1}
\end{table}

Because the oracle curvature and variance quantities are unavailable, we focus on comparing coordinate-wise PC-CFD and simultaneous-perturbation PC-CFD with fixed-tuning coordinate-wise CFD estimators. The latter use $\delta=\left(\widetilde{\sigma}^{2}(x_{0})/(4\widetilde{B}^{2})\right)(n/p)^{-1/6}$
with
\(
\left(\widetilde{B},\widetilde{\sigma}^{2}(x_{0})\right)
\in
\{(5,1),(10,1),(20,1)\}
\)
for Coordinate-Wise CFD-1, Coordinate-Wise CFD-2, and Coordinate-Wise CFD-3, respectively. Figure~\ref{fig.8} shows that coordinate-wise PC-CFD consistently improves on the fixed-tuning coordinate-wise estimators in the tested systems. At \(n=5000\), its MSE is more than \(50\%\) lower than those of the reported fixed-tuning benchmarks. This cross-model stability is the multi-dimensional analogue of the robustness observed in the single-dimensional experiments.

\begin{figure}[!h]
	\centering
	\begin{subfigure}[b]{0.42\textwidth}
		\centering
		\includegraphics[width=\textwidth]{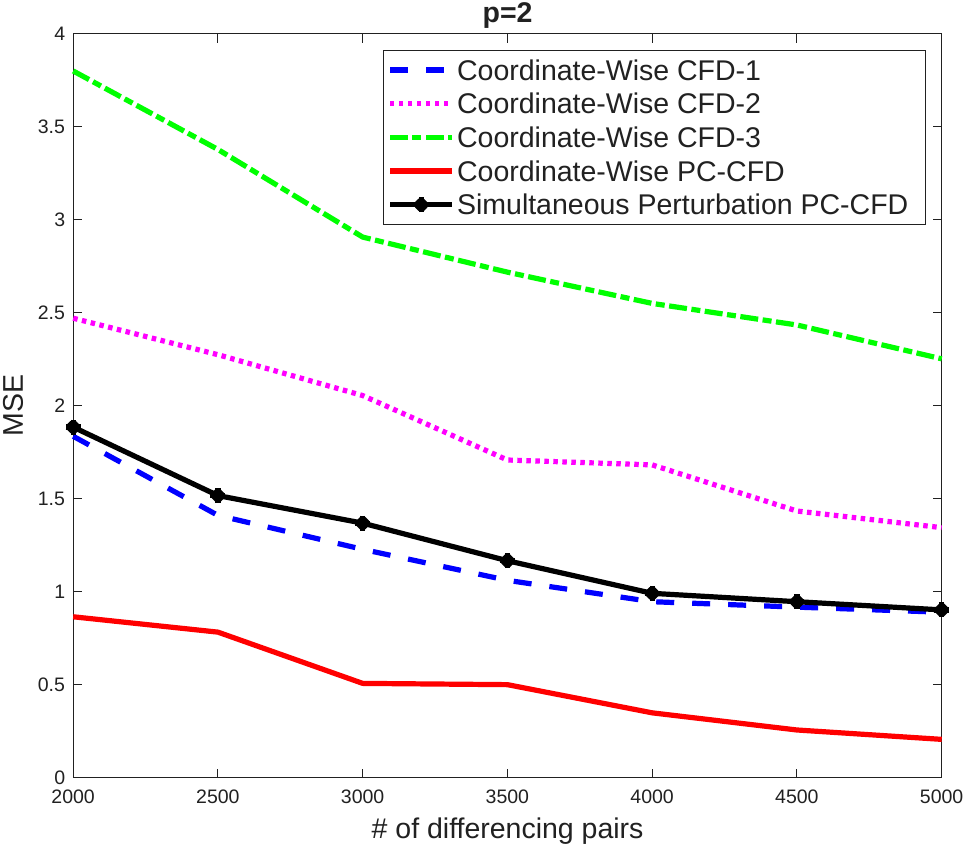}
	\end{subfigure}
	\hfill
	\begin{subfigure}[b]{0.42\textwidth}
		\centering
		\includegraphics[width=\textwidth]{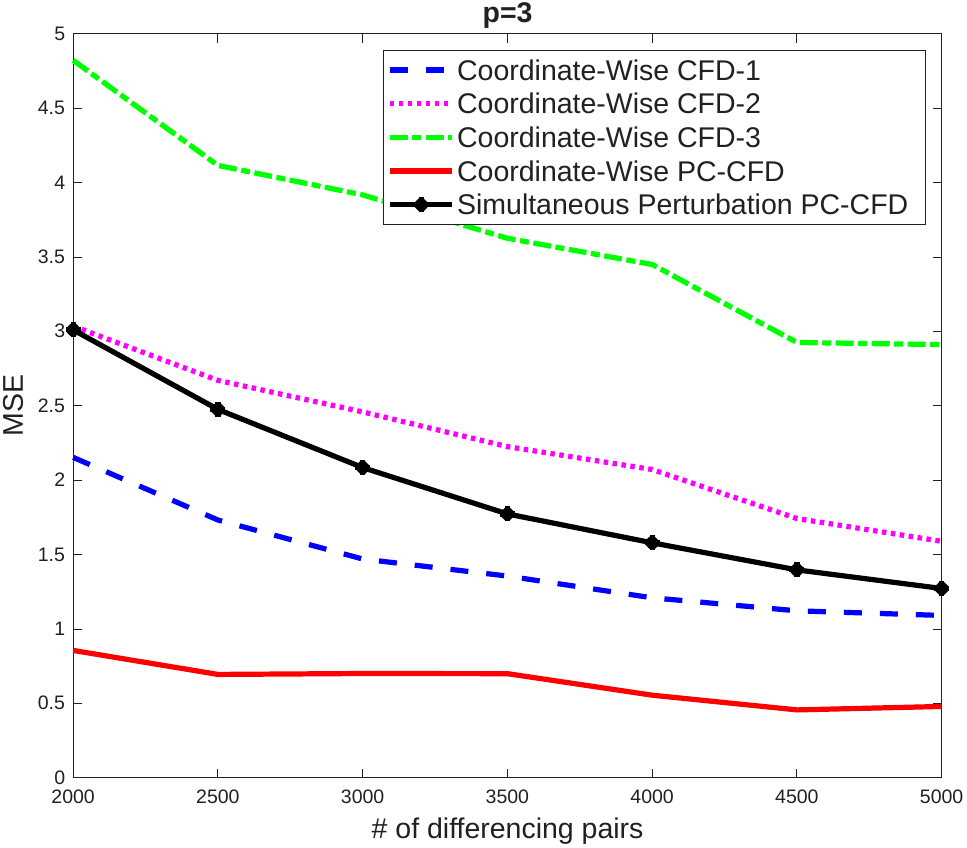}
	\end{subfigure}
	\hfill
	\begin{subfigure}[b]{0.42\textwidth}
		\centering
		\includegraphics[width=\textwidth]{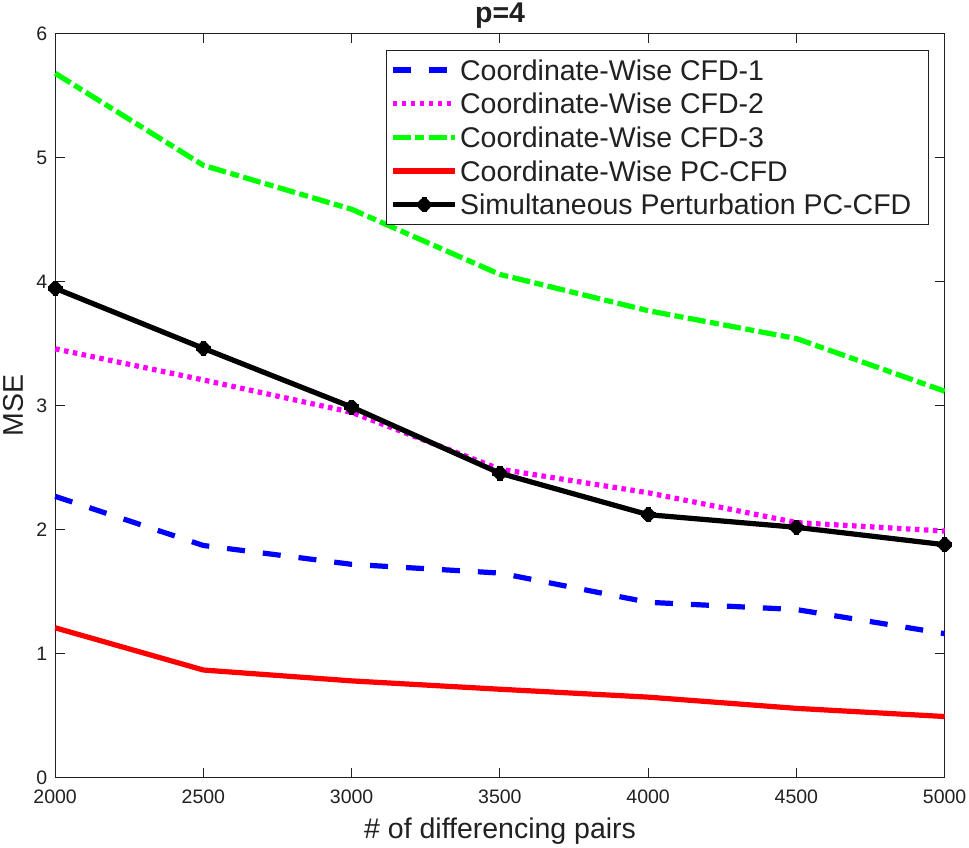}
	\end{subfigure}
	\hfill
	\begin{subfigure}[b]{0.42\textwidth}
		\centering
		\includegraphics[width=\textwidth]{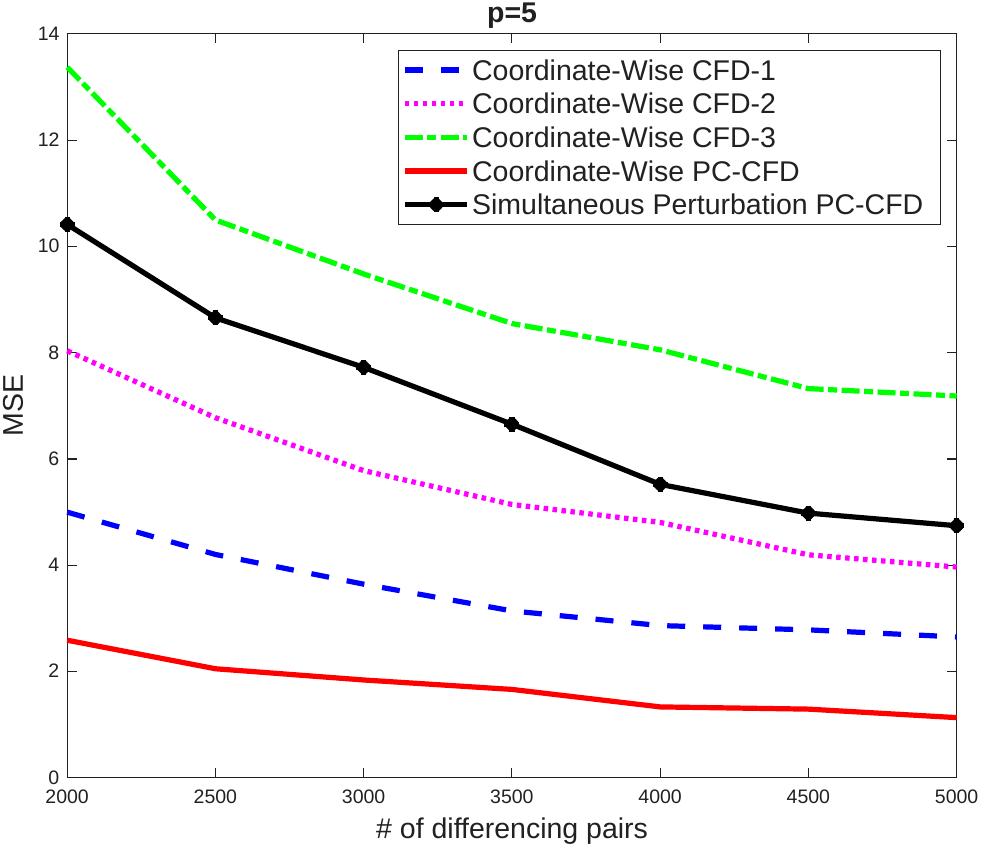}
	\end{subfigure}
	\caption{Empirical MSEs of all tested estimators for M/M/p queueing systems.}
	\label{fig.8}
\end{figure}

Coordinate-wise PC-CFD also performs better than simultaneous-perturbation PC-CFD in these particular queueing systems. Nonetheless, in all tested cases, simultaneous-perturbation PC-CFD still outperforms at least one fixed-tuning coordinate-wise estimator. To this end, since a priori we do not know if we have chosen a good fixed-tuning estimator, simultaneous-perturbation PC-CFD is arguably still safer than using fixed-tuning. Moreover, the relative performance between coordinate-wise and simultaneous-perturbation PC-CFD in this example should not be interpreted as universal. As shown in Section~\ref{sec:multi}, the oracle first-order MSE comparison is governed by \(B_{\mathrm{cw}}\) and \(B_{\mathrm{sp}}\). The observed results are
consistent with a derivative structure that favors the coordinate-wise curvature constant. As \(p\) increases, the MSE of simultaneous-perturbation PC-CFD deteriorates moderately, which may reflect the additional variability introduced by randomized directions and the combined pure and mixed third-order effects in \(G_1,\ldots,G_p\).

\subsubsection{Deep-Neural-Network Gradient Checking}
We finally consider gradient checking in a deep neural network trained for image classification on the MNIST dataset. The network has an input layer with \(784\) units corresponding to the \(28\times28\) image pixels, three hidden layers with \(256\), \(128\), and \(64\) units, and an output layer with \(10\) units. ReLU activations are used in the hidden layers, and a softmax activation is used in the output layer. The gradients generated by backpropagation are treated as the reference gradients.

Gradient checking compares the reference gradients with numerical finite-difference estimates. We declare the check successful when the relative gradient-estimation error is below \(10^{-2}\). Since the curvature and variance quantities are unavailable, we compare coordinate-wise PC-CFD and simultaneous-perturbation PC-CFD with the same three fixed-tuning specifications, \(
\left(\widetilde{B},\widetilde{\sigma}^{2}(x_{0})\right)
\in
\{(5,1),(10,1),(20,1)\}
\), used in the \(M/M/p\) experiments.

Table~\ref{tab:2} reports the number of differencing pairs required to satisfy the prescribed accuracy criterion. The three fixed-tuning estimators require between \(12{,}874\) and \(13{,}486\) pairs, whereas
coordinate-wise PC-CFD and simultaneous-perturbation PC-CFD require \(6{,}146\) and \(5{,}378\) pairs, respectively. Relative to the best reported fixed-tuning result, these values correspond to reductions of
approximately \(52\%\) and \(58\%\). Thus, both pilot-calibrated methods substantially reduce the simulation effort required to reach the target gradient-checking accuracy. In this example, sharing each function-value
pair across all gradient components makes simultaneous perturbation particularly effective.

\begin{table}[!h]
    \centering
    \small
    \begin{tabular}{cccccc}
        \hline
        & \multicolumn{3}{c}{Coordinate-Wise} 
        & Coordinate-Wise & Simultaneous Perturbation \\
        & CFD-1 & CFD-2 & CFD-3 & PC-CFD &  PC-CFD\\
        \hline
        \multirow{2}{*}{\# of differencing pairs}
        & \multirow{2}{*}{12,874} & \multirow{2}{*}{13,486} & \multirow{2}{*}{13,032} & \multirow{2}{*}{6,146} & \multirow{2}{*}{5,378} \\
        \\
        \hline
    \end{tabular}
    \caption{\# of differencing pairs required in the neural-network
    gradient-checking experiment.}
    \label{tab:2}
\end{table}

\section{Code and Data Disclosure}\label{sec:Code and Data Disclosure}
The code and data to support the numerical experiments in this paper can be found at \url{https://anonymous.4open.science/r/PC-CFD/}.

\section*{Acknowledgment}
Li's work was supported by grants from the National Natural Science Foundation of China (72595841, 72595840, 72192843), a grant from MOE Social Science Laboratory of Digital Economic Forecasts and Policy Simulation at UCAS, and MOE Social Sciences Innovative Group on Complex Systems Modeling in Economic Management in the Era of Digital Intelligence, University of Chinese Academy of Sciences. Lam's work was supported by the Columbia Dream Sports AI Innovation Award, the InnoHK Initiative, the Government of the HKSAR, and Laboratory for AI-Powered Financial Technologies. Peng's work was supported by
the National Natural Science Foundation of China under Grant 72325007, the Science and Technology Innovation Program of Hunan Province under Grant 2024RC7003, and Xiangjiang Laboratory under Grant 24XJJCYJ01001.

\bibliographystyle{unsrtnat}  
\bibliography{references}  

\section*{Proofs and Supplementary Materials}
\subsection*{Proof of Proposition~\ref{pro_general}.}
\proof
Note that
\begin{eqnarray*}
	\mathcal{R}_{G}^{*}
	&=&\lim_{n\to+\infty}\underset{\mathcal{P}\in\Omega_{\mathcal{P}},w_{i}>0}{\inf}\ \underset{\hat{f}(\cdot)\in\Omega_{\hat f}}{\sup}\mathbb{E}_{\mathcal{P}}\left[(\sum_{i=1}^{n}w_{i}-1)f^{(1)}(x_{0})n^{1/3}+B\sum_{i=1}^{n}w_{i}\alpha_{i}^{2}+\right.\\
	&&\left.\sum_{i=1}^{n}w_{i}R(\delta_{i})\alpha_{i}^{4}n^{-1/3}\right]^{2}+\mathbb{E}_{\mathcal{P}}\left[n\sum_{i=1}^{n}(w_{i}/2\alpha_{i})^{2}\eta^{2}(\delta_{i})\right].
\end{eqnarray*}
Consider the function $f(\cdot)$ with $f^{(1)}(x_{0})>0$, $f^{(3)}(x_{0})>0$, and $f^{(5)}(x)>0$, we have for any $\mathcal{P}\in\Omega_{\mathcal{P}}, \sum_{i=1}^{n}w_{i}>1, w_{i}>0$,
\begin{eqnarray*}
	&&\mathbb{E}_{\mathcal{P}}\left[(\sum_{i=1}^{n}w_{i}-1)f^{(1)}(x_{0})n^{1/3}+B\sum_{i=1}^{n}w_{i}\alpha_{i}^{2}+\sum_{i=1}^{n}w_{i}R(\delta_{i})\alpha_{i}^{4}n^{-1/3}\right]^{2}\\
	&\geq&\mathbb{E}_{\mathcal{P}}\left[(\sum_{i=1}^{n}w_{i}-1)f^{(1)}(x_{0})n^{1/3}\right]^{2}\\
	&=&\left[(\sum_{i=1}^{n}w_{i}-1)f^{(1)}(x_{0})n^{1/3}\right]^{2}.
\end{eqnarray*}
Therefore, we have
$$\lim_{n\to+\infty}\underset{\mathcal{P}\in\Omega_{\mathcal{P}},\sum_{i=1}^{n}w_{i}>1,w_{i}>0}{\inf}\ \underset{\hat{f}(\cdot)\in\Omega_{\hat f}}{\sup} n^{2/3}\mathbb{E}[(\widehat{\theta}_{G}-f^{(1)}(x_{0}))^{2}]=+\infty.$$
Similarly, we have the same result for $\mathcal{P}\in\Omega_{\mathcal{P}},\ \sum_{i=1}^{n}w_{i}<1, w_{i}>0$. 

When $\sum_{i=1}^{n}w_{i}=1, w_{i}>0$, we have for any $\hat{f}(\cdot)\in\Omega_{\hat f}$,
$$\mathbb{E}_{\mathcal{P}}\left[B\sum_{i=1}^{n}w_{i}\alpha_{i}^{2}\right]^{2}\leq B^{2}\mathbb{E}_{\mathcal{P}}\left[\sum_{i=1}^{n}w_{i}\alpha_{i}^{4}\right]=B^{2}\sum_{i=1}^{n}w_{i}\mathbb{E}_{\mathcal{P}}\left[\alpha_{i}^{4}\right]<+\infty,$$
and $$\mathbb{E}_{\mathcal{P}}\left[\sum_{i=1}^{n}w_{i}R(\delta_{i})\alpha_{i}^{4}n^{-1/3}\right]^{2}=O(n^{-2/3})<+\infty,$$ which is due to $\mathbb{E}[\alpha_{i}^{8}]<+\infty$. Especially, when $w_{i}=1/n,\ i=1,\ldots,n$,
$$\mathbb{E}_{\mathcal{P}}\left[n\sum_{i=1}^{n}(w_{i}/2\alpha_{i})^{2}\eta^{2}(\delta_{i})\right]\leq \frac{1}{4}\left[\max_{i}\eta^{2}(\delta_{i})\right]\frac{1}{n}\sum_{i=1}^{n}\mathbb{E}_{\mathcal{P}}\left[1/\alpha_{i}^{2}\right]<+\infty.$$
Therefore, we have
$$\lim_{n\to+\infty}\underset{\mathcal{P}\in\Omega_{\mathcal{P}},\sum_{i=1}^{n}w_{i}=1,w_{i}>0}{\inf}\ \underset{\hat{f}(\cdot)\in\Omega_{\hat f}}{\sup} n^{2/3}\mathbb{E}[(\widehat{\theta}_{G}-f^{(1)}(x_{0}))^{2}]<+\infty,$$
which means that
\begin{eqnarray*}
	&&\mathcal{R}_{G}^{*}\\
	&=&\lim_{n\to+\infty}\underset{\mathcal{P}\in\Omega_{\mathcal{P}},\sum_{i=1}^{n}w_{i}=1,w_{i}>0}{\inf}\ \underset{\hat{f}(\cdot)\in\Omega_{\hat f}}{\sup} n^{2/3}\mathbb{E}[(\widehat{\theta}_{G}-f^{(1)}(x_{0}))^{2}]\\
	&=&\lim_{n\to+\infty}\underset{\mathcal{P}\in\Omega_{\mathcal{P}},\sum_{i=1}^{n}w_{i}=1,w_{i}>0}{\inf}\ \underset{\hat{f}(\cdot)\in\Omega_{\hat f}}{\sup}\mathbb{E}_{\mathcal{P}}\left[B\sum_{i=1}^{n}w_{i}\alpha_{i}^{2}+\sum_{i=1}^{n}w_{i}R(\delta_{i})\alpha_{i}^{4}n^{-1/3}\right]^{2}+\mathbb{E}_{\mathcal{P}}\left[n\sum_{i=1}^{n}(w_{i}/2\alpha_{i})^{2}\eta^{2}(\delta_{i})\right].
\end{eqnarray*}

Similarly, we have $$\mathcal{R}_{0}\triangleq\lim_{n\to+\infty}\underset{\mathcal{P}\in\Omega_{\mathcal{P}},\sum_{i=1}^{n}w_{i}=1,w_{i}>0}{\inf}\ \underset{\hat{f}(\cdot)\in\Omega_{\hat f}}{\sup}\mathbb{E}_{\mathcal{P}}\left[B\sum_{i=1}^{n}w_{i}\alpha_{i}^{2}\right]^{2}+\mathbb{E}_{\mathcal{P}}\left[n\sum_{i=1}^{n}(w_{i}/2\alpha_{i})^{2}2\sigma^{2}(x_{0})\right]<+\infty.$$ Note that $\eta^{2}(\delta_{i})=2\sigma^{2}(x_{0})+\Theta(\delta^{s})$, then there exists positive $(N_{0},C_{1},C_{2},C_{3},C_{4})$ such that $\forall n>N_{0}$, the following two inequalities hold:
\begin{eqnarray*}
	&&\mathbb{E}_{\mathcal{P}}\left[B\sum_{i=1}^{n}w_{i}\alpha_{i}^{2}+\sum_{i=1}^{n}w_{i}R(\delta_{i})\alpha_{i}^{4}n^{-1/3}\right]^{2}+\mathbb{E}_{\mathcal{P}}\left[n\sum_{i=1}^{n}(w_{i}/2\alpha_{i})^{2}\eta^{2}(\delta_{i})\right]\\
	&\geq&\mathbb{E}_{\mathcal{P}}\left[B\sum_{i=1}^{n}w_{i}\alpha_{i}^{2}\right]^{2}-C_{1}n^{-1/3}+\mathbb{E}_{\mathcal{P}}\left[n\sum_{i=1}^{n}(w_{i}/2\alpha_{i})^{2}2\sigma^{2}(x_{0})\right]-C_{2}n^{-s/6},
\end{eqnarray*}
and for any $\mathcal{P}\in\Omega_{\mathcal{P}},\sum_{i=1}^{n}w_{i}=1,w_{i}>0$ such that $\mathbb{E}_{\mathcal{P}}\left[n\sum_{i=1}^{n}(w_{i}/2\alpha_{i})^{2}2\overline{\sigma}^2\right]<\mathcal{R}_{0}$,
\begin{eqnarray*}
	&&\mathbb{E}_{\mathcal{P}}\left[B\sum_{i=1}^{n}w_{i}\alpha_{i}^{2}+\sum_{i=1}^{n}w_{i}R(\delta_{i})\alpha_{i}^{4}n^{-1/3}\right]^{2}+\mathbb{E}_{\mathcal{P}}\left[n\sum_{i=1}^{n}(w_{i}/2\alpha_{i})^{2}\eta^{2}(\delta_{i})\right]\\
	&\leq&\mathbb{E}_{\mathcal{P}}\left[B\sum_{i=1}^{n}w_{i}\alpha_{i}^{2}\right]^{2}+C_{3}n^{-1/3}+\mathbb{E}_{\mathcal{P}}\left[n\sum_{i=1}^{n}(w_{i}/2\alpha_{i})^{2}2\sigma^{2}(x_{0})\right]+C_{4}n^{-s/6}.
\end{eqnarray*}
Therefore, we have $\forall n>N_{0}$,
\begin{eqnarray}
	&&\underset{\mathcal{P}\in\Omega_{\mathcal{P}},\sum_{i=1}^{n}w_{i}=1,w_{i}>0}{\inf}\ \underset{\hat{f}(\cdot)\in\Omega_{\hat f}}{\sup}\mathbb{E}_{\mathcal{P}}\left[B\sum_{i=1}^{n}w_{i}\alpha_{i}^{2}+\sum_{i=1}^{n}w_{i}R(\delta_{i})\alpha_{i}^{4}n^{-1/3}\right]^{2}+\mathbb{E}_{\mathcal{P}}\left[n\sum_{i=1}^{n}(w_{i}/2\alpha_{i})^{2}\eta^{2}(\delta_{i})\right]\nonumber\\
	&\geq&\underset{\mathcal{P}\in\Omega_{\mathcal{P}},\sum_{i=1}^{n}w_{i}=1,w_{i}>0}{\inf}\ \underset{\hat{f}(\cdot)\in\Omega_{\hat f}}{\sup}\left\{\mathbb{E}_{\mathcal{P}}\left[B\sum_{i=1}^{n}w_{i}\alpha_{i}^{2}\right]^{2}+\mathbb{E}_{\mathcal{P}}\left[n\sum_{i=1}^{n}(w_{i}/2\alpha_{i})^{2}2\sigma^{2}(x_{0})\right]\right\}-C_{1}n^{-1/3}-C_{2}n^{-s/6},\nonumber\\\label{eq_lb}
\end{eqnarray}
and
\begin{eqnarray}
	&&\underset{\mathcal{P}\in\Omega_{\mathcal{P}},\sum_{i=1}^{n}w_{i}=1,w_{i}>0}{\inf}\ \underset{\hat{f}(\cdot)\in\Omega_{\hat f}}{\sup}\mathbb{E}_{\mathcal{P}}\left[B\sum_{i=1}^{n}w_{i}\alpha_{i}^{2}+\sum_{i=1}^{n}w_{i}R(\delta_{i})\alpha_{i}^{4}n^{-1/3}\right]^{2}+\mathbb{E}_{\mathcal{P}}\left[n\sum_{i=1}^{n}(w_{i}/2\alpha_{i})^{2}\eta^{2}(\delta_{i})\right]\nonumber\\
	&\leq&\underset{\substack{\mathcal{P}\in\Omega_{\mathcal{P}},\sum_{i=1}^{n}w_{i}=1,w_{i}>0, \\ \mathbb{E}_{\mathcal{P}}\left[n\sum_{i=1}^{n}(w_{i}/2\alpha_{i})^{2}2\overline{\sigma}^2\right]<\mathcal{R}_{0}}}{\inf}\ \underset{\hat{f}(\cdot)\in\Omega_{\hat f}}{\sup} \mathbb{E}_{\mathcal{P}}\left[B\sum_{i=1}^{n}w_{i}\alpha_{i}^{2}+\sum_{i=1}^{n}w_{i}R(\delta_{i})\alpha_{i}^{4}n^{-1/3}\right]^{2}+\mathbb{E}_{\mathcal{P}}\left[n\sum_{i=1}^{n}(w_{i}/2\alpha_{i})^{2}\eta^{2}(\delta_{i})\right]\nonumber\\
	&\leq&\underset{\substack{\mathcal{P}\in\Omega_{\mathcal{P}},\sum_{i=1}^{n}w_{i}=1,w_{i}>0, \\ \mathbb{E}_{\mathcal{P}}\left[n\sum_{i=1}^{n}(w_{i}/2\alpha_{i})^{2}2\overline{\sigma}^2\right]<\mathcal{R}_{0}}}{\inf}\ \underset{\hat{f}(\cdot)\in\Omega_{\hat f}}{\sup} \left\{\mathbb{E}_{\mathcal{P}}\left[B\sum_{i=1}^{n}w_{i}\alpha_{i}^{2}\right]^{2}+\mathbb{E}_{\mathcal{P}}\left[n\sum_{i=1}^{n}(w_{i}/2\alpha_{i})^{2}2\sigma^{2}(x_{0})\right]\right\}+C_{3}n^{-1/3}+C_{4}n^{-s/6}.\nonumber\\\label{eq_ub}
\end{eqnarray}
In addition, since
$$\underset{\substack{\mathcal{P}\in\Omega_{\mathcal{P}},\sum_{i=1}^{n}w_{i}=1,w_{i}>0, \\ \mathbb{E}_{\mathcal{P}}\left[n\sum_{i=1}^{n}(w_{i}/2\alpha_{i})^{2}2\overline{\sigma}^2\right]\geq\mathcal{R}_{0}}}{\inf}\ \underset{\hat{f}(\cdot)\in\Omega_{\hat f}}{\sup} \mathbb{E}_{\mathcal{P}}\left[B\sum_{i=1}^{n}w_{i}\alpha_{i}^{2}\right]^{2}+\mathbb{E}_{\mathcal{P}}\left[n\sum_{i=1}^{n}(w_{i}/2\alpha_{i})^{2}2\sigma^{2}(x_{0})\right]>\mathcal{R}_{0},$$
then
$$\mathcal{R}_{0}=\lim_{n\to+\infty}\underset{\substack{\mathcal{P}\in\Omega_{\mathcal{P}},\sum_{i=1}^{n}w_{i}=1,w_{i}>0, \\ \mathbb{E}_{\mathcal{P}}\left[n\sum_{i=1}^{n}(w_{i}/2\alpha_{i})^{2}2\overline{\sigma}^2\right]<\mathcal{R}_{0}}}{\inf}\ \underset{\hat{f}(\cdot)\in\Omega_{\hat f}}{\sup} \mathbb{E}_{\mathcal{P}}\left[B\sum_{i=1}^{n}w_{i}\alpha_{i}^{2}\right]^{2}+\mathbb{E}_{\mathcal{P}}\left[n\sum_{i=1}^{n}(w_{i}/2\alpha_{i})^{2}2\sigma^{2}(x_{0})\right].$$

Taking limit of both sides of \eqref{eq_lb} and \eqref{eq_ub} as $n\to+\infty$, we obtain that
$$\mathcal{R}_{0}\leq\mathcal{R}_{\text{random}}^{*}\leq\mathcal{R}_{0},$$
which concludes the proposition.
\endproof

\subsection*{Proof of Theorem~\ref{thm:1}.}
\proof
Since $\Omega_{d}\subset\Omega_{\mathcal{P}}$, on the one hand, we have
\begin{gather*}
	\mathcal{R}_{G}^{*}\leq\mathcal{R}_{\text{deterministic}}^{*}.
\end{gather*}
On the other hand, suppose $(\mathcal{P}^{*},\{w_{i}^{*}\}_{i=1}^{n})$ is
minimax-optimal for~\eqref{worst-case MSE}. For each case of $\hat{f}(\cdot)\in\Omega_{\hat f}$,
\begin{eqnarray*}
	&&\left[B\sum_{i=1}^{n}w_{i}^{*}\mathbb{E}_{\mathcal{P}^{*}}[\alpha_{i}^{2}]\right]^{2}+\frac{1}{2}n\sigma^{2}(x_{0})\sum_{i=1}^{n}(w_{i}^{*})^{2}/\mathbb{E}_{\mathcal{P}^{*}}[\alpha_{i}^{2}]\\
	&\leq&\left[B\sum_{i=1}^{n}w_{i}^{*}\mathbb{E}_{\mathcal{P}^{*}}[\alpha_{i}^{2}]\right]^{2}+\frac{1}{2}n\sigma^{2}(x_{0})\sum_{i=1}^{n}(w_{i}^{*})^{2}\mathbb{E}_{\mathcal{P}^{*}}[1/\alpha_{i}^{2}]\\
	&\leq&\mathbb{E}_{\mathcal{P}^{*}}\left[B\sum_{i=1}^{n}w_{i}^{*}\alpha_{i}^{2}\right]^{2}+\mathbb{E}_{\mathcal{P}^{*}}\left[n\sum_{i=1}^{n}(w_{i}^{*}/2\alpha_{i})^{2}2\sigma^{2}(x_{0})\right],
\end{eqnarray*}
where the first inequality is due to the fact that $\mathbb{E}[1/X^2]\geq 1/\mathbb{E}[X^2]$, and the second inequality is due to the fact that $(\mathbb{E}[X])^{2}\leq\mathbb{E}[X^{2}]$. Therefore,
\begin{eqnarray*}
	\mathcal{R}_{G}^{*}&=&\lim_{n\to+\infty}\underset{\hat{f}(\cdot)\in\Omega_{\hat f}}{\sup} \mathbb{E}_{\mathcal{P}^{*}}\left[B\sum_{i=1}^{n}w_{i}^{*}\alpha_{i}^{2}\right]^{2}+\mathbb{E}_{\mathcal{P}^{*}}\left[n\sum_{i=1}^{n}(w_{i}^{*}/2\alpha_{i})^{2}2\sigma^{2}(x_{0})\right]\\
	&\geq&\lim_{n\to+\infty}\underset{\hat{f}(\cdot)\in\Omega_{\hat f}}{\sup}\left[B\sum_{i=1}^{n}w_{i}^{*}\mathbb{E}_{\mathcal{P}^{*}}[\alpha_{i}^{2}]\right]^{2}+\frac{1}{2}n\sigma^{2}(x_{0})\sum_{i=1}^{n}(w_{i}^{*})^{2}/\mathbb{E}_{\mathcal{P}^{*}}[\alpha_{i}^{2}]\\
	&\geq&\mathcal{R}_{\text{deterministic}}^{*}.
\end{eqnarray*}
Then we conclude the proof for the first equation of Theorem~\ref{thm:1}.

By H\"older's inequality, we have
$$\left[\sum_{i=1}^{n}w_{i}\alpha_{i}^{2}\right]^{1/2}\left[\sum_{i=1}^{n}(w_{i}/\alpha_{i})^{2}\right]^{1/2}\geq\sum_{i=1}^{n}(w_{i}\alpha_{i}^{2})^{1/2}(w_{i}/\alpha_{i})=\sum_{i=1}^{n}w_{i}^{3/2},$$
$$(\sum_{i=1}^{n}1^{3})^{1/3}(\sum_{i=1}^{n}w_{i}^{3/2})^{2/3}\geq\sum_{i=1}^{n}w_{i}=1.$$
Combining the above inequalities yields
$$\sum_{i=1}^{n}w_{i}\alpha_{i}^{2}\geq 1/(n\sum_{i=1}^{n}(w_{i}/\alpha_{i})^{2}).$$
For each case of $\hat{f}(\cdot)\in\Omega_{\hat f}$,
\begin{eqnarray*}
	&&B^{2}(\sum_{i=1}^{n}w_{i}\alpha_{i}^{2})^{2}+\frac{1}{2}n\sigma^{2}(x_{0})\sum_{i=1}^{n}(w_{i}/\alpha_{i})^{2}\\
	&\geq&B^{2}/[n^{2}(\sum_{i=1}^{n}(w_{i}/\alpha_{i})^{2})^{2}]+\frac{1}{2}n\sigma^{2}(x_{0})\sum_{i=1}^{n}(w_{i}/\alpha_{i})^{2}\\
	&=&B^{2}\tilde{\alpha}^{4}+\sigma^{2}(x_{0})/(2\tilde{\alpha}^{2}),
\end{eqnarray*}
where $\tilde{\alpha}=\left[1/(n\sum_{i=1}^{n}(w_{i}/\alpha_{i})^{2})\right]^{1/2}$.
Therefore, for any finite $n$,
\begin{eqnarray*}
    &&\underset{\alpha_{i}>0,\sum_{i=1}^{n}w_{i}=1,w_{i}>0}{\inf}\ \underset{\hat{f}(\cdot)\in\Omega_{\hat f}}{\sup}\left[B\sum_{i=1}^{n}w_{i}\alpha_{i}^{2}\right]^{2}+n\sum_{i=1}^{n}(w_{i}/2\alpha_{i})^{2}2\sigma^{2}(x_{0})\\
    &\geq& \underset{\alpha_{i}>0,\sum_{i=1}^{n}w_{i}=1,w_{i}>0}{\inf}\ \underset{\hat{f}(\cdot)\in\Omega_{\hat f}}{\sup} B^{2}\tilde{\alpha}^{4}+\sigma^{2}(x_{0})/(2\tilde{\alpha}^{2})\\
    &\geq& \underset{\alpha>0}{\inf}\ \underset{\hat{f}(\cdot)\in\Omega_{\hat f}}{\sup}B^{2}\alpha^{4}+\sigma^{2}(x_{0})/(2\alpha^{2}),
\end{eqnarray*}
where the right-hand side of the last inequality is independent of $n$. On the other hand, let $\alpha_*$ be minimax-optimal for~\eqref{thm1.2} and consider the deterministic choice $w_i=1/n$ and $\alpha_i=\alpha_*$ for all $i=1,\ldots,n$. Then
\begin{eqnarray*}
    &&\underset{\alpha_{i}>0,\sum_{i=1}^{n}w_{i}=1,w_{i}>0}{\inf}\ \underset{\hat{f}(\cdot)\in\Omega_{\hat f}}{\sup}\left[B\sum_{i=1}^{n}w_{i}\alpha_{i}^{2}\right]^{2}+n\sum_{i=1}^{n}(w_{i}/2\alpha_{i})^{2}2\sigma^{2}(x_{0})\\
    &\leq& \underset{\hat{f}(\cdot)\in\Omega_{\hat f}}{\sup}B^{2}\alpha_*^{4}+\sigma^{2}(x_{0})/(2\alpha_*^{2})\\
	&=&\underset{\alpha>0}{\inf}\ \underset{\hat{f}(\cdot)\in\Omega_{\hat f}}{\sup}B^{2}\alpha^{4}+\sigma^{2}(x_{0})/(2\alpha^{2}),
\end{eqnarray*}
where the right-hand side is again independent of $n$. Combining the above uniform bounds and letting $n\to\infty$, we obtain
$$\mathcal{R}_{\text{deterministic}}^{*}= \underset{\alpha>0}{\inf}\ \underset{\hat{f}(\cdot)\in\Omega_{\hat f}}{\sup}B^{2}\alpha^{4}+\sigma^{2}(x_{0})/(2\alpha^{2}).$$ Then we conclude the proof for the second equation of Theorem~\ref{thm:1}.

In addition, for each case of $\hat{f}(\cdot)\in\Omega_{\hat f}$, we have
$$B^{2}\alpha_{*}^{4}+\sigma^{2}(x_{0})/(2\alpha_{*}^{2})\geq \underset{\alpha>0}{\inf}\ B^{2}\alpha^{4}+\sigma^{2}(x_{0})/(2\alpha^{2}).$$
Therefore,
\begin{eqnarray*}
	&&\underset{\alpha>0}{\inf}\ \underset{\hat{f}(\cdot)\in\Omega_{\hat f}}{\sup}B^{2}\alpha^{4}+\sigma^{2}(x_{0})/(2\alpha^{2})\\
	&=&\underset{\hat{f}(\cdot)\in\Omega_{\hat f}}{\sup} B^{2}\alpha_{*}^{4}+\sigma^{2}(x_{0})/(2\alpha_{*}^{2})\\
	&\geq&\underset{\hat{f}(\cdot)\in\Omega_{\hat f}}{\sup}\ \underset{\alpha>0}{\inf}\ B^{2}\alpha^{4}+\sigma^{2}(x_{0})/(2\alpha^{2}),
\end{eqnarray*}
which concludes the last inequality of Theorem~\ref{thm:1}.
\endproof

\subsection*{Proof of Proposition~\ref{pro_B}, Corollary~\ref{col_B}, and Proposition~\ref{pro_sigma}.}
\proof
Recall that $m=n_1/2$ and
\[
\widehat{B}
=\frac{1}{m}\sum_{k=1}^{m}
\frac{\delta_{2}y_{1}^{(k)}-\delta_{1}y_{2}^{(k)}}{\delta_{2}\delta_{1}^{3}-\delta_{1}\delta_{2}^{3}}.
\]
Since the summands are i.i.d. across $k$,
\[
\mathbb{E}[\widehat B]
=\mathbb{E}\!\left[\frac{\delta_{2}y_{1}^{(k)}-\delta_{1}y_{2}^{(k)}}{\delta_{2}\delta_{1}^{3}-\delta_{1}\delta_{2}^{3}}\right],
\qquad
\text{Var}(\widehat B)
=\frac{1}{m}\text{Var}\!\left(\frac{\delta_{2}y_{1}^{(k)}-\delta_{1}y_{2}^{(k)}}{\delta_{2}\delta_{1}^{3}-\delta_{1}\delta_{2}^{3}}\right).
\]

\noindent\emph{(i) Bias expansion.}
Under Assumption~\ref{asp:1}, a Taylor expansion with the Lagrange-form remainder yields, for $i\in\{1,2\}$,
\[
\mathbb{E}[y_i^{(k)}]
=f^{(1)}(x_0)\delta_i + B\delta_i^3 + R(\delta_i)\delta_i^5,
\]
where $B=f^{(3)}(x_0)/6$ and $R(\delta)$ is the Lagrange-remainder coefficient induced by $f^{(5)}$ (as defined in Section~2.1).

Therefore,
\begin{align*}
\mathbb{E}[\widehat B]
&=\frac{\delta_2\mathbb{E}[y_1^{(k)}]-\delta_1\mathbb{E}[y_2^{(k)}]}{\delta_2\delta_1^3-\delta_1\delta_2^3}\\
&=\frac{\delta_2\big(f^{(1)}(x_0)\delta_1+B\delta_1^3+R(\delta_1)\delta_1^5\big)
      -\delta_1\big(f^{(1)}(x_0)\delta_2+B\delta_2^3+R(\delta_2)\delta_2^5\big)}
{\delta_2\delta_1^3-\delta_1\delta_2^3}\\
&=B+\frac{\delta_2R(\delta_1)\delta_1^5-\delta_1R(\delta_2)\delta_2^5}{\delta_2\delta_1^3-\delta_1\delta_2^3}.
\end{align*}
Define
\[
D(\delta_1,\delta_2)
\triangleq
\frac{\delta_2R(\delta_1)\delta_1^5-\delta_1R(\delta_2)\delta_2^5}
{(\delta_1^2+\delta_2^2)(\delta_2\delta_1^3-\delta_1\delta_2^3)},
\]
so that
\[
\mathbb{E}[\widehat B]=B+D(\delta_1,\delta_2)(\delta_1^2+\delta_2^2),
\]
which proves~\eqref{eq:EBhat}.

Now suppose $\delta_2=c\delta_1$ for some fixed $c\in(0,1)$. Then the above definition simplifies to
\[
D(\delta_1,c\delta_1)
=\frac{R(\delta_1)-c^4R(c\delta_1)}{1-c^4}.
\]
Since $R(\cdot)$ is bounded for small perturbations (by continuity of $f^{(5)}$ under Assumption~\ref{asp:1}) and $1-c^4>0$ is a fixed constant, $D(\delta_1,c\delta_1)$ is uniformly bounded for sufficiently small $\delta_1$.
Moreover, Assumption~\ref{asp:1} implies $R(\delta)\to f^{(5)}(x_0)/120$ as $\delta\to 0$, so
\[
D(\delta_1,c\delta_1)
=\frac{f^{(5)}(x_0)/120-c^4 f^{(5)}(x_0)/120}{1-c^4}+o(1)
=\frac{f^{(5)}(x_0)}{120}+o(1),
\qquad \text{as }\delta_1\to 0,
\]
which proves~\eqref{eq:D_asymp}.

\noindent\textup{(ii) Variance expansion.} Under Assumption~\ref{asp:2}, $\text{Var}(\hat f(x_0\pm\delta))=\sigma^2(x_0)+\Theta(\delta^s)$ as $\delta\to 0$.
Since the two function evaluations at $x_0\pm\delta_i$ are independent when $k=1$,
\[
\text{Var}(y_i^{(k)})
=\frac{\text{Var}(\hat f(x_0+\delta_i))+\text{Var}(\hat f(x_0-\delta_i))}{4}
=\frac{\sigma^2(x_0)}{2}+\Theta(\delta_i^s),
\qquad i\in\{1,2\}.
\]
Therefore, since $y_1^{(k)}$ and $y_2^{(k)}$ are independent,
\[
\text{Var}(\delta_2y_1^{(k)}-\delta_1y_2^{(k)})
=\frac{\sigma^2(x_0)}{2}(\delta_1^2+\delta_2^2)
+\Theta(\delta_1^{s}\delta_2^{2})
+\Theta(\delta_2^{s}\delta_1^{2}).
\]
Dividing by $(\delta_{2}\delta_{1}^{3}-\delta_{1}\delta_{2}^{3})^{2}$ and then by $m=n_1/2$ yields~\eqref{eq:VarBhat}.
\endproof

\proof
Write the mean-squared error of $\widehat{B}$ in the bias-variance form:
\[
\mathbb{E}\!\left[(\widehat{B}-B)^{2}\right]
=\big(\mathbb{E}[\widehat{B}]-B\big)^{2}+\text{Var}(\widehat{B}).
\]

By Proposition~\ref{pro_B}(i),
\[
\mathbb{E}[\widehat{B}]-B=D(\delta_1,\delta_2)\,(\delta_1^2+\delta_2^2),
\]
and under the pilot design $\delta_2=c\delta_1$ with fixed $c\in(0,1)$,
\[
D(\delta_1,\delta_2)=\frac{f^{(5)}(x_0)}{120}+o(1),\qquad\text{as }n_1\to\infty.
\]
Therefore,
\begin{align*}
\big(\mathbb{E}[\widehat{B}]-B\big)^2
&=D(\delta_1,\delta_2)^2(\delta_1^2+\delta_2^2)^2 \\
&=\left(\frac{f^{(5)}(x_{0})}{120}\right)^{2}(\delta_{1}^{2}+\delta_{2}^{2})^{2}
+o\!\left((\delta_1^2+\delta_2^2)^2\right) \\
&=\left(\frac{f^{(5)}(x_{0})}{120}\right)^{2}(\delta_{1}^{2}+\delta_{2}^{2})^{2}
+o(n_{1}^{4\beta}),
\end{align*}
where the last step uses $\delta_1,\delta_2=\Theta(n_1^\beta)$.

By Proposition~\ref{pro_B}(ii),
\[
\text{Var}(\widehat{B})
=\frac{(\delta_{1}^{2}+\delta_{2}^{2})\sigma^{2}(x_{0})/2+\Theta(\delta_{1}^{s}\delta_{2}^{2})+\Theta(\delta_{2}^{s}\delta_{1}^{2})}
{(n_{1}/2)(\delta_{2}\delta_{1}^{3}-\delta_{1}\delta_{2}^{3})^{2}}.
\]
Using $n_1/2$ in the denominator, the leading term becomes
\[
\frac{(\delta_{1}^{2}+\delta_{2}^{2})\sigma^{2}(x_{0})}
{n_{1}(\delta_{2}\delta_{1}^{3}-\delta_{1}\delta_{2}^{3})^{2}}.
\]
Moreover, since $\delta_1,\delta_2=\Theta(n_1^\beta)$, $\beta<0$ and $s>0$, the remaining terms satisfy
\[
\frac{\Theta(\delta_{1}^{s}\delta_{2}^{2})+\Theta(\delta_{2}^{s}\delta_{1}^{2})}
{n_{1}(\delta_{2}\delta_{1}^{3}-\delta_{1}\delta_{2}^{3})^{2}}
=o\!\left(n_1^{-1-6\beta}\right),
\]
which yields the stated expansion for $\text{Var}(\widehat B)$.

Combining the bias and variance expansions proves the MSE expansion in the corollary. Finally, if $-1/6<\beta<0$, then $n_1^{4\beta}\to0$ and $n_1^{-1-6\beta}\to0$, so $\mathbb{E}[(\widehat{B}-B)^2]\to0$, and hence $\widehat{B}$ is consistent.
\endproof

\proof
Recall that $m=n_1/2$ and
\[
\widehat{\sigma}^{2}(x_{0})
=
\frac{2}{m-1}\sum_{k=1}^{m}\bigl(y_1^{(k)}-\bar y_1\bigr)^2,
\qquad
\bar y_1=\frac{1}{m}\sum_{k=1}^{m}y_1^{(k)}.
\]
Let
\[
S_m^2 \triangleq \frac{1}{m-1}\sum_{k=1}^{m}\bigl(y_1^{(k)}-\bar y_1\bigr)^2,
\qquad
v_{\delta_1}\triangleq \text{Var}(y_1^{(k)}).
\]
For each $n_1$, the pilot outputs $\{y_1^{(k)}\}_{k=1}^{m}$ are i.i.d. across $k$, with distribution depending on $\delta_1$. By the uniformly bounded fourth-moment condition, the sample variance is weakly consistent along the sequence $\delta_1\to0$; that is,
\[
S_m^2-v_{\delta_1}\xrightarrow{p}0,
\qquad \text{as }m\to\infty .
\]
Indeed, this follows from the standard variance bound for the sample variance, since the uniformly bounded fourth moments imply
\[
\text{Var}(S_m^2)=O(1/m),
\]
and hence Chebyshev's inequality gives $S_m^2-v_{\delta_1}\xrightarrow{p}0$.

It remains to identify the limit of $v_{\delta_1}$. Since $y_1^{(k)}$ is formed from two independent function evaluations at $x_0+\delta_1$ and $x_0-\delta_1$, we have
\[
v_{\delta_1}
=
\text{Var}(y_1^{(k)})
=
\frac{\text{Var}(\hat f(x_0+\delta_1))+\text{Var}(\hat f(x_0-\delta_1))}{4}.
\]
By Assumption~\ref{asp:2},
\[
\text{Var}(\hat f(x_0\pm\delta_1))
=
\sigma^2(x_0)+\Theta(\delta_1^s),
\qquad \text{as }\delta_1\to0.
\]
Therefore,
\[
v_{\delta_1}
=
\frac{\sigma^2(x_0)}{2}+\Theta(\delta_1^s).
\]
Since $\delta_1\to0$ as $n_1\to\infty$, it follows that
\[
v_{\delta_1}\to \frac{\sigma^2(x_0)}{2}.
\]

Combining the two parts, we obtain
\[
\widehat{\sigma}^{2}(x_{0})
=
2S_m^2
=
2v_{\delta_1}+o_p(1)
\xrightarrow{p}
2\cdot \frac{\sigma^2(x_0)}{2}
=
\sigma^2(x_0),
\]
which proves the proposition.
\endproof

\subsection*{Proof of Theorem~\ref{thm:pc_oracle}, Theorem~\ref{thm:pc_second_order}, and Corollary~\ref{cor_discussion}.}

\proof
Let
\[
\widehat{\alpha}
\triangleq
\left(\frac{\widehat{\sigma}^{2}(x_0)}{4\widehat B^2}\right)^{1/6},
\qquad
\widehat{\delta}=n_2^{-1/6}\widehat{\alpha}.
\]
Also let
\[
\alpha^*
\triangleq
\left(\frac{\sigma^2(x_0)}{4B^2}\right)^{1/6}.
\]

By Corollary~\ref{col_B}, under \( \delta_2=c\delta_1 \ \text{with}\ c\in(0,1) \) and \( \delta_1=\Theta(n_1^\beta)\ \text{with}\ -1/6<\beta<0, \)
we have
\[
\widehat B\xrightarrow{p}B,\qquad \text{as}\ n\to+\infty.
\]
Moreover, by Proposition~\ref{pro_sigma}, we have
\[
\widehat{\sigma}^{2}(x_0)\xrightarrow{p}\sigma^2(x_0),\qquad \text{as}\ n\to+\infty.
\]
Since \(B\neq0\) under Assumption~\ref{asp:1} and \(\sigma^2(x_0)>0\), the continuous mapping theorem gives
\[
\widehat{\alpha}\xrightarrow{p}\alpha^*,\qquad \text{as}\ n\to+\infty.
\]
In addition, since \(n_1=\Theta(n^\gamma)\) with \(0<\gamma<1\), we have
\(n_1/n\to0\) and \(n_2/n\to1\) as \(n\to+\infty\).

We now analyze the MSE of the calibrated deployment estimator. Conditional on the pilot estimates, \(\widehat{\delta}\) is fixed, and the calibrated deployment stage is a standard CFD estimator with perturbation size \(\widehat{\delta}\). Hence, using the CFD bias-variance expansion from Section~\ref{sec:background}, we have
\[
\mathbb E\!\left[\widehat{\theta}_{\mathrm{PC}}\Big| \widehat B,\widehat{\sigma}^{2}(x_0)\right]
=
f^{(1)}(x_0)+B\widehat{\delta}^{2}+O(\widehat{\delta}^{4}),
\]
and
\[
\text{Var}\left(\widehat{\theta}_{\mathrm{PC}}\Big| \widehat B,\widehat{\sigma}^{2}(x_0)\right)
=
\frac{\sigma^2(x_0)}{2n_2\widehat{\delta}^{2}}
+
O\!\left(\frac{\widehat{\delta}^{s-2}}{n_2}\right).
\]
Therefore, by the conditional bias-variance decomposition,
\begin{eqnarray*}
&&\mathbb E\!\left[
\left(\widehat{\theta}_{\mathrm{PC}}-f^{(1)}(x_0)\right)^2
\right] \\
&=&
\mathbb E\!\left[
\left(
\mathbb E[\widehat{\theta}_{\mathrm{PC}}\mid \widehat B,\widehat{\sigma}^{2}(x_0)]
-
f^{(1)}(x_0)
\right)^2
\right]
+
\mathbb E\!\left[
\text{Var}\!\left(
\widehat{\theta}_{\mathrm{PC}}\Big| \widehat B,\widehat{\sigma}^{2}(x_0)
\right)
\right] \\
&=&
\mathbb E\!\left[
\left(B\widehat{\delta}^{2}+O(\widehat{\delta}^{4})\right)^2
+
\frac{\sigma^2(x_0)}{2n_2\widehat{\delta}^{2}}
+
O\!\left(\frac{\widehat{\delta}^{s-2}}{n_2}\right)
\right].
\end{eqnarray*}

Since \(\widehat{\alpha}\xrightarrow{p}\alpha^*\in(0,\infty)\), we have
\[
\widehat{\delta}^{6}=O_p(n_2^{-1})=o_p(n_2^{-2/3}),
\qquad
\frac{\widehat{\delta}^{s-2}}{n_2}=O_p(n_2^{-2/3-s/6})=o_p(n_2^{-2/3}).
\]
To convert these probability-order remainders into the corresponding expectation-level expansion, we use a standard tail-control argument. Specifically, we require that the contribution from the rare event on which \(\widehat{\alpha}\) is extremely small or large is asymptotically negligible. Under this regularity condition, the above remainders are \(o(n_2^{-2/3})\) in \(L^1\), and hence
\[
\mathbb E\!\left[
\left(\widehat{\theta}_{\mathrm{PC}}-f^{(1)}(x_0)\right)^2
\right]
=
\mathbb E\!\left[
B^2\widehat{\delta}^{4}
+
\frac{\sigma^2(x_0)}{2n_2\widehat{\delta}^{2}}
\right]
+
o(n_2^{-2/3}).
\]

The tail-control condition is mild and can be enforced, if desired, by using an asymptotically inactive truncation of the plug-in perturbation constant \(\widehat{\alpha}\), i.e.,
\[
\widehat{\alpha}_{\mathrm{tr}}
=
\min\{\overline{\alpha},\max\{\underline{\alpha},\widehat{\alpha}\}\},
\qquad
0<\underline{\alpha}<\alpha^*<\overline{\alpha}<\infty.
\]
Since \(\widehat{\alpha}\xrightarrow{p}\alpha^*\), we have \(\widehat{\alpha}_{\mathrm{tr}}\xrightarrow{p}\alpha^*.\)
Moreover, the truncated quantity is bounded away from both zero and infinity, so all terms involving \(\widehat{\alpha}_{\mathrm{tr}}^4\) and \(\widehat{\alpha}_{\mathrm{tr}}^{-2}\) are uniformly bounded. Thus, the required uniform integrability holds automatically for the truncated version. For notational simplicity, and because the truncation is asymptotically inactive, we continue to write \(\widehat{\alpha}\).

Moreover, since \(\widehat{\alpha}\xrightarrow{p}\alpha^*\), the continuous mapping theorem gives
\[
B^2\widehat{\alpha}^{4}
+
\frac{\sigma^2(x_0)}{2\widehat{\alpha}^{2}}\xrightarrow{p}
B^2(\alpha^*)^4
+
\frac{\sigma^2(x_0)}{2(\alpha^*)^2}
=
\mathcal R_{\mathrm{opt}}.
\]
Using the same tail-control argument, or equivalently the asymptotically inactive truncation described above, we obtain
\[
\mathbb E\left[B^2\widehat{\alpha}^{4}
+
\frac{\sigma^2(x_0)}{2\widehat{\alpha}^{2}}\right]\to \mathcal R_{\mathrm{opt}}.
\]
Combining this with \(n_2/n\to1\), we conclude that
\[
\lim_{n\to\infty}
n^{2/3}
\mathbb E\!\left[
\left(\widehat{\theta}_{\mathrm{PC}}-f^{(1)}(x_0)\right)^2
\right]
=
\mathcal R_{\mathrm{opt}}.
\]
Therefore, PC-CFD is oracle-optimal.
\endproof

\proof
Let
\[
\text{MSE}_{\mathrm{PC}}
\triangleq
\mathbb E\!\left[
\left(\widehat{\theta}_{\mathrm{PC}}-f^{(1)}(x_0)\right)^2
\right].
\]
Recall from the proof of Theorem~\ref{thm:pc_oracle} that
\[
\widehat{\alpha}
=
\left(\frac{\widehat{\sigma}^{2}(x_0)}{4\widehat B^2}\right)^{1/6},
\qquad
\widehat{\delta}=n_2^{-1/6}\widehat{\alpha}.
\]
Also define
\[
\alpha^*
=
\left(\frac{\sigma^2(x_0)}{4B^2}\right)^{1/6}.
\]
We use the same tail-control convention as in the proof of Theorem~\ref{thm:pc_oracle}; equivalently, one may use an asymptotically inactive truncation of \(\widehat{\alpha}\). This guarantees that the probability-order bounds below can be interpreted at the corresponding expectation level.

First, by the conditional bias-variance expansion for the calibrated deployment stage,
\begin{equation}\label{eq:pc_mse_second_start}
\text{MSE}_{\mathrm{PC}}
=
\mathbb E\!\left[
B^2\widehat{\delta}^{4}
+
\frac{\sigma^2(x_0)}{2n_2\widehat{\delta}^{2}}
\right]
+
O(n_2^{-1})
+
O(n_2^{-2/3-s/6}).
\end{equation}
Here, the term \(O(n_2^{-1})\) comes from the squared higher-order CFD bias, and the term \(O(n_2^{-2/3-s/6})\) comes from the local variance expansion in Assumption~\ref{asp:2}. Substituting \(\widehat{\delta}=n_2^{-1/6}\widehat{\alpha}\) into~\eqref{eq:pc_mse_second_start} gives
\begin{equation}\label{eq:pc_mse_second_alpha}
\text{MSE}_{\mathrm{PC}}
=
n_2^{-2/3}
\mathbb E\!\left[
B^2\widehat{\alpha}^{4}
+
\frac{\sigma^2(x_0)}{2\widehat{\alpha}^{2}}
\right]
+
O(n_2^{-1})
+
O(n_2^{-2/3-s/6}).
\end{equation}

We next quantify the effect of estimating the perturbation constant. Let
\[
\mathcal R(\alpha)
\triangleq
B^2\alpha^4+\frac{\sigma^2(x_0)}{2\alpha^2}.
\]
Then \(\mathcal R(\alpha)\) is minimized at \(\alpha^*\), and hence \(\mathcal R'(\alpha^*)=0\). Therefore, for \(\widehat{\alpha}\) in a neighborhood of \(\alpha^*\),
\[
\mathcal R(\widehat{\alpha})-\mathcal R(\alpha^*)
=
O\!\left((\widehat{\alpha}-\alpha^*)^2\right).
\]
It remains to bound \(\widehat{\alpha}-\alpha^*\). From Proposition~\ref{pro_B} and Corollary~\ref{col_B}, under \( \delta_2=c\delta_1 \ \text{with}\ c\in(0,1) \) and \( \delta_1=\Theta(n_1^\beta)\ \text{with}\ -1/6<\beta<0, \)
we have
\[
\mathbb E[(\widehat B-B)^2]
=
O(n_1^{4\beta})+O(n_1^{-1-6\beta}).
\]
In addition, by Proposition~\ref{pro_sigma} and Assumption~\ref{asp:2}, the variance estimator satisfies the rate structure
\[
\widehat{\sigma}^2(x_0)-\sigma^2(x_0)
=
O_p(n_1^{-1/2})+O(\delta_1^s).
\]
At the expectation level, under the same moment/tail-control convention,
\[
\mathbb E\!\left[
\left(\widehat{\sigma}^2(x_0)-\sigma^2(x_0)\right)^2
\right]
=
O(n_1^{-1})+O(n_1^{2s\beta}).
\]
Since the mapping
\[
(b,v)\mapsto \left(\frac{v}{4b^2}\right)^{1/6}
\]
is smooth in a neighborhood of \((B,\sigma^2(x_0))\), with \(B\neq0\) and \(\sigma^2(x_0)>0\), a Taylor expansion around \((B,\sigma^2(x_0))\) yields
\[
\mathbb E\!\left[
(\widehat{\alpha}-\alpha^*)^2
\right]
=
O\!\left(\mathbb E[(\widehat B-B)^2]\right)
+
O\!\left(\mathbb E[(\widehat{\sigma}^2(x_0)-\sigma^2(x_0))^2]\right)
=
O(n_1^{2s\beta})
+
O(n_1^{4\beta})
+
O(n_1^{-1-6\beta}).
\]
The \(O(n_1^{-1})\) contribution from estimating \(\sigma^2(x_0)\) is dominated by \(O(n_1^{-1-6\beta})\) because \(-1/6<\beta<0\). Consequently,
\begin{equation}\label{eq:R_alpha_second}
\mathbb E[\mathcal R(\widehat{\alpha})]
=
\mathcal R(\alpha^*)
+
O(n_1^{2s\beta})
+
O(n_1^{4\beta})
+
O(n_1^{-1-6\beta}).
\end{equation}
Since \(\mathcal R(\alpha^*)=\mathcal R_{\mathrm{opt}}\), combining~\eqref{eq:pc_mse_second_alpha} and~\eqref{eq:R_alpha_second} gives
\[
\text{MSE}_{\mathrm{PC}}
=
\mathcal R_{\mathrm{opt}}\,n_2^{-2/3}
\left(
1
+
O(n_1^{2s\beta})
+
O(n_1^{4\beta})
+
O(n_1^{-1-6\beta})
\right)
+
O(n_2^{-1})
+
O(n_2^{-2/3-s/6}).
\]

Now use \(n_1=\Theta(n^\gamma)\), \(0<\gamma<1\), and \(n_2=n-n_1\). Since \(n_1/n\to0\),
\[
n_2^{-2/3}
=
n^{-2/3}\left(1+O\!\left(\frac{n_1}{n}\right)\right)
=
n^{-2/3}\left(1+O(n^{\gamma-1})\right).
\]
Therefore,
\[
\text{MSE}_{\mathrm{PC}}
=
\mathcal R_{\mathrm{opt}}\,n^{-2/3}
\left(
1
+
O(n^{\gamma-1})
+
O(n^{2\gamma s\beta})
+
O(n^{4\gamma\beta})
+
O(n^{-\gamma(6\beta+1)})
+
O(n^{-1/3})
+
O(n^{-s/6})
\right).
\]
Equivalently,
\[
\text{MSE}_{\mathrm{PC}}
-
\mathcal R_{\mathrm{opt}}\,n^{-2/3}
=
O\!\left(
n^{-2/3-\kappa(\beta,\gamma)}
\right),
\]
where
\[
\kappa(\beta,\gamma)
=
\min\left\{
1-\gamma,\,
-2\gamma s\beta,\,
-4\gamma\beta,\,
\frac{1}{3},\,
\frac{s}{6},\,
\gamma(6\beta+1)
\right\}.
\]
This proves the theorem.
\endproof

\proof
For any given $s>0$ and $0<\gamma<1$, solving the following optimization problem yields the optimal choice of $\beta$:
\begin{gather}
	\underset{-1/6<\beta<0}{\min} \max\{\gamma-1,\ 2\gamma s\beta,\ 4\gamma\beta,\ -1/3,\ -s/6,\ -\gamma(6\beta+1)\}.\label{opt_second}
\end{gather}

When $s\geq2$, we have $2\gamma s\beta\leq4\gamma\beta$ and $-s/6\leq-1/3$, and thus \(-\kappa(\beta,\gamma)\) is $\max\{\gamma-1,4\gamma\beta,-1/3,-\gamma(6\beta+1)\}$. Note that $\max\{4\gamma\beta,-\gamma(6\beta+1)\}\geq-2\gamma/5$ and $\max\{\gamma-1,-2\gamma/5\}\geq-2/7>-1/3$. Therefore, when $5/7\leq\gamma<1$, $\gamma-1$ is the solution to \eqref{opt_second} by taking $\beta$ from $\left[\frac{1-2\gamma}{6\gamma},\frac{\gamma-1}{4\gamma}\right]$ such that $\max\{4\gamma\beta,-\gamma(6\beta+1)\}\leq\gamma-1$; when $0<\gamma\leq5/7$, $-2\gamma/5$ is the solution to \eqref{opt_second} by taking $\beta=-1/10$ such that $\max\{4\gamma\beta,-\gamma(6\beta+1)\}=-2\gamma/5$.

When $3/2\leq s\leq2$, we have $2\gamma s\beta\geq4\gamma\beta$ and $-s/6\geq-1/3$, and thus \(-\kappa(\beta,\gamma)\) is $\max\{\gamma-1,2\gamma s\beta,-s/6,-\gamma(6\beta+1)\}$. Note that $\max\{2\gamma s\beta,-\gamma(6\beta+1)\}\geq-2s\gamma/(2s+6)$ and $\max\{\gamma-1,-2s\gamma/(2s+6)\}\geq-2s/(4s+6)\geq-s/6$ when $s\geq3/2$. Therefore, when $\frac{2s+6}{4s+6}\leq\gamma<1$, $\gamma-1$ is the solution to \eqref{opt_second} by taking $\beta$ from $\left[\frac{1-2\gamma}{6\gamma},\frac{\gamma-1}{2\gamma s}\right]$ such that $\max\{2\gamma s\beta,-\gamma(6\beta+1)\}\leq\gamma-1$; when $0<\gamma\leq\frac{2s+6}{4s+6}$, $-2s\gamma/(2s+6)$ is the solution to \eqref{opt_second} by taking $\beta=-1/(2s+6)$ such that $\max\{2\gamma s\beta,-\gamma(6\beta+1)\}=-2s\gamma/(2s+6)$.

When $0<s\leq3/2$, \(-\kappa(\beta,\gamma)\) is $\max\{\gamma-1,2\gamma s\beta,-s/6,-\gamma(6\beta+1)\}$. When $1-s/6\leq \gamma<1$, $\gamma-1$ is the solution to \eqref{opt_second} by taking $\beta$ from $\left[\frac{1-2\gamma}{6\gamma},\frac{\gamma-1}{2\gamma s}\right]$ such that $\max\{2\gamma s\beta,-\gamma(6\beta+1)\}\leq\gamma-1$; when $(s+3)/6\leq \gamma\leq1-s/6$, $-s/6$ is the solution to \eqref{opt_second} by taking $\beta$ from $\left[\frac{s-6\gamma}{36\gamma},\frac{-1}{12\gamma}\right]$ such that $\max\{2\gamma s\beta,-\gamma(6\beta+1)\}\leq-s/6$; when $0<\gamma\leq(s+3)/6$, $-2s\gamma/(2s+6)$ is the solution to \eqref{opt_second} by taking $\beta=-1/(2s+6)$ such that $\max\{2\gamma s\beta,-\gamma(6\beta+1)\}=-2s\gamma/(2s+6)$.
\endproof

\subsection*{Proof of Proposition~\ref{pro_SP}.}
\proof
For \(i=1,\ldots,p\) and \(j=1,\ldots,n\), define 
\[ Z_{i,j}(h) \triangleq \frac{ \hat f(\bm x_0+h\bm\Delta_j) - \hat f(\bm x_0-h\bm\Delta_j) }{ 2h\delta_{i,j} }. \] 
Then \[ \widehat\theta_{\mathrm{SP},i} = \frac{1}{n}\sum_{j=1}^n Z_{i,j}(h). \]
Because the perturbation directions and simulation outputs are independent across replications, the random vectors \(\{(Z_{1,j}(h),\ldots,Z_{p,j}(h))'\}_{j=1}^n\) are independent and identically distributed.

We first characterize the bias. Conditional on \(\bm\Delta_j\), the unbiasedness of the simulation outputs gives 
\[ \mathbb E\!\left[ Z_{i,j}(h)\mid\bm\Delta_j \right] = \frac{ f(\bm x_0+h\bm\Delta_j) - f(\bm x_0-h\bm\Delta_j) }{ 2h\delta_{i,j} }. \] 
Under Assumption~\ref{asp:multi_smooth}, a symmetric Taylor expansion around \(\bm x_0\) yields that as $h\to 0$,
\begin{gather} 
\mathbb E\!\left[ Z_{i,j}(h)\mid\bm\Delta_j \right] = \sum_{k=1}^p (\nabla f(\bm x_0))_k \frac{\delta_{k,j}}{\delta_{i,j}} + \frac{h^2}{6} \sum_{k_1,k_2,k_3=1}^p \nabla_{k_1k_2k_3}^3f(\bm x_0) \frac{ \delta_{k_1,j}\delta_{k_2,j}\delta_{k_3,j} }{ \delta_{i,j} } + O(h^4).\label{taylor-z}
\end{gather} 
Since \(\delta_{i,j}^{-1}=\delta_{i,j}\), mutual independence and symmetry of the Rademacher variables imply \( \mathbb E\!\left[ \delta_{k,j}\delta_{i,j} \right] = \mathbf 1_{\{k=i\}}\). Moreover, 
\begin{gather*} 
\mathbb E\!\left[ \sum_{k_1,k_2,k_3=1}^p \nabla_{k_1k_2k_3}^3f(\bm x_0) \delta_{k_1,j}\delta_{k_2,j}\delta_{k_3,j} \delta_{i,j} \right] = \nabla_{iii}^3f(\bm x_0) + 3\sum_{k\ne i}\nabla_{ikk}^3f(\bm x_0) = G_i. 
\end{gather*} 
Therefore, \[ \mathbb E[\widehat\theta_{\mathrm{SP},i}] - (\nabla f(\bm x_0))_i = \frac{h^2}{6}G_i+O(h^4). \] 
It follows that the squared bias of the full gradient estimator is 
\begin{gather} 
\left\| \mathbb E[\widehat{\bm\theta}_{\mathrm{SP}}] - \nabla f(\bm x_0) \right\|_2^2 = \sum_{i=1}^p \left( \frac{h^2}{6}G_i+O(h^4) \right)^2 = \frac{h^4}{36}\sum_{i=1}^pG_i^2+o(h^4) = B_{\mathrm{sp}}^2h^4+o(h^4). \label{eq:sp_bias_proof} 
\end{gather}

We next characterize the variance. By independence across replications, 
\begin{gather*} 
\mathbb E\!\left[ \left\| \widehat{\bm\theta}_{\mathrm{SP}} - \mathbb E[\widehat{\bm\theta}_{\mathrm{SP}}] \right\|_2^2 \right] = \text{tr} \left( \text{Cov}( \widehat{\bm\theta}_{\mathrm{SP}}) \right) = \frac{1}{n} \sum_{i=1}^p \text{Var}(Z_{i,j}(h)). 
\end{gather*}
By the law of total variance, 
\begin{gather*} \text{Var}(Z_{i,j}(h)) = \text{Var}\!\left( \mathbb E[Z_{i,j}(h)\mid\bm\Delta_j] \right) + \mathbb E\!\left[ \text{Var} \left( Z_{i,j}(h)\mid\bm\Delta_j \right) \right]. 
\end{gather*}
The Taylor expansion~\eqref{taylor-z} and the bounded support of the Rademacher directions imply that as $h\to 0$,
\[ \text{Var}\!\left( \mathbb E[Z_{i,j}(h)\mid\bm\Delta_j] \right) = O(1). \] 
Under independent sampling at the two perturbed points, 
\[ \text{Var} \left( Z_{i,j}(h)\mid\bm\Delta_j \right) = \frac{ \sigma^2(\bm x_0+h\bm\Delta_j) + \sigma^2(\bm x_0-h\bm\Delta_j) }{ 4h^2\delta_{i,j}^2 }. \] 
Since \(\delta_{i,j}^2=1\), Assumption~\ref{asp:multi_var} gives that as $h\to 0$,
\begin{gather*} 
\mathbb E\!\left[ \text{Var} \left( Z_{i,j}(h)\mid\bm\Delta_j \right) \right] = \frac{\sigma^2(\bm x_0)}{2h^2} + O(h^{s-2}). 
\end{gather*} 
Consequently, \[ \text{Var}(Z_{i,j}(h)) = \frac{\sigma^2(\bm x_0)}{2h^2} \left(1+O(h^s)+O(h^2)\right)= \frac{\sigma^2(\bm x_0)}{2h^2} \left(1+o(1)\right). \] 
Summing over \(i=1,\ldots,p\), we obtain 
\begin{gather} 
\mathbb E\!\left[ \left\| \widehat{\bm\theta}_{\mathrm{SP}} - \mathbb E[\widehat{\bm\theta}_{\mathrm{SP}}] \right\|_2^2 \right] = \frac{p\sigma^2(\bm x_0)}{2nh^2} \left(1+o(1)\right). \label{eq:sp_variance_proof} 
\end{gather}

Combining~\eqref{eq:sp_bias_proof} and \eqref{eq:sp_variance_proof} with the bias-variance decomposition gives that as $h\to 0$ and $n\to\infty$,
\[ \mathbb E\!\left[ \left\| \widehat{\bm\theta}_{\mathrm{SP}} - \nabla f(\bm x_0) \right\|_2^2 \right] = \left(B_{\mathrm{sp}}^2+o(1)\right)h^4 + \frac{p\sigma^2(\bm x_0)}{2nh^2} \left(1+o(1)\right). \]
For the non-degenerative case of $B_{\mathrm{sp}}>0$, it remains to optimize the leading terms. Define 
\[ \Psi_n(h) \triangleq B_{\mathrm{sp}}^2h^4 + \frac{p\sigma^2(\bm x_0)}{2nh^2}. \]
Differentiating with respect to \(h\) gives 
\[ \Psi_n'(h) = 4B_{\mathrm{sp}}^2h^3 - \frac{p\sigma^2(\bm x_0)}{nh^3}. \] 
Thus, the unique positive minimizer satisfies 
\[ h = \left( \frac{ p\sigma^2(\bm x_0) }{ 4nB_{\mathrm{sp}}^2 } \right)^{1/6}. \] 
Substituting this perturbation scale into \(\Psi_n(h)\) yields 
\[ \Psi_n(h) = 3\left( \frac{ p\sigma^2(\bm x_0)B_{\mathrm{sp}} }{ 4n } \right)^{2/3}. \]
Therefore, 
\[ \text{MSE}_{\mathrm{SP,opt}} = 3\left( \frac{ p\sigma^2(\bm x_0)B_{\mathrm{sp}} }{ 4n } \right)^{2/3} (1+o(1)). \]
\endproof

\subsection*{Proof of Theorem~\ref{thm:pcsp}.}
\proof
Let \( D_h\triangleq h_2h_1^3-h_1h_2^3\). Since \(h_2=ch_1\) with \(c\in(0,1)\), \( D_h=c(1-c^2)h_1^4=\Theta(h_1^4) \). Recall that \(m=n_1/2\).

We first establish the convergence rates of the pilot estimators $\widehat G_1,\ldots,\widehat G_p$. For \(i=1,\ldots,p\), define 
\[ W_{i,j} \triangleq \frac{ 6\left( h_2y(h_1\bm\Delta_j) - h_1y(h_2\bm\Delta_j) \right)\delta_{i,j} }{ D_h }, \qquad j=1,\ldots,m, \]
so that \[ \widehat G_i=\frac{1}{m}\sum_{j=1}^m W_{i,j}. \] 
By the Taylor expansion~\eqref{taylor_Gi}, 
\begin{eqnarray*} 
\mathbb E[W_{i,j}] &=& \frac{ 6h_2\left[ h_1(\nabla f(\bm x_0))_i + h_1^3G_i/6 + O(h_1^5) \right] }{ D_h } - \frac{ 6h_1\left[ h_2(\nabla f(\bm x_0))_i + h_2^3G_i/6 + O(h_2^5) \right] }{ D_h } \\
&=& G_i+O(h_1^2). 
\end{eqnarray*}
The use of the same Rademacher direction at the two pilot scales makes the first-order directional-derivative term cancel within each paired replication. By Assumption~\ref{asp:multi_var} and the bounded support of the Rademacher directions,
\[
\text{Var}\!\left(
y(h\bm\Delta_j)\mid\bm\Delta_j
\right)
=
O(1)
\]
uniformly over \(\bm\Delta_j\) as \(h\to0\). Conditional on \(\bm\Delta_j\), the simulation outputs generated at the two perturbation scales are mutually independent. Therefore,
\begin{gather*}
\text{Var}\!\left(
W_{i,j}\mid\bm\Delta_j
\right)=
\frac{36}{D_h^2}
\left[
h_2^2
\text{Var}\!\left(
y(h_1\bm\Delta_j)\mid\bm\Delta_j
\right)
+
h_1^2
\text{Var}\!\left(
y(h_2\bm\Delta_j)\mid\bm\Delta_j
\right)
\right]
=
O(h_1^{-6}).
\end{gather*}
Therefore, \(\mathbb E\!\left[
\text{Var}\!\left(
W_{i,j}\mid\bm\Delta_j
\right)
\right]=O(h_1^{-6})\).
Under Assumption~\ref{asp:multi_smooth}, a symmetric Taylor expansion around \(\bm x_0\) yields that as $h_1\to 0$,
\[
\mathbb E\!\left[
W_{i,j}\mid\bm\Delta_j
\right]
=
\sum_{k_1,k_2,k_3=1}^p \nabla_{k_1k_2k_3}^3f(\bm x_0) \delta_{k_1,j}\delta_{k_2,j}\delta_{k_3,j} \delta_{i,j}
+
O(h_1^2).
\]
Since the Rademacher direction has bounded support, \(
\text{Var}\!\left(
\mathbb E\!\left[
W_{i,j}\mid\bm\Delta_j
\right]
\right)
=
O(1)\).
It then follows from the law of total variance that
\begin{gather*}
\text{Var}(W_{i,j})=
\mathbb E\!\left[
\text{Var}\!\left(
W_{i,j}\mid\bm\Delta_j
\right)
\right]
+
\text{Var}\!\left(
\mathbb E\!\left[
W_{i,j}\mid\bm\Delta_j
\right]
\right)
=
O(h_1^{-6})+O(1)
=
O(h_1^{-6}).
\end{gather*}
Since \(W_{i,1},\ldots,W_{i,m}\) are independent and identically
distributed,
\begin{equation} \mathbb E\!\left[ (\widehat G_i-G_i)^2 \right] = O(h_1^4) + O\!\left(\frac{1}{mh_1^6}\right). \label{eq:Gi_rate_pcsp} 
\end{equation}
The reverse triangle inequality gives 
\[ \left| \widehat B_{\mathrm{sp}}-B_{\mathrm{sp}} \right| \leq \frac{1}{6} \left(\sum_{i=1}^p(\widehat G_i-G_i)^2\right)^{1/2}. \] 
Therefore, \eqref{eq:Gi_rate_pcsp} implies 
\begin{equation} 
\mathbb E\!\left[ \left( \widehat B_{\mathrm{sp}}-B_{\mathrm{sp}} \right)^2 \right] = O(h_1^4) + O\!\left(\frac{1}{mh_1^6}\right). \label{eq:Bsp_rate_pcsp} 
\end{equation} 

We next analyze the pilot variance estimator. Conditional on \(\bm\Delta_j\), a symmetric Taylor expansion gives \[ \mathbb E[y(h_1\bm\Delta_j)\mid\bm\Delta_j] = h_1\nabla f(\bm x_0)^\prime\bm\Delta_j + O(h_1^3). \] 
By Assumption~\ref{asp:multi_smooth} and the bounded support of the Rademacher directions, \[ \text{Var}\!\left( \mathbb E[y(h_1\bm\Delta_j)\mid\bm\Delta_j] \right) = O(h_1^2). \]
Moreover, Assumption~\ref{asp:multi_var} yields 
\begin{gather*} 
\mathbb E\!\left[ \text{Var}(y(h_1\bm\Delta_j)\mid\bm\Delta_j) \right] = \mathbb E\!\left[ \frac{ \sigma^2(\bm x_0+h_1\bm\Delta_j) + \sigma^2(\bm x_0-h_1\bm\Delta_j) }{4} \right] = \frac{\sigma^2(\bm x_0)}{2} + O(h_1^s). 
\end{gather*} 
It follows from the law of total variance that \[ 2\text{Var}(y(h_1\bm\Delta_j)) = \sigma^2(\bm x_0) + O(h_1^s) + O(h_1^2). \] 
Since \(\widehat\sigma^2(\bm x_0)\) is twice the sample variance of \(y(h_1\bm\Delta_1),\ldots,y(h_1\bm\Delta_m)\), it follows from the unbiasedness of the sample variance that \( \mathbb E\!\left[ \widehat\sigma^2(\bm x_0) \right] = 2\text{Var}\!\left( y(h_1\bm\Delta_j) \right) \). Hence, \[ \mathbb E\!\left[ \widehat\sigma^2(\bm x_0) \right] - \sigma^2(\bm x_0) = O(h_1^s)+O(h_1^2). \] 
Under the uniform fourth-moment condition for the randomized pilot outputs, \[ \text{Var}\!\left( \widehat\sigma^2(\bm x_0) \right) = O(m^{-1}). \] Therefore, \begin{equation} 
\mathbb E\!\left[ \left( \widehat\sigma^2(\bm x_0)-\sigma^2(\bm x_0) \right)^2 \right] = O(m^{-1}) + O(h_1^{2s}) + O(h_1^4). \label{eq:sigma_rate_pcsp} 
\end{equation} 

Define the oracle and calibrated scale constants by \[ q^* \triangleq \left( \frac{p\sigma^2(\bm x_0)}{4B_{\mathrm{sp}}^2} \right)^{1/6}, \qquad \widehat q \triangleq \left( \frac{ p\widehat\sigma^2(\bm x_0) }{ 4\widehat B_{\mathrm{sp}}^2 } \right)^{1/6}. \]
Then the oracle and calibrated deployment perturbations are \( h^*=q^*n_2^{-1/6}\) and  \(\widehat h=\widehat q\,n_2^{-1/6} \). Let \(\mathcal F_1\) denote the sigma-field generated by all pilot samples. Conditional on \(\mathcal F_1\), the calibrated deployment perturbation \(\widehat h\) is fixed, and the \(n_2\) deployment replications are independent of the pilot samples. Applying the expansion in Proposition~\ref{pro_SP} conditionally on \(\mathcal F_1\) gives 
\begin{gather} 
\mathbb E\!\left[ \left. \left\| \widehat{\bm\theta}_{\mathrm{SPPC}} - \nabla f(\bm x_0) \right\|_2^2 \right| \mathcal F_1 \right] = B_{\mathrm{sp}}^2\widehat h^4 + \frac{p\sigma^2(\bm x_0)}{2n_2\widehat h^2} + O(\widehat h^6) + O(n_2^{-1}) + O\!\left( n_2^{-1}\widehat h^{s-2} \right). \label{eq:conditional_mse_pcsp} 
\end{gather}
The term \(O(\widehat h^6)\) comes from the higher-order squared-bias terms, \(O(n_2^{-1})\) comes from the variance induced by the random directions, and \(O(n_2^{-1}\widehat h^{s-2})\) comes from the local variance remainder in Assumption~\ref{asp:multi_var}. Because \(B_{\mathrm{sp}}>0\) and \(\sigma^2(\bm x_0)>0\), the mapping \[ (a,b)\longmapsto \left(\frac{pa}{4b^2}\right)^{1/6} \] is continuously differentiable in a neighborhood of \((\sigma^2(\bm x_0),B_{\mathrm{sp}})\). Using~\eqref{eq:Bsp_rate_pcsp}, \eqref{eq:sigma_rate_pcsp}, and the same tail-control argument as in the proofs of Theorems~\ref{thm:pc_oracle} and~\ref{thm:pc_second_order}, we obtain 
\begin{equation} 
\mathbb E\!\left[ \left( \widehat q-q^* \right)^2 \right] = O\!\left( h_1^{2s} + h_1^4 + \frac{1}{m} + \frac{1}{mh_1^6} \right)= O\!\left( h_1^{2s} + h_1^4 + \frac{1}{mh_1^6} \right), \label{eq:q_rate_pcsp} 
\end{equation} 
where the last equality follows because the term \(m^{-1}\) is dominated by \((mh_1^6)^{-1}\) as \(h_1\to0\). Let \[ \mathcal{Q}(q)\triangleq B_{\mathrm{sp}}^2 q^4 + \frac{p\sigma^2(\bm x_0)}{2 q^2}.\] Then \(\mathcal{Q}(q)\) is uniquely minimized at \(q^*\), and hence \(\mathcal Q'(q^*)=0\). 
Since \(\mathcal Q\) is twice continuously differentiable and its second derivative is bounded in a neighborhood of \(q^*\), Taylor's theorem gives
\[ \mathcal{Q}(\widehat q)-\mathcal{Q}(q^*)=O\left(\left( \widehat q-q^* \right)^2\right)\]
whenever \(\widehat q\) lies in that neighborhood. Under the same tail-control argument, combining these bounds with~\eqref{eq:conditional_mse_pcsp} and \eqref{eq:q_rate_pcsp} yields
\begin{equation} 
\mathbb E\!\left[ \left\| \widehat{\bm\theta}_{\mathrm{SPPC}} - \nabla f(\bm x_0) \right\|_2^2 \right] = \frac{\mathcal{Q}(q^*)}{n_2^{2/3}} \left[ 1 + O\left(h_1^{2s} + h_1^4 + \frac{1}{mh_1^6}\right) + O(n_2^{-1/3}) + O(n_2^{-s/6}) \right]. \label{eq:pcsp_n2_expansion} 
\end{equation}

Now use \(n_1=\Theta(n^\gamma)\), \(0<\gamma<1\), and \(n_2=n-n_1\). Since \(n_1/n\to0\), \[ n_2^{-2/3} = n^{-2/3} \left( 1+O(n^{\gamma-1}) \right). \]
Since \(m=n_1/2\), \(h_1=\Theta(n_1^\beta)\), and \(-1/6<\beta<0\), 
\begin{eqnarray*} 
&&\mathbb E\!\left[ \left\| \widehat{\bm\theta}_{\mathrm{SPPC}} - \nabla f(\bm x_0) \right\|_2^2 \right]\\
&=& 3\left( \frac{ p\sigma^2(\bm x_0)B_{\mathrm{sp}} }{4 } \right)^{2/3} n^{-2/3}\Big[ 1 + O(n^{\gamma-1}) + O(n^{2\gamma s\beta}+n^{4\gamma\beta}+n^{-\gamma(6\beta+1)}) + O(n^{-1/3}) + O(n^{-s/6}) \Big]. 
\end{eqnarray*} 
Equivalently, \[ \mathbb E\!\left[ \left\| \widehat{\bm\theta}_{\mathrm{SPPC}} - \nabla f(\bm x_0) \right\|_2^2 \right] = 3\left( \frac{ p\sigma^2(\bm x_0)B_{\mathrm{sp}} }{4} \right)^{2/3} n^{-2/3} \left( 1+O\!\left( n^{-\kappa(\beta,\gamma)} \right) \right), \] where \[ \kappa(\beta,\gamma) = \min\left\{ 1-\gamma,\, -2\gamma s\beta,\, -4\gamma\beta,\, \frac{1}{3},\, \frac{s}{6},\, \gamma(6\beta+1) \right\}. \] 
It follows immediately that \[ \lim_{n\to\infty} n^{2/3} \mathbb E\!\left[ \left\| \widehat{\bm\theta}_{\mathrm{SPPC}} - \nabla f(\bm x_0) \right\|_2^2 \right] = 3\left( \frac{ p\sigma^2(\bm x_0)B_{\mathrm{sp}} }{4} \right)^{2/3}. \] This completes the proof.
\endproof

\end{document}